\documentclass[pra,twocolumn,showpacs,amsmath,amssymb,superscriptaddress,floatfix,longbibliography,nofootinbib]{revtex4-2}
\usepackage{amsmath,amsthm,amssymb,amsfonts,enumitem,float,graphicx,physics,hyperref}
\usepackage{color}
\usepackage[dvipsnames]{xcolor}
\usepackage[normalem]{ulem}
\usepackage{subfigure}

\usepackage[normalem]{ulem} 
\usepackage{graphicx}
\usepackage{dcolumn}
\usepackage{bm}
\usepackage{physics}
\usepackage{blochsphere}
\usepackage{mathtools}
\usepackage{multirow}
\usepackage{subfigure}
\usetikzlibrary{angles, intersections, quotes}

\usepackage{hyperref}
\hypersetup{
    colorlinks=true,       
    linkcolor=cyan,          
    citecolor=magenta,        
    filecolor=magenta,      
    urlcolor=cyan,           
    runcolor=cyan
}
\usepackage[capitalise]{cleveref} 

\def\a{\alpha}
\def\b{\beta}
\def\g{\gamma}
\def\d{\delta}

\def\r{\rho}     
\def\s{\sigma}

\def\D{\Delta}

\def\>{\rangle}
\def\<{\langle}
\def\Tr{\mathrm{Tr}}

\newcommand{\ketb}[2]{|{#1}\>\!\<#2|}

\newcommand{\mc}{\mathcal}

\newcommand{\bes} {\begin{subequations}}
\newcommand{\ees} {\end{subequations}}
\newcommand{\bea} {\begin{align}}
\newcommand{\eea} {\end{align}}
\newcommand{\beq}{\begin{equation}}
\newcommand{\eeq}{\end{equation}}

\newtheorem{mytheorem}{Theorem}
\newtheorem{mylemma}{Lemma}
\newtheorem{mycorollary}{Corollary}
\newtheorem{myproposition}{Proposition}

\newtheorem{mydefinition}{Definition}
\Crefname{myproposition}{Proposition}{Propositions}
\Crefname{mylemma}{Lemma}{Lemmas}

\newcounter{example}
\newenvironment{Example}
{
    \par\refstepcounter{example} 
    \textbf{Example \theexample:} 
}
{
    \par\vspace{\baselineskip} 
}

\newcommand{\lind}{\mathcal{L}}

\newcommand{\bc}{\mathbf{c}}
\newcommand{\bh}{\mathbf{h}}
\newcommand{\bl}{\mathbf{l}}
\newcommand{\bv}{\mathbf{v}}
\newcommand{\bw}{\mathbf{w}}

\DeclareMathOperator\arctanh{arctanh}

\begin{document}
    \title{Quantum property preservation}
    \author{Kumar Saurav}
\affiliation{Center for Quantum Information Science \& Technology}
\affiliation{Department of Electrical \& Computer Engineering}
    \author{Daniel A. Lidar}
\affiliation{Center for Quantum Information Science \& Technology}
\affiliation{Department of Electrical \& Computer Engineering}
\affiliation{Department of Physics \& Astronomy}
\affiliation{Department of Chemistry\\
University of Southern California, Los Angeles, CA 90089, USA}

\begin{abstract}
Quantum property preservation (QPP) is the problem of maintaining a target property of a quantum system for as long as possible.
This problem arises naturally in the context of open quantum systems subject to decoherence. Here, we develop a general theory to formalize and analyze QPP. We characterize properties encoded as scalar functions of the system state that can be preserved time-locally via continuous control using smoothly varying, time-dependent control Hamiltonians. 
The theory offers an intuitive geometric interpretation involving the level sets of the target property and the stable and unstable points related to the noise channel.
We present solutions for various noise channels and target properties, which are classified as trivially controllable, uncontrollable, or controllable.
In the controllable scenario, we demonstrate the existence of control Hamiltonian singularities and breakdown times, beyond which property preservation fails.
QPP via Hamiltonian control is complementary to quantum error correction, as it does not require ancilla qubits or rely on measurement and feedback. It is also complementary to dynamical decoupling, since it uses only smooth Hamiltonians without pulsing and works in the regime of Markovian open system dynamics.
From the perspective of control theory, this work addresses the challenge of tracking control for open quantum systems.
\end{abstract}
    \maketitle
    
\section{Introduction}
\label{sec:intro}

Quantum control is the generalization of the well-established classical control theory to the quantum domain, dating back at least to the early 1980s~\cite{TARN1980109,huang1983controllability}.
The field has since expanded greatly into mathematics~\cite{Mikobook}, physics~\cite{Wiseman:book,jacobs2014quantum}, chemistry~\cite{Brumer:book,Brif:2010fu}, and quantum technology~\cite{Koch:2022aa}. 
In the context of quantum information and computation, quantum control is often used to design improved quantum logic gates~\cite{Ramakrishna:96,Palao:02,Grace:2007aa,Motzoi:2009aa,Hsieh:10,Doria:2011aa}. It has long been recognized that quantum error correction can be viewed as a quantum control problem as well~\cite{Ahn:01,Sarovar:04,Sarovar:05,Mabuchi:2009aa}. 

Much of the work in quantum control to date has been devoted to questions of controllability and optimal state-to-state or operator-to-operator transformations. Here, we introduce the complementary problem of \emph{quantum property preservation} (QPP). We define this task more precisely below, but in essence, it amounts to selecting an arbitrary function of the quantum state and preserving its initial value in the presence of a decohering environment~\cite{alicki_quantum_2007,Breuer:book}. QPP is less demanding than full state or operator preservation but is nevertheless of fundamental and practical interest. 

The fundamental interest in QPP arises in the context of understanding the limits of what can be preserved in a restricted and minimalistic scenario where the only resource available is \emph{smooth Hamiltonian control}.
Namely, unlike in quantum error avoidance~\cite{Zanardi:97c,Lidar:1998fk} or correction~\cite{shor_scheme_1995,Steane:96a,Gottesman:1996fk}, we explicitly avoid introducing any encoding or measurement-based feedback, whether classical or coherent~\cite{Lloyd:2000aa}. Moreover, the only control we allow is a smooth, time-dependent Hamiltonian that preserves the property time-locally, 
i.e., we avoid the ``bang-bang'' control setting of dynamical decoupling~\cite{Viola:98,Viola:99a,Khodjasteh:2011aa}, which only recovers a state at the end of each decoupling cycle and only applies under non-Markovian dynamics.
We note that the general theory of quantum error correction and suppression already provides a framework that allows for (stroboscopic) QPP as a special case~\cite{,Knill:2000dq,Blume-Kohout:2008ec,Lidar-Brun:book}. By excluding the resources assumed by the former, preserving an arbitrary quantum state becomes impossible, but, as we show here, it remains possible to preserve specific properties of quantum states. 

The practical interest in QPP without encoding or feedback and using only smooth controls arises, in part, from its very minimalism. It also arises naturally in situations where it is possible to precompute a stabilizing state trajectory so that a given property can be preserved over time as intended. When this is the case, the control Hamiltonian (which becomes a function of the instantaneous system state) can be implemented without requiring continuous state tracking. While we do not discuss specific applications here, this is a relevant scenario for, e.g., metrology and sensing, where QPP can demonstrate an advantage~\cite{tbp-tracking}. Some examples of critical quantum properties to preserve in an open system setting are coherence, fidelity, mean energy, or, more generally, any expectation value.
     
More formally, the main questions we address in this work are the following: (1) Which class of properties $f(\rho)$ can be preserved given a decohering environment, where $f$ is the desired (scalar) function of the instantaneous state $\rho$? (2) Given a property $f$, what control Hamiltonian needs to be applied to preserve $f(\rho)$ over a specified evolution period? (3) How long can this period be? (4) What is the underlying geometry of the control landscape? We consider these questions in both the general setting of finite-dimensional Hilbert spaces and in the special case of a single qubit, where the underlying geometry can be visualized intuitively.

The QPP problem we have defined can be placed within the framework of \emph{tracking control}, which has a rich history in classical control theory
\cite{Hirschorn:79,Hirschorn:88,Jakubczyk:93}.
It was extensively studied in the context of closed quantum systems~\cite{Chen:95,Lu:95,Gross:93,Zhu:99,Zhu:98,Zhu:03}. 

This work is organized as follows. \cref{sec:general_theory} provides the background and notation, and develops the general theory for arbitrary Hilbert spaces. \cref{sec:qubit_theory} specializes these ideas in the qubit setting. Examples for two selected properties (coherence and fidelity) and different noise channels are discussed in \cref{sec:examples}. \cref{sec:conclusion} concludes. Additional technical details are provided in the Appendix.

\section{Theory and Methods}
\label{sec:general_theory}

\subsection{Background}
\label{sec:background}

The most general form of a time-local master equation governing the dynamics of open quantum systems described by a state (density matrix) $\r(t)$ is~\cite{alicki_quantum_2007,Breuer:book}:
\begin{subequations}
\label{eq:matrix_le}
\begin{align}
\label{eq:rhodot}
\dot{\r} &= \mc{L}(\r) = \lind_H(\rho) + \lind_D(\rho)\\
\label{eq:L_H}
\lind_H(\rho) &= -i[H,\rho]\\
\label{eq:L_a}
\lind_D(\rho) &= \sum_{\alpha} \gamma_\alpha(t) (L_\alpha \rho L^\dagger_\alpha - \frac{1}{2} \{L^\dagger_\alpha L_\alpha, \rho\})\ ,
\end{align}
\end{subequations}
where where the dot denotes $\partial/\partial t$, $\mc{L}$ is the Liouvillian, $\lind_D$ the dissipator, and the unitary dynamics are generated by $\lind_H$. We allow both $\lind_D$ and $\lind_H$ to be time-dependent, either explicitly or implicitly. When the rates $\g_\a(t) \ge 0$ $\forall \a$ the dynamics are Markovian; otherwise the master equation describes general non-Markovian, time-local dynamics~\cite{rivas2010entanglement,breuer2016colloquium,Chruscinski:2022aa}. We nevertheless refer to the $L_\alpha$ as Lindblad operators (even in the non-Markovian case) and assume that they are traceless, such that the traceful part is absorbed into the Hamiltonian $H$. 

In the tracking control setting, we view $H$ as the control Hamiltonian $H[\r(t)]$, which is allowed to depend on the instantaneous state $\r(t)$ of the system, i.e., an implicit time dependence. Here, we do not consider the complications associated with measuring the state and the inevitable delay in making $H$ dependent upon the measurement result. In this sense, the results we present in this work are upper bounds on the achievable performance of tracking control. Alternatively, the predicted performance can be realized in cases where the Hamiltonian can be precomputed.

We consider a finite-dimensional Hilbert space $\mc{H}$ with $d=\dim(\mc{H})<\infty$, and the space of bounded linear operators $\mc{B}(\mc{H})$ acting on $\mc{H}$ equipped with the Hilbert-Schmidt inner product $\< A,B\> \equiv \Tr(A^\dagger B)$.  Quantum states are represented by density operators $\rho\in \mc{B}_+(\mc{H})$ (positive trace-class operators acting on $\mc{H}$) with unit trace: $\Tr\rho =1$. A \emph{nice operator basis} for $\mc{H}$ is a set of orthonormal Hermitian operators $\{F_j\}_{j=0}^{J}\in\mc{B}(\mc{H})$ such that $F_0 \propto I$ where $I$ is the identity operator. Consequently, $F_0 = \frac{1}{\sqrt{d}}I$, $\<F_j, F_k\>=\d_{jk}$ and for $j\ge 1$ the operators $F_j$ are traceless.

Expanding the state in the nice operator basis yields $\rho(t) = \frac{1}{d}I+\sum_{j=1}^{J} v_j(t) F_j$, where the \emph{coherence vector} $\bv \in  \mc{M}^{(d)} \subset \mathbb{R}^{J}$ has coordinates given by $v_j = \< F_j, \rho \>$, and $J\equiv d^2-1$. 
$\mc{M}^{(d)}$ is a convex set that is topologically equivalent to a sphere~\cite{bengtsson2006geometry}; see \cref{app:cohvec} for more details. The coherence vector completely characterizes the state, and instead of working in $\mc{B}(\mc{H})$ we can equivalently work in $\mathbb{R}^{J}$ with the standard Euclidian inner product $\mathbf{a}\cdot\mathbf{b} = \mathbf{a}^T \mathbf{b} = \mathbf{b}^T \mathbf{a} = \sum_j a_j b_j$. Rewriting \cref{eq:matrix_le} in matrix form for the coherence vector gives:
 \beq
  \label{eq:le}
     \dot{\bv} = (Q+R) \bv + \bc\ ,
     \eeq
where $Q$ and $R$ are the matrices corresponding to $\lind_H$ and $\lind_D$, respectively, and $\bc$ is non-zero only in the non-unital case [i.e., when $\lind_D(I)\ne 0$], and depends purely on $\lind_D$.
All the vectors and matrices in \cref{eq:le} are real-valued.
Moreover, $Q$ is anti-symmetric, and $R$ is symmetric when $d=2$ (qubit) or when the Lindblad operators are all Hermitian~\cite{ODE2QME}.
In the context of the coherence vector, we refer to $Q$ as the control matrix and refer to $D\equiv (R,\bc)$ as the dissipator, since both $R$ and $\bc$ depend only on $\mc{L}_D$.
Note that in the tracking control setting, $Q = Q(\bv)$, i.e., the control depends on the state.
We use both the density operator and coherence vector formalisms below, emphasizing one over the other as convenient.

\subsection{Realizable trajectories}
\label{sec:realizable-traj}

We refer to differentiable operators from a non-empty interval of the real numbers to $\mc{B}_+(\mc{H})$ as \emph{trajectories} and solutions of \cref{eq:le} with a given Hamiltonian $H$ as \textit{$H$-realizable trajectories}.  We assume throughout that $\|H\|<\infty$, where the norm can be chosen as needed to reflect desirable physical constraints, such as a maximum power constraint.
Denoting the set of $H$-realizable trajectories by $\mathbb{P}_H$, the set $\mathbb{P}$ of \emph{realizable trajectories} is then the union of the sets $\mathbb{P}_H$.

Alternatively, from a control theory perspective, we can start from a given realizable trajectory $\r$, agnostic of the control that realizes it. Then, an equivalent definition is:
\beq
\label{eq:prt-set}
\mathbb{P} = \{\r(\cdot) : \exists H(\cdot) \text{ s.t. } \r,H \text{ satisfy \cref{eq:matrix_le}}\}
\eeq

\begin{proof}
Let $\cup \mathbb{P}_H = \mathbb{P}^1$ and denote the set in \cref{eq:prt-set} by $\mathbb{P}^2$. If $\r \in \mathbb{P}^1$, then $\r \in \mathbb{P}_H$ for some $H$.
By definition, the pair $(\r, H)$ satisfies \cref{eq:matrix_le} and therefore $\r \in \mathbb{P}^2$.
This proves $ \mathbb{P}^1 \subseteq \mathbb{P}^2$.
Now suppose $\s \in \mathbb{P}^2$.
By definition, for some $H$, $(\s, H)$ satisfies \cref{eq:matrix_le}, which means $\s \in \mathbb{P}_H$.
Therefore $\s \in \mathbb{P}^1$.
So $ \mathbb{P}^2 \subseteq \mathbb{P}^1$.
\end{proof}

Note that using the nice operator basis, realizable trajectories, which are density matrix solutions of \cref{eq:matrix_le}, map bijectively to coherence vector solutions of \cref{eq:le}. Based on convenience, we will alternate between these two possible representations of solutions, and refer to both as realizable trajectories.

We will require the following notion of equivalence of trajectories:

\begin{mydefinition}
\label{def:equiv-traj}
Two trajectories, $\r_1:I_1\mapsto \mc{B}_+(\mc{H})$ and $\r_2:I_2\mapsto \mc{B}_+(\mc{H})$ are said to be \emph{equivalent} iff there exists a bijective continuously differentiable map $\varphi:I_1\mapsto I_2$ such that $\forall t\in I_1$:  (1) $\r_1(t)=\r_2(\varphi(t))$ and (2) $\varphi'(t) > 0$.
\end{mydefinition}

From hereon, we use the notation $\s(u)$ to denote a trajectory in $\mc{B}_+(\mc{H})$ and $\bl(u)$ to denote a trajectory in $\mc{M}^{(d)}$, both parametrized by $u\in I_1=[0,1]$. We reserve the notation $\r(t)$ and $\bv(t)$ to denote the density matrix and coherence vector parametrized by time $t\in I_2=[0,t_f]$, where $t_f$ denotes the final time. The map $\varphi:I_1\mapsto I_2$ transforms the dimensionless parameter $u$ to time $t$.

\subsection{General tracking control framework}

The target function can be described as a scalar function of the system state $\rho$, which can be written as $f[\r(t)]$, but we suppress the explicit time-dependence from here on. We assume throughout this work that the function $f$ is real-valued and differentiable. An example is the state purity $f(\rho)=\Tr \rho^2 \equiv P(\rho)$. The gradient of the scalar function $f$ is:
\beq
\grad f (\r) = \sum_{j=1}^{J} \frac{\partial f}{\partial v_j} F_j\ .
\label{eq:gradf}
\eeq
where $\bv$ is the coherence vector for $\r$ and $F_j$ are the basis elements in the nice operator basis. Note that since $f$ is real-valued, the coefficients of $\grad f$ in the nice operator basis are real, and therefore, the $d\times d$ traceless matrix on the r.h.s. of \cref{eq:gradf} is Hermitian. This representation of $\grad f$ is appropriate when it is considered an operator in $\mc{B}(\mc{H})$, in combination with the density operator. Equivalently, $\grad f = (\frac{\partial f}{\partial v_1},\dots,\frac{\partial f}{\partial v_{J}})$ is a vector in $\mathbf{R}^{J}$, which is appropriate when it is considered in combination with the coherence vector. Both representations are used below, and the appropriate representation will be clear by context.

The goal of time-local QPP is to keep the target property constant for as long as possible, i.e., to find a control Hamiltonian such that $f[\r(t)] = f[\r(0)]$ for all $t\in [0,t_f]$, while maximizing $t_f$. This is equivalent to $\frac{d}{dt}f[\r(t)] = 0$ for all $t\in [0,t_f]$, i.e.:
\begin{equation}
\label{eq:constraint}
    \< \grad f, \dot{\r} \> = 0 = \grad f \cdot \dot{\bv} \ .
\end{equation}
Here, on the l.h.s. $\grad f \in \mc{B}(\mc{H})$ while on the r.h.s. $\grad f \in \mathbf{R}^{J}$.

Any function $f$ independent of $\rho$ is trivially preservable since $\grad f = 0$ in \cref{eq:constraint}. For the rest of this work, we only consider non-trivial functions $f$ with a non-zero gradient over $\mc{B}_+(\mc{H})$.

\cref{eq:constraint} formalizes that to preserve $f$, the path of the coherence vector should be restricted to the \emph{level set} of $f$, i.e., the hypersurface on which $f$ is constant (recall that if $f$ is differentiable, then $\grad f$ at a point is either orthogonal to the level set of $f$ at that point or zero). 

Therefore, the equations that govern the tracking control problem are the \emph{dynamics equation} \cref{eq:matrix_le} for the density matrix or \cref{eq:le} for the coherence vector, along with the respective forms of the \emph{constraint equation}
\bes
\label{eq:constraint_matrix_dynamics_eq-all}
\begin{align}
\label{eq:constraint_matrix_dynamics_eq}
            &\< \grad f, -i[H, \rho] +\lind_D \rho \> = 0\\
\label{eq:constraint_matrix_dynamics_eq-CV}
            & \grad f \cdot [(Q+R) \bv + \bc] = 0\ .
\end{align}    
\ees
Note that in the tracking control problem setting, the control Hamiltonian $H$ and $Q$ are explicit functions of $\r(t)$ and $\bv(t)$, respectively, instead of $t$. 

\begin{mydefinition}
\label{def:f-pres}
For a given target property $f$, a control Hamiltonian $H$ 
for which the constraint equation is satisfied, is called an \emph{$f$-preserving Hamiltonian}.
\end{mydefinition}

\begin{mydefinition}
\label{def:f-pres-traj}
A solution of the dynamics and constraint equations for $t\in [0,t_f]$ is called an \emph{$f$-preserving trajectory} $\r(t)$ [or $\bv(t)$].
\end{mydefinition}

\subsection{Uncontrollable and trivially controllable target properties}
\label{sec:triv}

Two exceptional cases of target properties arise from the constraint equation, which we call trivially controllable and uncontrollable.

\subsubsection{Trivially controllable $f$}

\begin{mydefinition}
A target property $f$ is \emph{trivially controllable} when for all $\rho$
\beq
\< \grad f, \lind_D \rho \> = 0 = \grad f\cdot(R\bv + \bc)\ .
\eeq
\end{mydefinition}
This holds when the dissipator itself moves the state along the level set of $f$. In this case, $f$ is already invariant under the action of the dissipator, and an additional control Hamiltonian is not needed to preserve it, i.e., we can set $H = Q = 0$. 

As a simple example that we will use throughout to illustrate the theory, consider:
\begin{Example}
\label{ex:1}
Markovian dephasing of a single qubit: $L_\a=\s^z$ and $\g_\a(t) = \g$ in \cref{eq:L_a}. Equivalently, $R=\mathrm{diag}(-2\g,-2\g,0)$ and $\bc=\mathbf{0}$.
\end{Example} 
In this case, any target property $f$ for which $\grad f\cdot R\bv = \frac{\partial f}{\partial{v_x}}v_x+\frac{\partial f}{\partial{v_y}}v_y=0$ is trivially controllable. One solution of this partial differential equation is $f(\bv)=f(v_z)$. Another solution is any target property that depends only on the ratio $v_y/v_x$, i.e., $f(\bv)=g(v_y/v_x)$, where $g$ is an arbitrary function.

\subsubsection{Uncontrollable $f$}
\label{sec:unc-f}

\begin{mydefinition}
A target property $f$ is \emph{uncontrollable} when it is not trivially controllable and
\begin{equation}
\label{eq:inffective_h_condition_matrix}
\< \grad f, -i[H, \rho]\>=\< H, -i[\rho, \grad f]\>=0\ ,
\end{equation}
whenever $\< \grad f, \lind_D \rho \> \neq 0$.
\end{mydefinition}

We impose that $f$ is not trivially controllable since otherwise, we cannot guarantee the existence of a $\rho$ such that $\< \grad f, \lind_D \rho \> \neq 0$.

For a given state $\rho$, the uncontrollability condition~\eqref{eq:inffective_h_condition_matrix} holds for all control Hamiltonians iff $i[\rho, \grad f]=0$. At such points, any change in the target property value $f$ depends solely on the dissipator, and the control Hamiltonian has no effect. 

In terms of the coherence vector, uncontrollable target properties are equivalently characterized by
\beq
\label{eq:ineffective_h_condition_vector}
\grad f \cdot Q\bv = (\grad f)^T Q\bv = -\bv^T Q \grad f= 0\ ,
\eeq
whenever $\grad f\cdot(R\bv + \bc) =  (\grad f)^T (R\bv + \bc) \ne 0$, using the antisymmetry of $Q$ and the symmetry of the dot product over real vector spaces. \cref{eq:ineffective_h_condition_vector} holds for all $Q$ iff $\grad f \propto \bv$.

\begin{myproposition} 
\label{prop:1}
Any target property that depends only on the magnitude $v\equiv \| \bv\|$ of the coherence vector, and for which $\bv^T (R\bv + \bc) \ne 0$, is uncontrollable.
\end{myproposition} 

\begin{proof}
For this class of properties $f(\bv)$ can be written as a scalar function $f(v)$, so that:
\beq
        \grad f(\bv) = \grad f(v) = \frac{d f(v)}{dv}\frac{\bv}{v} \propto \bv\ .
\eeq
Then $(\grad f)^T Q\bv = -\bv^T Q \grad f$ implies $\bv^T Q \bv=0$ and \cref{eq:ineffective_h_condition_vector} holds. The property
$(\grad f)^T (R\bv + \bc) \ne 0$ holds by assumption.
\end{proof}

\begin{Example}
\label{ex:2}
An example of an uncontrollable target property is the purity: $P(\rho) = \frac{1}{d} + v^2$. In the case of Markovian dephasing (\cref{ex:1}), $R\bv' + \bc \neq \bf{0}$ for states $\bv' \neq (0,0,v_z) $. For any such $\bv'$, \cref{eq:ineffective_h_condition_vector} holds.
\end{Example}

\subsubsection{Controllable target properties}
Having defined and characterized trivially controllable and uncontrollable target properties, we now define:

\begin{mydefinition}
A \emph{controllable target property} is a target property that is neither trivially controllable nor uncontrollable.
\end{mydefinition}

An example of a controllable target property is coherence preservation subject to dephasing, as discussed below, in \cref{sec:coherence}.

\subsection{$f$-preserving Hamiltonians for controllable target properties}

Next, we identify necessary conditions $f$-preserving Hamiltonians must satisfy when a target property is controllable.

\begin{mytheorem}
\label{th:main}
An $f$-preserving control Hamiltonian satisfies
\beq
\label{eq:matrix_hamiltonian_constraint_explicit}
\< H, i [\rho, \grad f] \> = \< \grad f, \lind_D \rho \> \ .
\eeq
Assuming $[\rho, \grad f]\ne 0$, a particular (but not unique) solution of \cref{eq:matrix_hamiltonian_constraint_explicit} is:
\beq
H = i\frac{\< \grad f, \lind_D \rho \>}{\norm{[\rho, \grad f]}^2} [\rho, \grad f]\ .
\label{eq:basic_control-H}
\eeq
\end{mytheorem}

\begin{proof}
\cref{eq:matrix_hamiltonian_constraint_explicit} follows directly from rearranging \cref{eq:constraint_matrix_dynamics_eq} using the identity $\< \grad f, i[H, \rho] \> = \< H, i[\rho, \grad f] \>$.  \cref{eq:basic_control-H} satisfies \cref{eq:matrix_hamiltonian_constraint_explicit} by inspection. 
\end{proof}

Note that, based on \cref{eq:basic_control-H}, the overall scale of $f$ cancels out, and the $f$-preserving component of the control Hamiltonian is linear in the dissipator.

Using \cref{def:f-pres-traj} of an $f$-preserving trajectory, we can now state that such trajectories arise in particular when the system evolution is obtained by applying a control Hamiltonian which obeys \cref{eq:matrix_hamiltonian_constraint_explicit}. 

It is clear that trivially controllable target properties correspond to \cref{eq:matrix_hamiltonian_constraint_explicit} being satisfied with $H=0$ for all states $\rho$. Likewise, uncontrollable target properties correspond to the l.h.s. of \cref{eq:matrix_hamiltonian_constraint_explicit} vanishing with the r.h.s. remaining non-zero. In \cref{eq:basic_control-H}, uncontrollable target properties correspond to $H$ diverging. We call a state and a time at which this happens
a \emph{breakdown point} and \emph{breakdown time}, respectively, and return to these concepts in \cref{sec:breakdown_pts_f_traj}.

\subsubsection{Geometric interpretation}
\cref{eq:matrix_hamiltonian_constraint_explicit} has a geometric interpretation: the dissipator can cause the state to move away from the level set, and since $\grad f$ is orthogonal to the level set, this is quantified by the magnitude of $\< \grad f, \lind_D \rho \>$ or $\grad f \cdot (R\bv +\bc)$. The control Hamiltonian can also push $\r$ away from the level set; however, its action is \emph{orthogonal} to $\r$ (since they appear together in a commutator, akin to the vector product; see \cref{sec:qubit_theory-a}). 
Therefore, to counteract the dissipator, the relevant Hamiltonian component is the one that is orthogonal to both $\rho$ and $\grad f$ (i.e., $\< H, i[\rho, \grad f] \>$), and the required magnitude is the same as that given by the dissipator.

\subsubsection{Control parameters}

The $d^2$ matrix elements of $H$ in a nice operator basis are the \emph{control parameters}. We call a control parameter \emph{relevant} when it is needed for satisfying the constraint  \cref{eq:matrix_hamiltonian_constraint_explicit}. Next, we address how many control parameters are relevant.

\begin{myproposition}
\label{prop:control-params}
    Let $m_0$ denote the multiplicity of the $0$ eigenvalue of $i[\rho, \grad f]$; then $H$ has exactly $d-m_0$ relevant control parameters with which to satisfy the constraint \cref{eq:matrix_hamiltonian_constraint_explicit}. In particular, the number of relevant control parameters is at most $d$.
\end{myproposition}

\begin{proof}
    The operator $i[\rho, \grad f]$ is Hermitian and has a spectral decomposition. After diagonalization it can be written as $\sum_{j=0}^{d-1}\beta_j \ketb{\beta_j}{\beta_j}$. In the same basis, $H$ can be expanded as $\sum_{j,k=0}^{d-1} h_{jk} \ketb{\beta_j}{\beta_k}$. Thus,  \cref{eq:matrix_hamiltonian_constraint_explicit} becomes:
    \beq
    \sum_{j=0}^{d-1} h_{jj}\beta_j = \< \grad f, \lind_D \rho \> \ .
    \eeq
    This is a linear equation involving scalars on both sides. Since $m_0$ is the number of zero $\b_j$'s, the sum on the l.h.s. is over the $d-m_0$ terms for which $\b_j\ne 0$. The corresponding $h_{jj}$ are the relevant control parameters for satisfying the constraint \cref{eq:matrix_hamiltonian_constraint_explicit}.
\end{proof}

Note that the remaining $d^2-(d-m_0)$ control parameters in $H$ are irrelevant to time-local property preservation, i.e., they play no role in satisfying \cref{eq:matrix_hamiltonian_constraint_explicit}. However, in practice, it may be advantageous to utilize these control parameters as well (instead of setting them to zero) if this simplifies the implementation of the control Hamiltonian. Moreover, modifying these parameters can also change the breakdown time (e.g., see \cref{subsec:bit-flip-discussion}).

\subsection{Stable points}
 \label{sec:hd_stable_points}
The goal of preserving the target property for as long as possible motivates the following definition: 

\begin{mydefinition}
\label{def:stable-points}
A \emph{stable point} of a given dissipator is a system state $\r$ for which the control Hamiltonian can be chosen so that $\r$, and hence any target property, is indefinitely preserved, i.e., $\dot{\r}=0$ $\forall t$. 
\end{mydefinition}

Note that this definition is closely related to the concept of decoherence-free subspaces in the Markovian setting~\cite{Zanardi:98a, ShabaniLidar:05a}, but unlike the latter, we do not assume the availability of additional qubits for encoding.

The converse of a stable point is:
\begin{mydefinition}
\label{def:unstable-points}
A \emph{unstable point} is a system state $\r$ for which $\dot{\r}\ne 0$ for all control Hamiltonians $H$.
\end{mydefinition}

Under an appropriate control, the system state does not change when it is at a stable point, and therefore, neither does the purity $P(\r)=\Tr\r^2$. However, a constant purity is not a sufficient condition for stable points:
\begin{myproposition}
\label{prop:Pdot=0}
For $d>2$, $\dot{P}=0$ is a necessary, but not sufficient, condition for stable points.
\end{myproposition}
\begin{proof}
 The rate of change of the purity  is
\beq
\label{eq:dotP}
\dot{P}(\r)=2\Tr(\r\dot{\r}) = 2\Tr[\r \mc{L}_D(\r)]\ ,
\eeq
or, equivalently, $P(\bv) = \frac{1}{d}+\|\bv\|^2$ and
\beq
\label{eq:dotP-v}
\dot{P}= 2\bv\cdot\dot{\bv} = 2\bv \cdot (R\bv+\bc)\ .
\eeq
Note that $\dot{P}$ is independent of the control Hamiltonian, since $\Tr(\r[H,\r])=0$.
Consider a coherence vector $\bv \in \mathbb{R}^{d^2-1}$, with a corresponding $d\times d$-dimensional density matrix $\r$. For a stable point, $\dot \bv = 0$; therefore $\dot{P}= 2\bv \cdot \dot{\bv} = 0$, which proves necessity. For insufficiency, note first that $\bv$ has a $d^2-2$ dimensional orthogonal space. Following the same argument as in \cref{prop:control-params}, we can diagonalize the density matrix as $\r=\sum_{j=0}^{d-1}\b_j \ketb{\b_j}{\b_j}$ and write $H$ in the same basis as $H=\sum_{k,l=0}^{d-1}h_{kl} \ketb{\b_k}{\b_l}$, so that $-i[H,\rho]=-i\sum_{k\ne l} h_{kl}(\b_l-\b_k)\ketb{\b_k}{\b_l}$. Thus, $-i[H,\r]$ belongs to an orthogonal subspace of dimension at most $d^2-d < d^2-2$ when $d>2$. Therefore there exist states $\rho$ and dissipators $\lind_D$ for which no $H$ can found such that $-i[H,\rho] + \lind_D\rho = 0$.
\end{proof}
We show later (\cref{prop:locus}) that in the qubit case ($d=2$), $\dot{P}=0$ is also a sufficient condition for stable points.

Note that stable points are inherent to the dissipator, not the target property. The definition of stable points along with \cref{eq:rhodot} implies that there exists a Hamiltonian $H$ such that $\mc{L}(\r) = \lind_H(\rho) + \lind_D(\rho)=0$.
The following proposition builds on this and further characterizes stable points:

\begin{myproposition}
    \label{thm:stable_condition}
    A system state is stable if it satisfies $\Pi_\lambda \lind_D (\rho) \Pi_\lambda=0$ for every eigenspace projector $\Pi_\lambda$ of $\rho$.
\end{myproposition}
        
The proof needs the following ingredients and lemma:
Given a Hermitian matrix $B$ with eigenspace projectors $\{\Pi_l\}$, define the set $\mathcal{E}_B\equiv\{A=-i[H,B]: H=H^\dag\}$ and $\mathcal{E}_B'\equiv\{A: \Pi_l A \Pi_l=0 \;\forall\; \Pi_l$\}.

 \begin{mylemma}
 \label{lm:commutator_structure}
$A\in \mathcal{E}_B \iff A \in \mathcal{E}_B'$.
\end{mylemma}

This lemma can be viewed as characterizing the action of the Hamiltonian ($H$) as inducing a zero instantaneous change in the eigenspace of the density matrix ($B$).

\begin{proof}
Assume that $A \in \mathcal{E}_B$. Then $A=\lind_H(B)=-i[H,B]$ for some Hamiltonian $H$. The spectral decomposition of $B$ is $\sum_{j=1}^d\lambda_j \Pi_j$, and $H$ can be written in the same basis as $\sum_{j,k=1}^d h_{jk}\ketb{\lambda_j}{\lambda_k}$, where $\Pi_j \ket{\lambda_k} = \d_{\lambda_j \lambda_k}\ket{\lambda_k}$. Expanding the commutator yields:
\beq
\label{eq:commutator}
            A=-i[H,B] = -i\sum_{j,k=1}^d h_{jk}(\lambda_k-\lambda_j)\ketb{\lambda_j}{\lambda_k}\ ,
\eeq
so that for every eigenspace projector $\Pi_l$ of $B$:
\beq
\Pi_l A \Pi_l = -i \sum_{j,k=1}^d h_{jk}(\lambda_k-\lambda_j) \Pi_l \d_{\lambda_l \lambda_j}\d_{\lambda_k \lambda_l} = 0 \ ,
\eeq
so $A \in \mathcal{E}_B'$.
    
Now assume that $A \in \mathcal{E}_B'$. Then $\Pi_l A \Pi_l=0$ for all eigenspace projectors $\Pi_l$ of $B$. Consequently, $A$ can be expanded in the eigenbasis of $B$ as $\sum_{k,l} a_{kl}\ketb{\lambda_k}{\lambda_l}$ such that $a_{kl}=0$ whenever $\lambda_k=\lambda_l$. The Hamiltonian $H$ given by $\sum_{j,k} h_{jk}\ketb{\lambda_j}{\lambda_k}$ where
\beq
\label{eq:h_jk}
h_{jk} =
\begin{cases*}
          -i\frac{a_{jk}}{\lambda_j-\lambda_k} & if $\lambda_j\neq\lambda_k$ \\
          0        & otherwise
\end{cases*}\ ,
\eeq
satisfies the relation $A=-i[H,B]$, which can be verified using \cref{eq:commutator}. Therefore $A \in \mathcal{E}_B$.
\end{proof}

We are now ready to prove \cref{thm:stable_condition}.

\begin{proof}
    If $\Pi_\lambda \lind_D (\rho) \Pi_\lambda=0$ for every eigenspace projector $\Pi_\lambda$ of $\rho$, then by \cref{lm:commutator_structure}, $\lind_D (\rho)$ can be written as $-i[H',\rho]$ for some Hermitian matrix $H'$. We can set the Hamiltonian to be $-H'$, leading to $\dot{\r} = 0$ in \cref{eq:matrix_le}, thus preserving that state and by extension, any target property.
\end{proof}

Assuming that stable points exist for a given dissipator, the ultimate goal of the quantum tracking control protocol could be seen as \emph{guiding the system from a given initial state to a stable point}, if possible. The reason is self-evident: reaching a stable point guarantees that the target property will be preserved. However, the stable point might not be reachable from a given initial state by applying a finite control Hamiltonian. We discuss this in detail next.

\subsection{Breakdown points, $f$-preserving trajectories}
\label{sec:breakdown_pts_f_traj}

At the opposite extreme of stable points are singularities of the control Hamiltonian associated with uncontrollable target properties. We call such singularities \emph{breakdown points}; they are points at which an $f$-preserving control Hamiltonian [\cref{eq:basic_control-H}] diverges. More formally:

\begin{mydefinition}
\label{def:breakdown-gen}
For a given target property $f$ and dissipator $\lind_D$, a state is called a \emph{breakdown point} $\rho_b$ 
if (a) $[\rho_b, \grad f] = 0$ and (b) $\< \grad f, \lind_D \rho_b \>\neq 0$.
The \emph{breakdown time} $t_b$ is the evolution time to a breakdown point. 
\end{mydefinition}

The first example of a breakdown point and time was identified in Ref.~\cite{LidarSchneider:04} in the context of a dephasing Lindbladian (\cref{ex:1}) with coherence as the target property. More generally, singularities are a well-known
feature of tracking control \cite{Chen:95,Lu:95,Gross:93,Zhu:99}.

\begin{myproposition}
\label{th:no-breakdown-on-f}
    An $f$-preserving trajectory does not pass through breakdown points.
\end{myproposition}
\begin{proof}
    Suppose $\rho$ is a breakdown point on the $f$-preserving trajectory. By definition, $i[\rho, \grad f] = 0$. But from \cref{eq:matrix_hamiltonian_constraint_explicit}, this means $\< \grad f, \lind_D \rho_b \> = 0$, which contradicts our assumption.
\end{proof}

Next, we argue that the set of breakdown points introduces holes or partitions in the $f$ level set hypersurface. This hypersurface is a geometric object in $\mc{M}^{(d)}$, the convex set of coherence vectors mentioned in \cref{sec:background}.
The control Hamiltonian can preserve the target property within each such partition but not between different partitions. 
To formalize this, let the level set of $f$ be represented by the surface $S_f\in\mc{M}^{(d)}$, let $\mathbb{B}_{f, D}$ be the set of breakdown points given $f$ and $\lind_D$ (or $D$), and let $S'_f \equiv S_f\setminus\mathbb{B}_{f, D}$ (the level set minus the set of breakdown points). 

\begin{mydefinition}
Two states $\rho_a$ and $\rho_{b}$ with respective coherence vector $\bv_a\in S_f$ and $\bv_b\in S_f$ are \emph{disconnected} if for all differentiable trajectories $\bl: [0,1] \mapsto \mc{M}^{(d)}$ with beginning and end points $\bl(0)=\bv_{a}$ and $\bl(1)=\bv_{b}$, there exists $u\in(0,1)$ s.t. $\bl(u) \notin S'_f$.
\end{mydefinition}
In other words, every trajectory must pass through one or more breakdown points to connect the disconnected states $\rho_a$ and $\rho_{b}$.
The following clarifies what this means in terms of target-property-preserving trajectories:

\begin{myproposition}
\label{th:disconnect_breakdown}
A system starting from a point in $S'_f$ cannot reach a disconnected point in $S'_f$ via an $f$-preserving trajectory.
\end{myproposition}
    
\begin{proof} 
Suppose an $f$-preserving trajectory does exist between two disconnected states $\rho_{a}$ and $\rho_{b}$. Consider a state $\rho'$ and corresponding coherence vector $\bv'$ on the portion of the trajectory lying outside $S'_f$. At this point, either $f(\rho') \neq f(\rho_{a})$ (not on the level set) or $\bv'\in \mathbb{B}_{f, D}$ (a breakdown point, i.e., the control Hamiltonian diverges). 
\end{proof}

An $f$-preserving trajectory can reach a stable point from an unstable point, but not always: there are cases where a stable point can only be reached from another stable point (see \cref{sec:dephasing-bd}).

\subsection{Control-independent characterization of $f$-preserving trajectories}

Suppose one wishes to steer a system from a given initial point $\r$ (e.g., an unstable point) to a final point $\s$ (e.g., a stable point) while preserving the target $f$.
From a Hamiltonian-first perspective, we choose a control Hamiltonian and evolve the state.
However, this approach does not readily yield a Hamiltonian control and state trajectory compatible with the final state and the $f$-preservation constraint.
Here, we address an alternative way of approaching the problem: we choose a trajectory instead of a Hamiltonian and determine whether such a given trajectory between two such given points is realizable, as defined in \cref{sec:realizable-traj}.

\subsubsection{A control-independent characterization of realizable trajectories}
\label{sec:characteristics_of_f_trajectory}

The system's evolution is captured by two different processes in \cref{eq:rhodot}: $\lind_H$, which generates unitary dynamics, and $\lind_D$, which generates non-unitary dynamics.
This implies that the dissipation-adjusted system evolution is unitary: $\dot{\r}(t) - \lind_D\rho(t) = \lind_H\rho(t)$. 
It turns out that the set of all trajectories generated by varying Hamiltonians can be characterized by this insight via \cref{lm:commutator_structure}.
We formalize this next.

Consider a parametrized trajectory $\bl : [0,1] \to \mc{M}^{(d)}$.
This corresponds to an operator-valued trajectory $\s : [0,1] \to \mc{B}_+(\mc{H})$ using the nice operator basis.
Let $\{\lambda_i\}$ and $\{\Pi_i\}$ represent the eigenvalues and corresponding eigenprojectors of a Hermitian operator.
Recall that a realizable trajectory is a solution of \cref{eq:matrix_le} for some Hamiltonian $H$, or, equivalently, \cref{eq:le} for some $Q$.

\begin{mylemma}
 \label{th:hd_general_trajectory}
 Let $\s(u) = \sum_i \lambda_i(u)\Pi_i(u)$ be the spectral decomposition of $\s(u)$. Then, the corresponding trajectory $\bl(u)$ is equivalent to a realizable trajectory
 iff $\ \forall \Pi_i(u)$, $u \in [0,1]$:
    \begin{equation}
        \label{eq:18}
            \Pi_i(u) [\partial_u \s(u) - c(u) \lind_D \s(u)]\Pi_i(u) =  0\ ,
        \end{equation}
 for some $c(u)$ where $0 < c(u) < \infty$.
\end{mylemma}

\begin{proof}
$\implies$: 
    Since, by assumption, \cref{eq:18} holds for all eigenspace projectors $\Pi_i(u)$, by \cref{lm:commutator_structure} there exists a Hermitian operator $H(u)$ such that:
    \begin{align}
            \label{duCu}
            &\partial_u\s(u) - c(u) \lind_D \s(u) = -i[H(u), \s(u)] \ .
    \end{align}
At this point, it is already clear that $\s(u)$ represents a solution of \cref{eq:matrix_le} for some Hamiltonian, and therefore corresponds to a realizable trajectory. However, some extra work is required to account for the parametrization.

    Reparameterizing with $dt = c(u)du$ yields a transform from the interval $I_1=[0,1]$ to a new interval $I_2 = [0,\varphi(1)]$ where $\varphi(s) = \int_{0}^{s} c(u) du$. Since $c(u)$ is positive, the inverse transform $M^{-1}:I_2\to I_1$ is also well defined. Defining $\rho(t)\equiv \s(\varphi^{-1}(t))$ and $H_0(t)\equiv \frac{1}{c(\varphi^{-1}(t))}H(\varphi^{-1}(t))$ for $t\in I_2$, \cref{duCu} becomes
    \beq
      \dot{\r}(t) = -i[H_0(t),\rho(t)] + \lind_D \rho(t)\ .
    \eeq
    Since $c(u)$ is finite and non-zero $\forall u \in I_1$, it follows that $H_0(t)$ is finite $\forall t \in I_2$.
    Moreover, if the nice operator basis representation of $\rho(t)$ is $\bv(t)$, then $\bv(t)=\bv(\varphi(u)) = \bl(u)$ where $\varphi$ is continuously differentiable, and $\partial_u \varphi(u) > 0$. Therefore, by \cref{def:equiv-traj}, $\bl(u)$ is equivalent to the realizable trajectory $\bv(t)$.
    
    $\impliedby$: Suppose $\bl : [0,1] \mapsto \mc{M}^{(d)}$ is equivalent to a realizable trajectory $\bv : [0,t_f] \mapsto \mc{M}^{(d)}$, which corresponds to a density operator evolution given by $\rho(t)$. Then, by \cref{def:equiv-traj}, $\bl(u) = \bv(\varphi(u))$ for some continuously differentiable map $\varphi : [0,1] \mapsto [0,t_f]$. In the nice operator basis, this corresponds to $\s(u)=\rho(\varphi(u))$ where $\s$ is the representation of $\bl$ in the operator space. Since $\bv$ is realizable
    \beq
      \dot{\r}(t) = -i[H(t),\rho(t)] + \lind_D \rho(t)\ .
    \eeq
    for some $H(t)$. Substituting $t=\varphi(u)$ gives
        \begin{align}
            \frac{1}{\partial_u \varphi(u)} \partial_u\s(u) &= -i[H(\varphi(u)),\s(u)] + \lind_D \s(u)\ ,
\end{align}
which implies
\beq
\partial_u\s(u) - \partial_u \varphi(u)\lind_D \s(u) = -i[\partial_u \varphi(u) H(\varphi(u)),\s(u)]\ ,
\eeq
    where $\partial_u \varphi(u) > 0$. Finally, applying \cref{lm:commutator_structure} gives the required result.
\end{proof}

Once we have a realizable trajectory, we can find the control Hamiltonian which realizes it:
\begin{mytheorem}
\label{th:realizable-trajectory-ham}
Let $\s(u), c(u)$ be as in \cref{th:hd_general_trajectory}. Then the actual realizable trajectory is given by the state $\r(t) = \s(\varphi^{-1}(t))$ where $\varphi(s) = \int_{0}^{s} c(u) du$.

Moreover, a control Hamiltonian which realizes this trajectory is 
\bes
\begin{align}
H(t) &= \sum_{j,k} h_{jk}(t)\ketb{\lambda_j(t)}{\lambda_k(t)} \\
h_{jk}&=-i\frac{\bra{\lambda_j}\dot{\r}-\lind_D\r\ket{\lambda_k}}{\lambda_j-\lambda_k}\ ,
\end{align}
\ees
where $\{\lambda_j(t)\}$, $\{\ket{\lambda_j(t)}\}$ denote the eigenvalues and eigenvectors of $\rho(t)$ respectively.
\end{mytheorem}

\begin{proof}
$\r(t) = \s(\varphi^{-1}(t))$ follows from the proof of \cref{th:hd_general_trajectory}. Then, $\dot{\r} = -i[H,\r] + \lind_D\r$. Rearranging gives $-i[H,\r] = \dot{\r} - \lind_D\r$. The particular form of $H$ follows from the proof of \cref{lm:commutator_structure}, specifically \cref{eq:h_jk}.
\end{proof}

An example of designing the control Hamiltonian based on the trajectory is provided in \cref{app:bitflip-sol}.

\subsubsection{Result for $f$-preserving trajectories}

The corollary below emphasizes that characterizing realizable $f$-preserving trajectories requires enforcing $f$-preservation in any one of the equivalent trajectories.
\begin{mycorollary}
    $f(\bl(u))$ in \cref{th:hd_general_trajectory} is constant over the trajectory iff $\bl(u)$ is equivalent to a realizable $f$-preserving trajectory.
\end{mycorollary}
\begin{proof}
    $\implies$: By \cref{th:hd_general_trajectory}, $\bl(u)$ is equivalent to a realizable trajectory $\bv(t)$, where $\bv(t)=\bl(\varphi(u))$ for some continuously differentiable $\varphi$. Then $f(\bv(t)) = f(\bl(\varphi(u))) = \text{const}$. Therefore, $\bv(t)$ is an $f$-preserving trajectory.
    
    $\impliedby$: $\bl(u)$ is equivalent to a realizable trajectory, $\bv(t)$ where $\bv(\varphi(t))=\bl(u)$ for some continuously differentiable $\varphi$. By assumption, $f(\bv(t)) = \text{const}$. Therefore $f(\bl(u)) = f(\bv(\varphi(t))) = \text{const}$.
\end{proof}

This answers the question that motivated this subsection of how to steer a system from a given initial point $\r$ (e.g., an unstable point) to a final point $\s$ (e.g., a stable point) while preserving the target $f$: we choose potential trajectories and check whether they are realizable trajectories. Crucially, the check can be done without knowing the control Hamiltonian due to \cref{th:hd_general_trajectory}.
Once a realizable trajectory is found, the control Hamiltonian can be found using \cref{th:realizable-trajectory-ham}.

\section{The single-qubit case: general theory}
\label{sec:qubit_theory}

We now specialize to the single-qubit case to gain intuition and illustrate the general theory. Due to the simplicity of the Bloch sphere picture, we present our results in this section using the Bloch vector formalism. Thus, we now consider differentiable trajectories $\bl:[0,1]\mapsto S^2$, the unit sphere in $\mathbb{R}^3$.

\subsection{Qubit tracking control}
\label{sec:qubit_theory-a}

Both the Hamiltonian and the density matrix can be expanded in the (unnormalized) Pauli basis: $H = h_0 I+ \bh \cdot\boldsymbol{\s}$ with $\bh  = (h_x,h_y,h_z)\in \mathbb{R}^3$, and $\r = \frac{1}{2}(I+\bv\cdot\boldsymbol{\s})$, where $\bv \in S^2$ (the unit sphere in $\mathbb{R}^3$) is now the Bloch vector, as opposed to the coherence vector.\footnote{The Bloch vector and the coherence vector differ by normalization, since the latter is defined via $\r = \frac{1}{2}I+\bv\cdot\boldsymbol{\s}$.} In this case, the unitary dynamics component in \cref{eq:le} is given by $Q\bv = 2\bh \cross\bv$, where $\cross$ denotes the cross product in $\mathbb{R}^3$ (see \cref{app:Bloch}). We tabulate the dissipators $D=(R,\bc)$ for various noise noise channels in \cref{app:Bloch}.

Therefore, the dynamics [\cref{eq:le}] and constraint [\cref{eq:constraint_matrix_dynamics_eq-CV}] equations that govern the tracking control problem in the single-qubit case are, respectively:
\begin{subequations}
\label{eq:reference}
\begin{align}
\label{eq:reference_dynamics_eq}
         & \dot{\bv} = 2\bh \cross \bv + R\bv + \bc\\ 
\label{eq:reference_constraint_eq}
        & \grad f \cdot( 2\bh \cross \bv + R\bv + \bc) = 0\ ,
\end{align}
\end{subequations}
where in the tracking control problem $\bh = \bh (\bv)$.

Note that the cross product is the specialization of the commutator (over the vector space of Hermitian matrices) to $\mathbb{R}^3$: they share the standard properties of bilinearity, anticommutativity, and the Jacobi identity. Moreover, $[A,B]$ is likewise orthogonal to both $A$ and $B$, i.e., $\< A, [A,B] \> = \< B, [A,B] \> = 0$, and both satisfy the cyclic property $\< A,[B,C]\> = \<C,[A,B]\>$. Finally, the cross-product vanishes when both the input vectors are collinear, similar to the commutator (which vanishes if the operators share the same eigenspace, i.e., if they are diagonal in the same basis).

\subsection{Uncontrollable target properties}

Using \cref{eq:ineffective_h_condition_vector} and $Q\bv = 2\bh \cross\bv$, uncontrollable target properties are characterized by: 
\beq
 \label{eq:inffective_h_condition}
        \grad f \cdot (\bh (\bv) \cross \bv) = \bh (\bv) \cdot (\bv \cross \grad f) =0\ .
\eeq
For a given Bloch vector $\bv$, this condition holds if $\bh (\bv)$ is orthogonal to $\grad f \cross \bv$, and for all control Hamiltonians $\bh$ iff $\grad f \propto \bv$. 

For a qubit (but not in higher-dimensional Hilbert spaces), \cref{prop:1} rules out the preservation of any target property which is purely a function of the eigenvalues of the density matrix, since these are given by $\lambda_{1,2} = \frac{1}{2}(1 \pm v)$ (see \cref{app:Bloch}). This includes the
set of R\'enyi $\alpha$-entropies: $f(\bv) = \frac{1}{1-\alpha}\ln\left(\lambda_1^\alpha+\lambda_2^\alpha\right)$, which reduces to the von Neumann entropy $-\lambda_1 \ln(\lambda_1)-\lambda_2 \ln(\lambda_2)$ for $\alpha \to 1$.

\subsection{Stable point loci}

Stable points are Bloch vectors for which $\bh (\bv)$ can be chosen so that $\dot{\bv}=0$ $\forall t$ (\cref{def:stable-points}). In the qubit case, the purity $P(\bv) = \frac{1}{2}(1+\|\bv\|^2)$ and
\beq
\label{eq:dotP-v-qubit}
\dot{P}= \bv\cdot\dot{\bv} = \bv \cdot (R\bv+\bc)\ .
\eeq

We now specialize \cref{prop:Pdot=0} to the qubit case, showing that $\dot{P}=0$ is necessary and sufficient for stable points.
\begin{myproposition}
\label{prop:locus}
In the single-qubit case, the locus of stable points is independent of the control Hamiltonian and is characterized by a constant purity, i.e., it satisfies
\beq 
\label{eq:pD}
\dot{P}  = 0\ .
\eeq
Conversely, the unstable points are characterized by $\dot{P} \ne 0$. 
\end{myproposition}

\begin{proof}
\cref{eq:dotP-v-qubit} is independent of the control Hamiltonian $\bh$, since 
$\bv\cdot \bh\cross \bv=0$. If $\bv\cdot(R\bv+\bc) =0$, then $R\bv+\bc$ is orthogonal to $\bv$ (or a zero vector) and hence can be written as $2\mathbf{h'}\cross \bv$ for some vector $\mathbf{h'}$.
We can set the Hamiltonian to be $-\mathbf{h'}$, which leads to $\dot{\bv} = 0$, and the point is stable.
Conversely, if $\dot{P} \neq 0$, then using \cref{eq:dotP-v-qubit} we have $\bv\cdot\dot{\bv} \ne 0$, i.e., $\dot{\bv} \ne 0$.
Therefore, the point is unstable.
\end{proof}

\begin{Example}
\label{ex:3}
Consider \cref{ex:1} again. The stable points are given by \cref{eq:pD}, which yields $\bv \cdot (R\bv+\bc) = \bv\cdot \text{diag}(-2\g,-2\g,0)\bv = -2\g(v_x^2+v_y^2)=0$, i.e., the $z$-axis. Points off the $z$-axis are unstable.
\end{Example} 

We proceed to give a geometric characterization of the stable point loci. As mentioned above, when working in a nice operator basis, $R$ in \cref{eq:matrix_le} is a symmetric matrix in the single qubit case. We can then diagonalize $R$ via an orthogonal matrix $O$: $R = ODO^T$ where $D = -\text{diag}(d_1,d_2,d_3)$ is a negative semi-definite diagonal matrix of $R$'s eigenvalues (some but not all of the eigenvalues can be zero), and $O$'s columns are $R$'s corresponding eigenvectors. The dissipator can then be rewritten as:
\begin{align}
            R\bv+\bc 
            = O(D\bw+\bc')O^T\ .
\end{align}
where $\bw = O^T\bv O,\bc' = O^T\bc O$. Note that the vector $\bc'$ is constrained to be in the column space of $D$. \cref{eq:pD} then implies that the locus of stable points satisfies $\bv^T(R\bv+\bc) = O \bw^T (D\bw+\bc') O^T=0$, i.e., the quadratic form
\begin{align}
\sum\limits_{i=1}^3\!{}^{\prime} d_i (w_i -r_i)^2 = r^2\ , \quad r_i = \frac{1}{2}\frac{c'_i}{d_i}\ ,
\label{eq:locus}
\end{align}
where $r = \|\mathbf{r}\|$, $\mathbf{r} = (r_1,r_2,r_3)$, and the notation $\sum\!{}^{\prime}$ indicates that if $d_i=0$ then the corresponding term is absent in the sum.  

\subsubsection{Unital channels}

For unital channels, $\bc=\mathbf{0}$, which reduces \cref{eq:locus} to $\sum_{i=1}^3 d_i w_i^2 =0$, which means that the locus of stable points is limited to:
\begin{enumerate}
\item The origin: $r=0$ and all $d_i>0$, or
\item A line: two of the $d_i>0$ and the line is along the axis corresponding to the zero $d_i$, or
\item A plane: two of the $d_i$ are zero, and the plane is through the origin, orthogonal to the axis corresponding to the nonzero $d_i$.
\end{enumerate}

\subsubsection{Non-unital channels}
For non-unital channels, where $\bc\ne\mathbf{0}$, \cref{eq:locus} describes:
\begin{enumerate}
\item An ellipsoid centered at $\mathbf{r}$: all $d_i>0$.
\item An elliptic cylinder: two of the $d_i>0$, the cylinder extends along the axis corresponding to the zero $d_i$.
\item Two parallel planes: e.g., $d_2=d_3=0$ and the two planes are through the origin and through $w_1=2r$, orthogonal to the axis corresponding to $d_1\ne 0$.
\end{enumerate}

Below, we give examples of all these different cases using various noise channels.

\subsection{Stable points and purity preservation}

Stable points are related to purity preservation in the single-qubit case:
\begin{myproposition}
\label{prop:distinct_purity}
If there exist two times $t_1<t_2$ such that the purity $P(t_1)=P(t_2)$, then the system must pass through a stable point $t'\in (t_1,t_2)$.
\end{myproposition}
    
\begin{proof}
The purity of a qubit is given by $P(t) = \frac{1}{2} [1+v^2(t)]$, where $v=\|\bv\|$. Since $P(t_1) = P(t_2)$, then by Rolle's theorem there exists a time $t'\in (t_1,t_2)$ such that
\beq
\label{eq:P'}
        0 = \dot{P}(t)|_{t'}\ .
\eeq
The last equality represents a stable point.
\end{proof}
Thus, any \emph{periodic} Bloch vector trajectory passes through a stable point. Note that this result does not generalize to higher-dimensional spaces since it relies on \cref{prop:locus}.

\subsection{$f$-preserving Hamiltonians and breakdown points}
Using the Bloch vector formalism, it is instructive to specialize \cref{th:main} to the single-qubit case.

\begin{mycorollary}
\label{th:main-qubit}

For any controllable target property of a qubit, the $f$-preserving component of the tracking control Hamiltonian may always be written as
\beq
 \label{eq:basic_control}
        \bh  = \frac{\grad f \cdot (R\bv + \bc)}{2\|\grad f \cross \bv \|^2} \grad f \cross \bv\ .
\eeq
\end{mycorollary}

\begin{proof}
This follows immediately from \cref{eq:basic_control-H} upon replacing $\lind_D \rho$ with $R\bv + \bc$ [\cref{eq:le}], $i[\rho,\grad f]$ with $\grad f \cross \bv$ (i.e., the commutator with the vector product, as argued above), and the coherence vector by the Bloch vector.
The factor of $2$ is an artifact of using $\boldsymbol{\sigma}$ as a basis for $\bh$ instead of $\boldsymbol{\s}/2$.
\end{proof}

We present an alternative proof that does not rely on \cref{eq:basic_control-H}. This proof has the advantage of exposing the components of $\bh$ that play no role in $f$-preservation but are crucial in extending the breakdown time, as will be shown later.

\begin{proof}
It follows from \cref{eq:inffective_h_condition} that to avoid the uncontrollable case, $\bv$ and $\grad f$ must be linearly independent vectors. Therefore the set $\{\bv, \grad f, \grad f \cross \bv\}$ forms a basis for $\mathbb{R}^3$. Expanding $\bh $ in this basis as 
\beq
\label{eq:h-decomp-basis}
\bh  = \alpha_1 \bv + \alpha_2 \grad f + \alpha_3 \grad f \cross \bv\ ,
\eeq 
and substituting into \cref{eq:reference_constraint_eq}, we obtain:
\begin{align}
\label{eq:find-alpha3}
&\grad f \cdot( 2(\alpha_1 \bv + \alpha_2 \grad f + \alpha_3 \grad f \cross \bv)\cross \bv + R\bv + \bc) = 0\ .
\end{align}
The terms involving $\a_1$ and $\a_2$ vanish, while $\grad f \cdot [(\grad f \cross \bv)\cross \bv] = -\|\grad f \cross \bv\|^2$, which implies:
\begin{align}
\label{eq:alpha_3}
\alpha_3 = \frac{1}{2}\frac{\grad f \cdot (R\bv + \bc)}{\|\grad f \cross \bv\|^2}\ .
\end{align}
The other two terms in \cref{eq:h-decomp-basis} vanish in the constraint equation [\cref{eq:reference_constraint_eq}], so they play no role in preserving $f$. Therefore the $f$-preserving component of $\bh$ is $\alpha_3 \grad f \cross \bv$.
\end{proof}

\cref{th:main-qubit} is consistent with our analysis of uncontrollable target properties in \cref{sec:unc-f}, since if $\grad f \propto \bv$, then $\bh \to\infty$ in \cref{eq:basic_control}. This also motivates specializing the definition of breakdown points (\cref{def:breakdown-gen}) to the qubit case:

\begin{mydefinition}
\label{def:breakdown-qubit}
For a given target property $f$ and dissipator $D$, we call a Bloch vector a \emph{breakdown point} $\bv_b$ if (a) $\grad f|_{\bv_b} \propto \bv_b$ and (b) $\dot{P}(\bv_b) \neq 0$ (i.e., $\bv_b$ is unstable).
\end{mydefinition}
Thus, a breakdown point is a point where $\bh $ [\cref{eq:basic_control}] diverges. As before, the \emph{breakdown time} $t_b$ is the time when the state reaches a breakdown point. 

In the context of \cref{ex:1}, the breakdown condition $\grad f|_{\bv_b} \propto \bv_b$ yields the entire $(x,y)$ plane, but note that the condition $\dot{P}_D(\bv_b)\ne 0$ excludes the origin. We find the breakdown time below, in \cref{sec:dephasing-bd}.

Note that \cref{th:main-qubit} fixes the component of $\bh$ in the $\grad f \cross \bv$ direction. In practice, keeping components along the directions $\grad f$ and $\bv$ non-zero may be advantageous, in case this simplifies the implementation of the control Hamiltonian. Moreover, the component along $\grad f$ can also affect the breakdown time (e.g., see \cref{subsec:bit-flip-discussion}).

The reason that only one component of $\bh$ appears in \cref{eq:basic_control} has a simple geometric interpretation: the component of $\bh $ along $\bv$ has no effect, while the $\grad f$ component moves the state \emph{along} $f$'s level set. Counteracting the dissipator requires $\bh $'s component $\bv \cross \grad f$ since both it and the dissipator move $\bv$ \emph{away} from $f$'s level set.

\begin{table*}
\begin{tabular}{ |c|c|c|c|c|}
         \hline
        Target property$\rightarrow$ & Coherence magnitude & Uhlmann Fidelity \\
         Noise channel $\downarrow$ & $f=v_x^2+v_y^2 $ & $f=\Tr(\sqrt{\sqrt{\sigma}\rho\sqrt{\sigma}})$ \\
         \hline \rule{0pt}{2em}
         dephasing ($Z$)&$-\frac{\gamma}{v_z}(v_y,-v_x,0)$&$\gamma [(k_0\bv-\bw) \cdot (v_x,v_y,0)]\frac{\bw\cross\bv}{\|\bw\cross\bv\|^2}$\\
         \hline \rule{0pt}{2em}
         bit-flip ($X$)&$\frac{-\gamma v_y^2}{v_z f}(v_y,-v_x,0)$&$\gamma [(k_0\bv-\bw) \cdot (0,v_y,v_z)]\frac{\bw\cross\bv}{\|\bw\cross\bv\|^2}$\\
         \hline \rule{0pt}{2em}
         bit-phase-flip ($Y$) &$\frac{-\gamma v_x^2}{v_z f}(v_y,-v_x,0)$&$\gamma [(k_0\bv-\bw) \cdot (v_x,0,v_z)]\frac{\bw\cross\bv}{\|\bw\cross\bv\|^2}$\\
         \hline \rule{0pt}{2em}
         depolarizing&$\frac{-2\gamma}{3v_z}(v_y,-v_x,0)$&$\frac23\gamma [(k_0\bv-\bw) \cdot \bv]\frac{\bw\cross\bv}{\|\bw\cross\bv\|^2}$\\
         \hline \rule{0pt}{2em}
         relaxation at temperature $T$&$\frac{-\gamma}{4v_z} (v_y,-v_x,0)$&$\frac14\gamma [(k_0\bv-\bw) \cdot (v_x,v_y,2(v_z-2a)]\frac{\bw\cross\bv}{\|\bw\cross\bv\|^2}$\\
         \hline \rule{0pt}{2em}
         \multirow{2}{*}{\parbox{3.8cm}{relaxation at temperature $T$ ($\g_1$)+ dephasing ($\g_d$)}} & &\\
         &$\frac{-\g_2}{2v_z} (v_y,-v_x,0)$&$\frac12 \g_2 [(k_0\bv-\bw) \cdot (v_x,v_y,2(v_z-2a)]\frac{\bw\cross\bv}{\|\bw\cross\bv\|^2}$\\
         \hline
\end{tabular}
        \caption{The control Hamiltonian [\cref{eq:basic_control}] necessary to preserve two target properties (coherence and fidelity) for different environments. The system state $\rho$ is represented by the Bloch vector $\bv=(v_x,v_y,v_z)$ and $\sigma$ by $\bw$, $k_0^2\equiv\frac{1-\|\bw\|^2}{1-\|\bv\|^2}$, $\g_2 \equiv 2\g_d + \frac{\g_1}{2}$ and $a\equiv [1+\exp(-\b \D)]^{-1}-1/2$ where $\D$ is the qubit energy gap and $\b=1/T$ is the inverse temperature.}
     \label{table:h_example_values}
\end{table*}

\subsection{Realizable trajectories}

Let $\bl:[0,1]\mapsto  S^2$ denote an arbitrary parameterized Bloch sphere trajectory that is not necessarily a solution of the master equation [\cref{eq:reference_dynamics_eq}].

\begin{mydefinition}
\label{def:P_D}
The \emph{trajectory purity} is
\beq
P[\bl(u)] = \frac{1}{2} [1+\|\bl(u)\|^2]\ .
\eeq
The \emph{dissipator-induced purity rate} is
\beq
\label{eq:P_D}
\partial_u {P}_D[\bl(u)] \equiv \bl(u)\cdot(R\bl(u)+\bc)\ .
\eeq 
\end{mydefinition}
The former is the formal purity associated with the trajectory $\bl(u)$. The latter can be interpreted as the rate of purity change due to the dissipator $D=(R,\bc)$ for the trajectory $\bl(u)$. We emphasize that this purity rate is not necessarily related to the solution of the master equation either.

From \cref{eq:dotP-v-qubit}, if $\bl$ were equivalent (in the sense of \cref{def:equiv-traj}) to a solution of \cref{eq:reference_dynamics_eq}, then $\partial_u {P}_D = \partial_u{P}$.
Conversely, is a trajectory that satisfies $\partial_u{P}_D \propto \partial_u{P}$ always realizable?
The following result resolves this question:

\begin{myproposition}
\label{th:general_trajectory}
$\bl(u)$ is equivalent to a realizable trajectory iff 
\beq
\label{eq:P_traj}
\partial_u P[\bl(u)] = c(u) \partial_u  P_D[\bl(u)]\ ,
\eeq 
where $0<c(u)<\infty$ $\forall u\in [0,1]$.
\end{myproposition}

\begin{proof}
From \cref{th:hd_general_trajectory}, if the Bloch vector trajectory $\bl(u)$ corresponds to a density matrix trajectory $\s(u)$, then \cref{eq:18} is satisfied iff $\bl(u)$ (which is the same as the coherence vector up to a scalar multiplicative factor) is equivalent to a realizable trajectory. \cref{duCu} can be written as
    \begin{align}
    \label{eq:33}
    \partial_u\bl(u) - c(u) [R \bl(u)+\bc] = 2\bh\cross\bl(u)\ .
    \end{align}
In $\mathbb{R}^3$, orthogonal directions are uniquely determined by the cross product, i.e., if $\mathbf{a},\mathbf{b}\in \mathbb{R}^3$, $\mathbf{a}\cdot\mathbf{b}=0 \iff \mathbf{a}=\mathbf{b}\cross\mathbf{c}$ for some non-zero $\mathbf{c}\in\mathbb{R}^3$. Thus, \cref{eq:33} holds iff
    \begin{align}
    \label{eq:35}
        \bl(u)\cdot[\partial_u\bl(u) - c(u) (R\,\bl(u)+\bc)]=0\ ,
    \end{align}
which is \cref{eq:P_traj}. Since $u$ was arbitrary, the statement must be true for all $u$.
\end{proof}

\section{The single-qubit case: examples}
\label{sec:examples}

In this section, we discuss several examples to illustrate the general theory. \cref{table:h_example_values} gives the values of the control Hamiltonian  [\cref{eq:basic_control}, all other components of $\bh$ set to 0] for two target properties: the square of the coherence magnitude 
\beq
f = [\Tr(\r \s^x)]^2 + [\Tr(\r \s^y)]^2 = v_x^2+v_y^2\ ,
\eeq
and the Uhlmann fidelity between two states $\r$ and $\sigma$:
\beq
f= \Tr(\sqrt{\sqrt{\sigma}\rho\sqrt{\sigma}})\ ,
\eeq

See \cref{app:example_details} for the derivation of most of the results in \cref{table:h_example_values}. \cref{fig:stable_region_plot} illustrates the geometry of the Bloch sphere for three different noise channels and the two target properties above in terms of level sets, stable points, and breakdown points.

In the following, we use the shorthand ``trajectory $\bl$ can be realized" to mean that trajectory $\bl$ is equivalent to a realizable trajectory.

\begin{figure*}
     \centering
        \subfigure[Dephasing with $f=$ coherence]{\includegraphics[width=.3\textwidth]{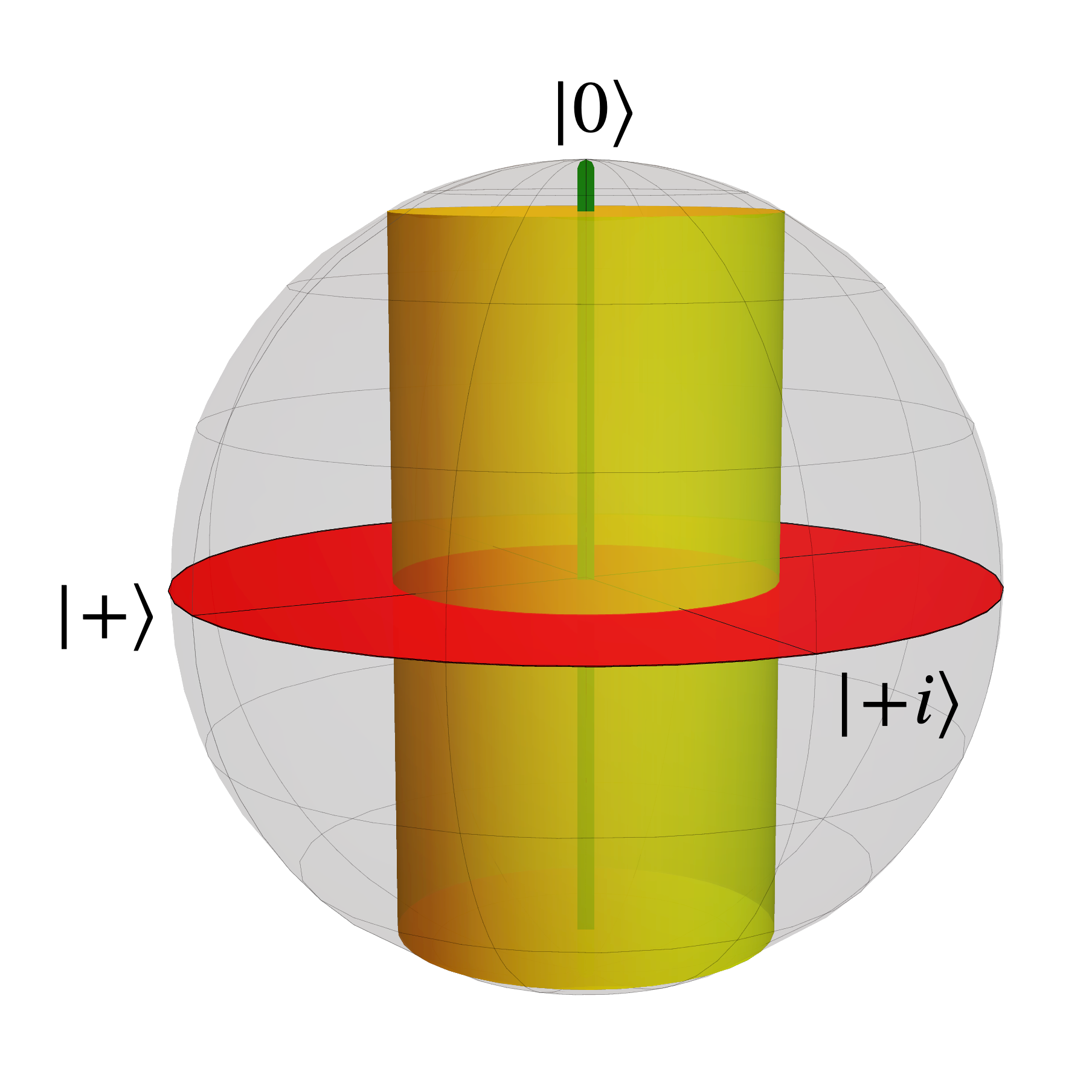}
        \label{fig:dephasing_coherence}}
        \subfigure[Bit-flip with $f=$ coherence]{\includegraphics[width=.3\textwidth]{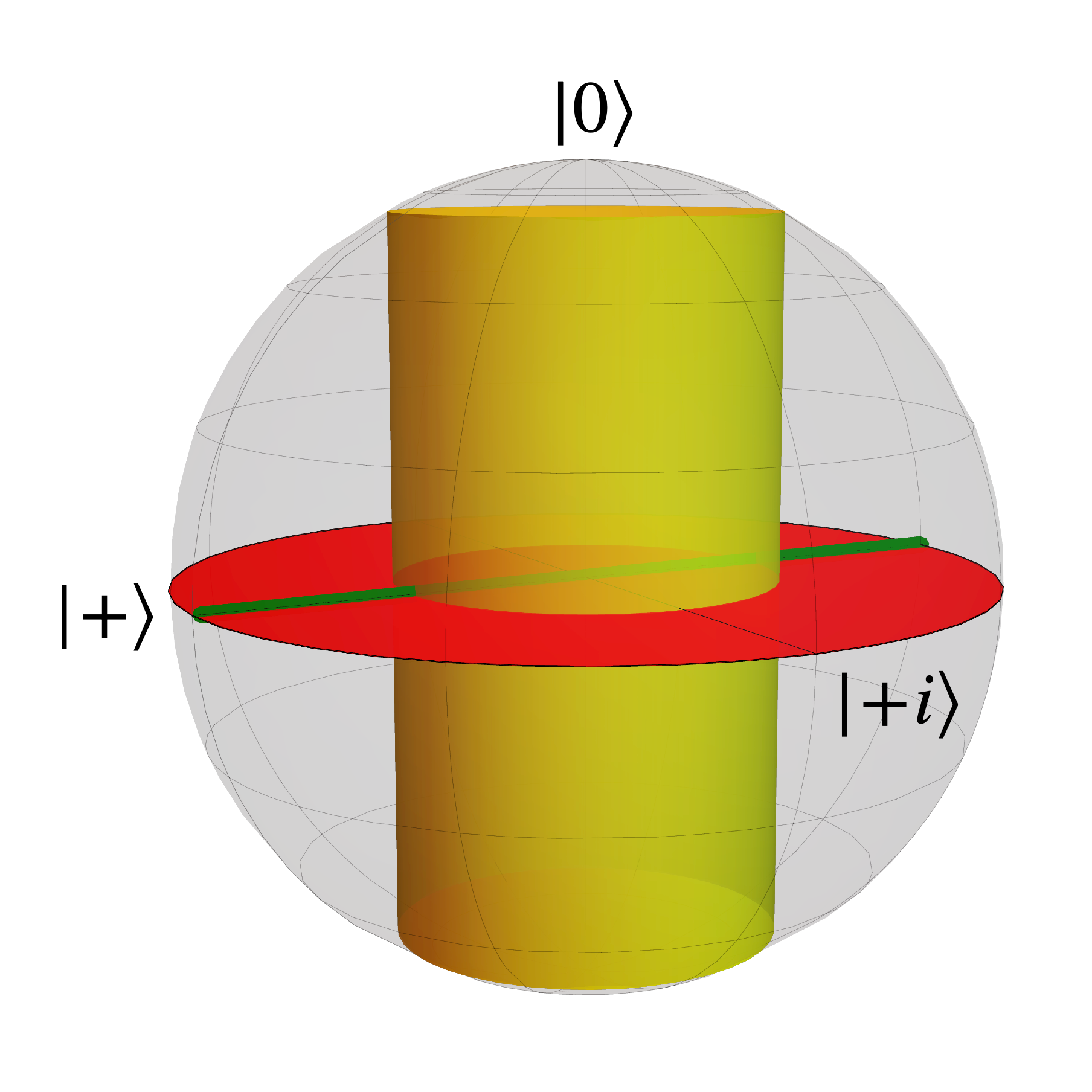}
        \label{fig:bitflip_coherence}}
        \subfigure[Relaxation with $f=$ coherence]{\includegraphics[width=.3\textwidth]{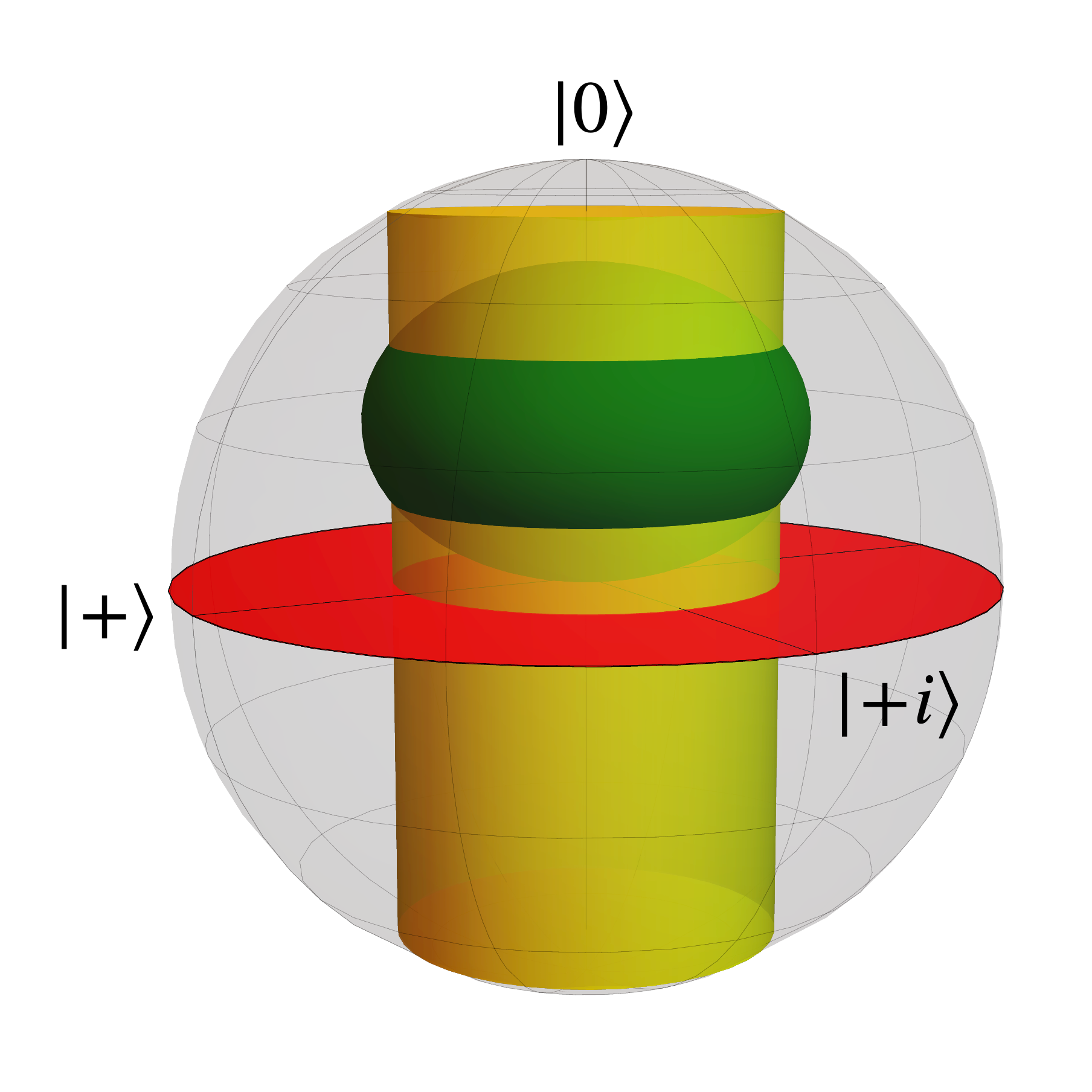}
        \label{fig:relaxation_coherence}}
        \subfigure[Dephasing with $f=$ fidelity]{\includegraphics[width=.3\textwidth]{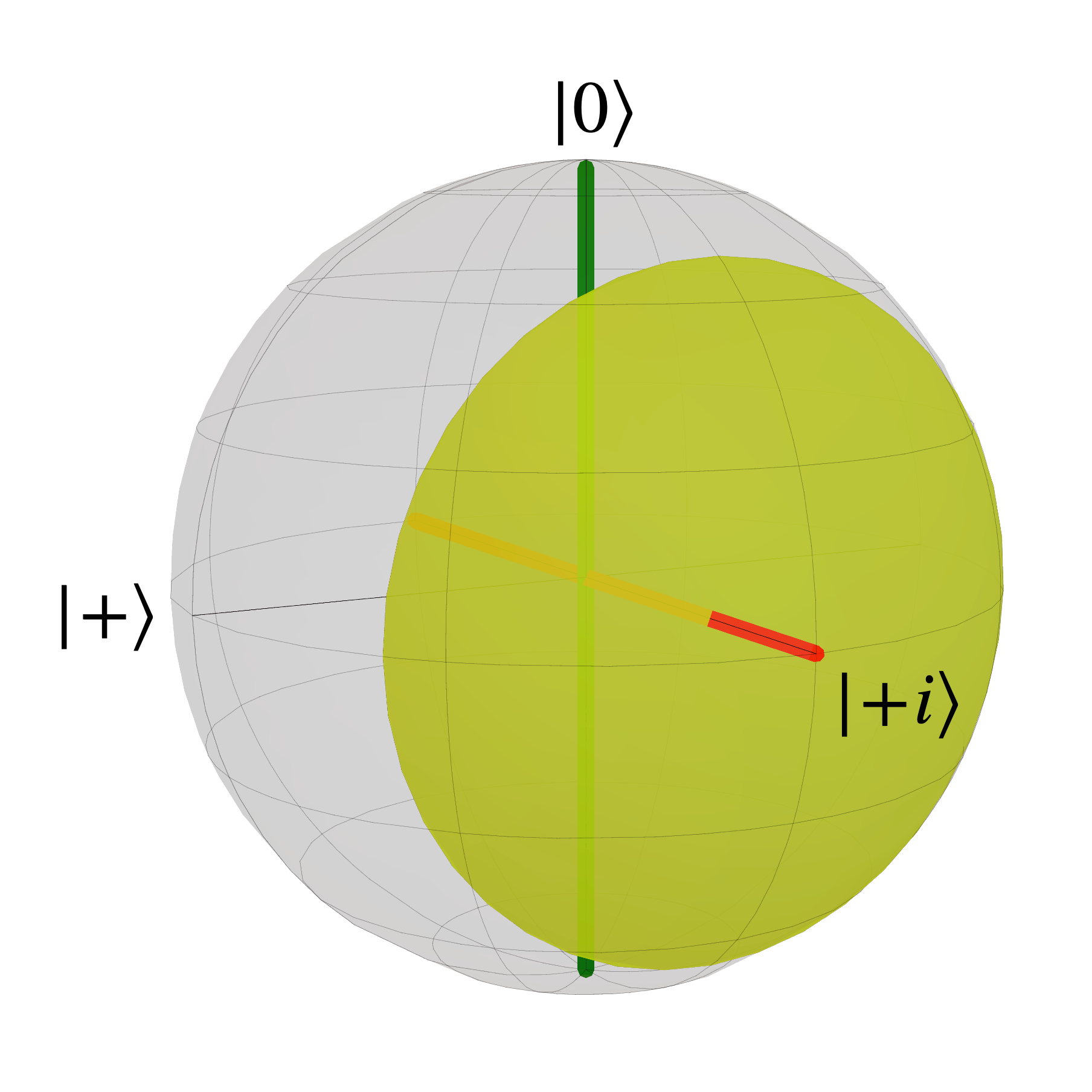}
        \label{fig:dephasing_fidelity}}
        \subfigure[Bit-flip with $f=$ fidelity]{\includegraphics[width=.3\textwidth]{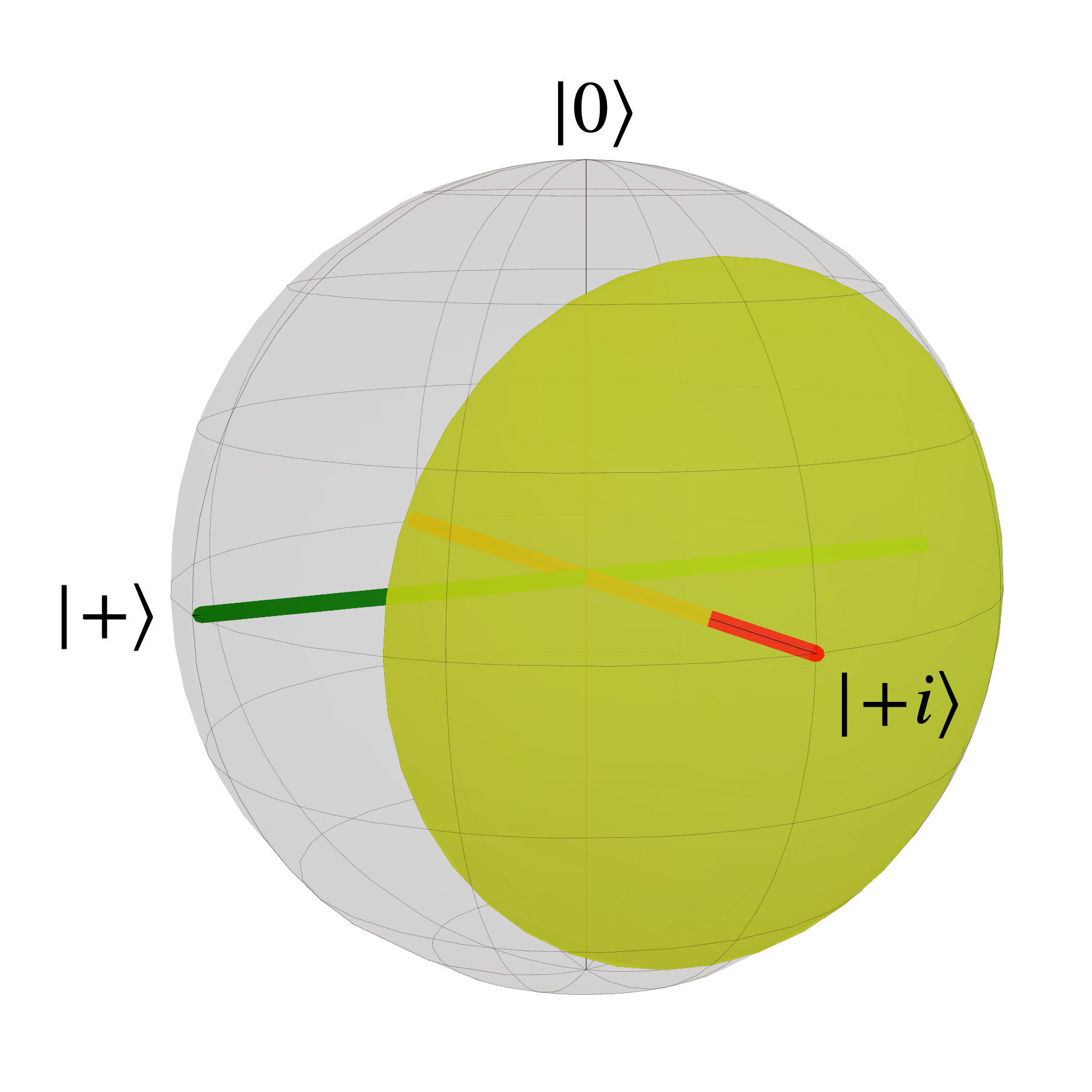}
        \label{fig:bitflip_fidelity}}
        \subfigure[Relaxation with $f=$ fidelity]{\includegraphics[width=.3\textwidth]{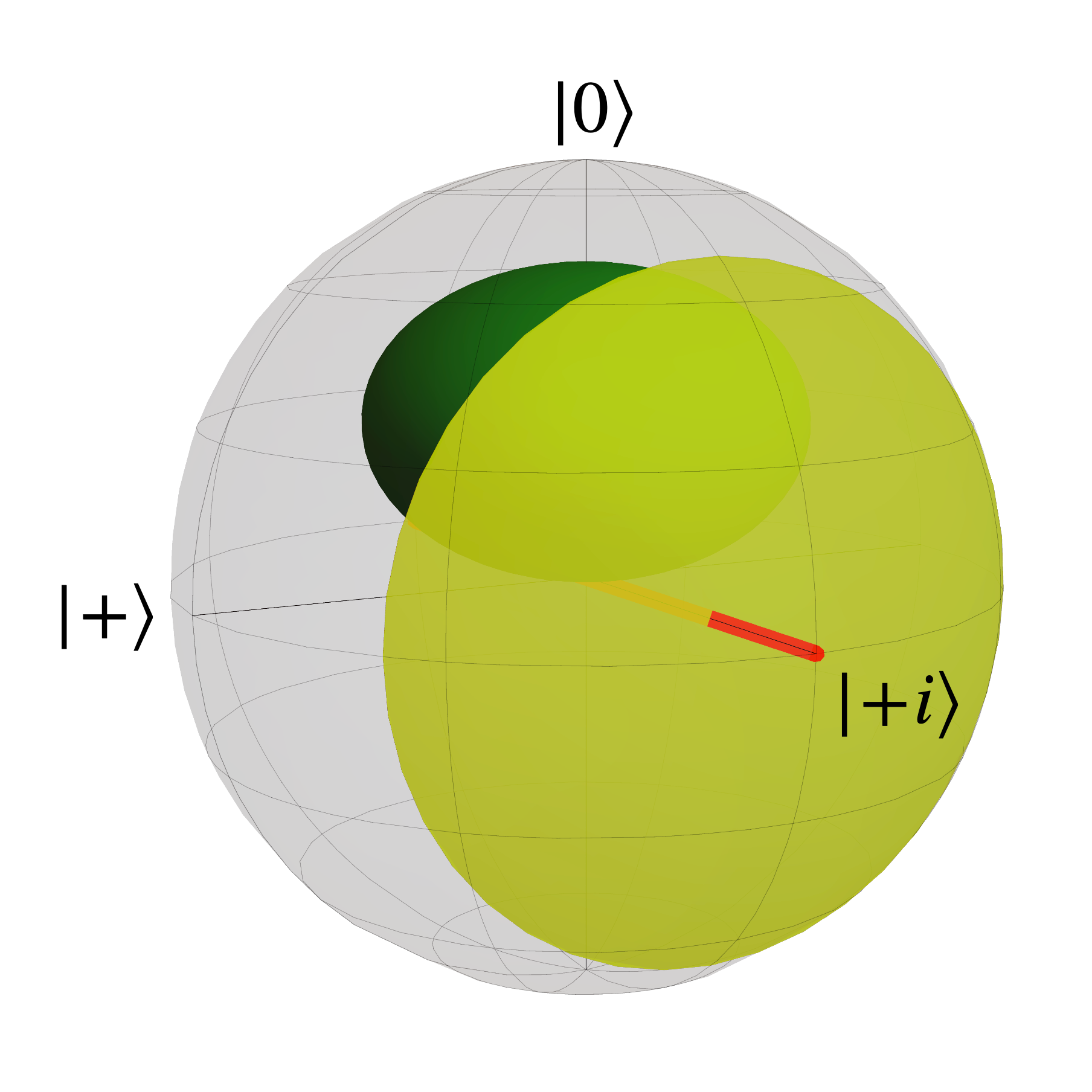}
        \label{fig:relaxation_fidelity}}
    \caption{Regions of level sets, stable points, and breakdown points for three different noise channels: dephasing (left column), bit-flip (middle column), and relaxation with $\beta\Delta=2$ (right column). The target properties are the coherence magnitude (top row) and the fidelity with $\ket{+i}$ (bottom row). The grey spheres represent the Bloch sphere of a qubit. The stable region [$\bv\cdot(R\bv+\bc)=0$] and breakdown region ($\grad f \propto \bv$) are indicated by green and red, respectively. The yellow regions represent the level sets for a given initial state with initial coherence $v_x^2+v_y^2=0.22$ in panels (a)-(c) and initial fidelity $0.8$ with $\bw=(0,1,0)$ in panels (d)-(f).}
    \label{fig:stable_region_plot}
\end{figure*}

\begin{figure}[h]
\includegraphics[width=0.36\textwidth]{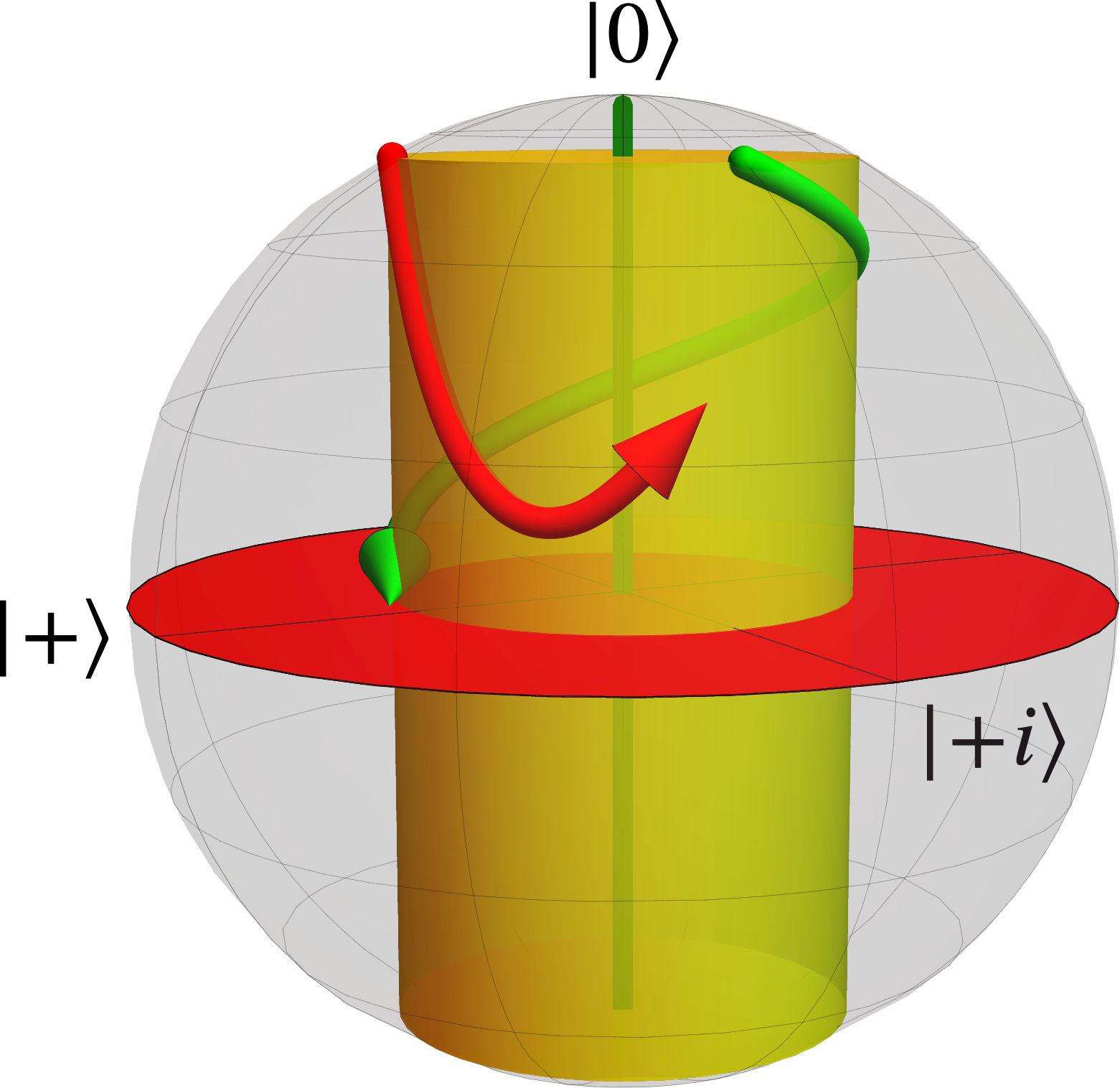}
        \caption{Illustration of \cref{th:general_trajectory} for coherence preservation under dephasing (\cref{ex:1}). The red trajectory has an increasing purity and, therefore, cannot be realized. The green trajectory has decreasing purity and terminates before the breakdown region. Since it satisfies all conditions of \cref{th:general_trajectory}, it is realizable.}
     \label{fig:gen_trajectory_illustration}
\end{figure}

\begin{figure}[h]
\includegraphics[width=0.36\textwidth]{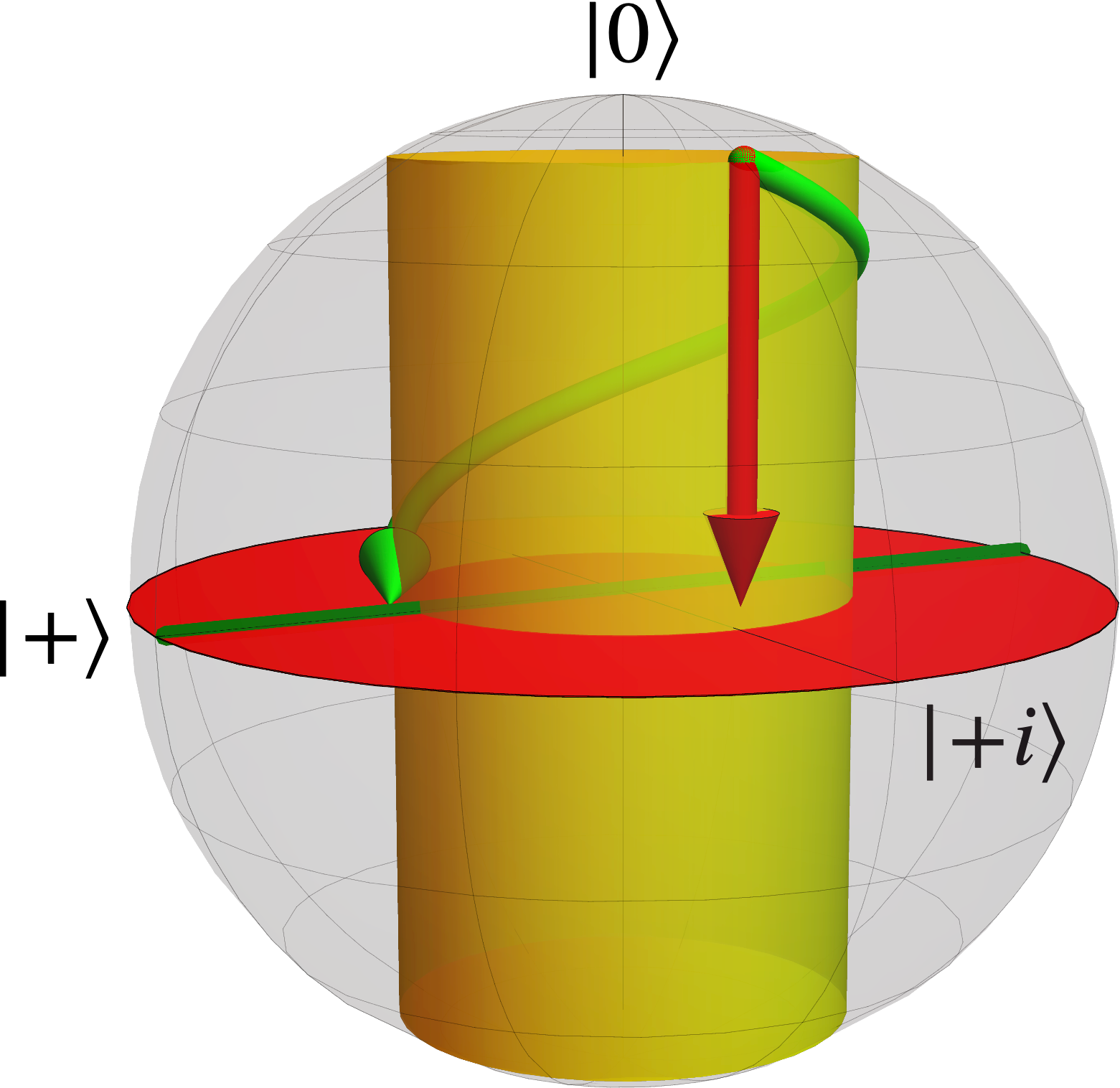}
        \caption{Landscape of the Bloch sphere for the bit-flip channel, showing the cylindrical level set (yellow surface), the breakdown set in the $(x,y)$ plane (red disk), the set of stable points along the $x$-axis (green line), and two trajectories: a coherence-preserving trajectory that travels along the level set towards the stable set (green arrow), and a trajectory with a breakdown point at its end (red arrow).}
     \label{fig:bitflip-coherence}
\end{figure}

\begin{figure*}
     \centering
        \subfigure[]{\includegraphics[width=.45\textwidth]{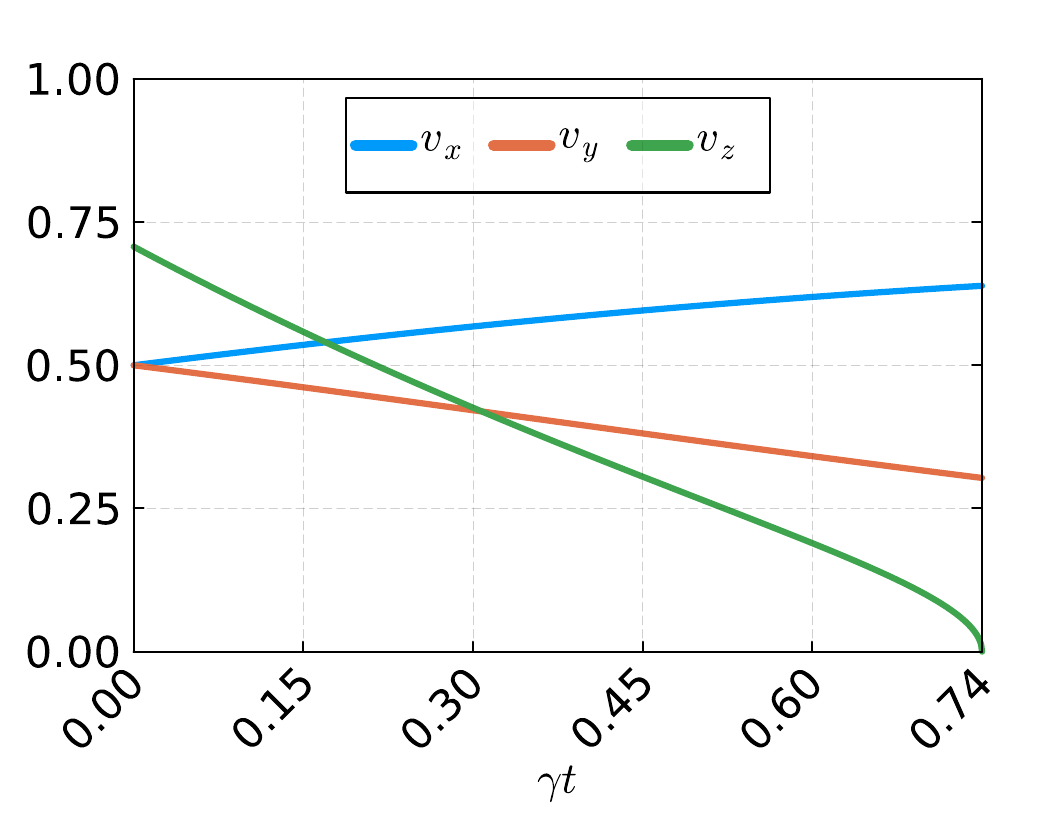}
        \label{fig:coherence_bloch_vector}}
        \subfigure[]{\includegraphics[width=.45\textwidth]{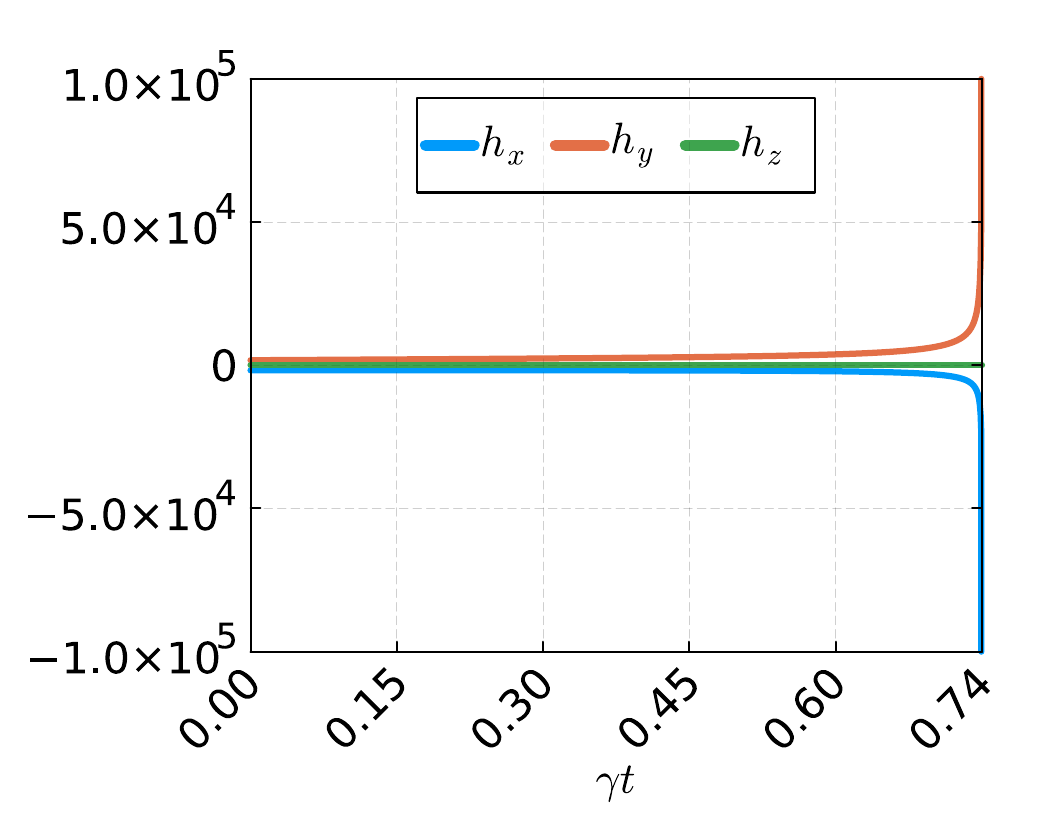}
        \label{fig:coherence_hamiltonian}}
    \caption{Time dependence of the Bloch vector components (a) and the control Hamiltonian components (b) for preserving the coherence magnitude ($v_x^2+v_y^2=0.5$) subject to the bit-flip channel, using only the necessary ($\a_3$) component of the $f$-preserving Hamiltonian [\cref{eq:basic_control}]. The initial state is $(1/2,1/2,1/\sqrt{2})$. Note that the Hamiltonian diverges near the breakdown point ($v_z=0, v_x\neq 0$).}
    \label{fig:t-dep-coherence}
\end{figure*}

\subsection{Coherence}
\label{sec:coherence}

We first focus on the square of the coherence magnitude, which we refer to as the ``coherence'' for simplicity, and whose target constant value we denote by $f_0$. Coherence is the experimental signature of the superposition structure (in a given basis) of quantum systems, which is one of the fundamental attributes of quantum mechanics~\cite{Glauber:1963aa}. Preserving coherence is, therefore, a task of fundamental interest, and the coherence-generating power of quantum operations has been studied extensively~\cite{Styliaris:2018aa}.

\subsubsection{Level set and breakdown points}
\label{sec:LS_BP}

Note that the coherence evades the no-go result of \cref{prop:1}: it does not depend only on $\|\bv\|$.
Indeed, for the coherence magnitude $f = f_0 = v_x^2(0)+v_y^2(0)$, we have $\grad f = (2v_x,2v_y,0)$ and the level set is a cylinder of radius $\sqrt{f_0}$, whose axis is the $z$-axis, as illustrated by the cylinders in \cref{fig:stable_region_plot}(a)-(c).
It follows from \cref{def:breakdown-qubit} that the set of breakdown points is the entire $(x,y)$ plane except for the stable points in the plane.
The role of the control landscape will be more evident below, where we consider several noise models.

\subsubsection{Dephasing}
\label{sec:dephasing-bd}

We revisit \cref{ex:1}, for which the level set with non-zero coherence magnitude $f_0=v_x^2(0)+v_y^2(0)$ forms a vertical cylinder. What are the realizable trajectories?
As shown in \cref{ex:3}, the unstable points are the points not on the $z$-axis.
For any such point we have $\partial_u P_D(\bl)=\bl \cdot(R\bl+\bc)=-2\gamma(l_x^2+l_y^2)<0$ where $\gamma$ is the dephasing rate.
However, for stable points, $\partial_u P_D(\bl)=\bl \cdot(R\bl+\bc)=0$ from \cref{prop:locus} and \cref{eq:dotP-v-qubit}.
Thus, as a consequence of \cref{th:general_trajectory}, a trajectory can be realized iff 
\begin{enumerate}
\item For the portion of the trajectory not on the $z$-axis, the derivative of the trajectory purity is also negative (i.e., the trajectory purity decreases) and 
\item For points on the $z$-axis on the trajectory, the derivative of the trajectory purity is zero.
\end{enumerate}

$\bl$ is equivalent to an $f$-preserving trajectory iff $\bl$ satisfies the conditions above and lies on a level set of $f$.
For coherence preservation with $f_0>0$, the level set contains no points on the $z$-axis.
Therefore, by \cref{th:no-breakdown-on-f,th:general_trajectory}, any trajectory on the cylindrical level set which (i) has a decreasing purity and (ii) does not cross the breakdown region, i.e., the $(x,y)$ plane, is realizable.
An example is illustrated in \cref{fig:gen_trajectory_illustration}: the purity of the red trajectory increases and, therefore, is not realizable.
However, the green trajectory is realizable because it has a decreasing purity throughout and terminates before the breakdown region. For such a trajectory, after decreasing, $l_z(t)$ must eventually reach $0$, i.e., the $(x,y)$ plane, which is the region of breakdown points.
The control Hamiltonian diverges at this point, and due to \cref{th:disconnect_breakdown}, it cannot reach a disconnected point.

Let us find the corresponding breakdown time. Using $v^2 = f_0+v_z^2$, we obtain:
\begin{align}
            \partial_t({v}_z^2) =  \partial_t({v}^2) = 2\dot{P}(\bv)  = -4\g f_0\ ,
\end{align}
whose solution is $v_z^2(t) = v_z^2(0)-4\gamma f_0 t$. Setting $v_z(t_b)=0$ yields the breakdown time as 
\beq
\label{eq:t_b-dephasing}
t_b = \frac{v_z^2(0)}{4\gamma f_0}\ , 
\eeq
which decreases, as might be expected, in inverse proportion to the dephasing rate $\g$ and the initial coherence value $f_0$. Note that $t_b$ is independent of the Bloch vector trajectory; as will be evident in later examples, this is not a general result.

Conversely, an $f$-preserving trajectory can only reach a stable point if it starts at a stable point, i.e., if the trajectory is a point on the $z$-axis.

\begin{figure}[h]
\includegraphics[width=0.36\textwidth]{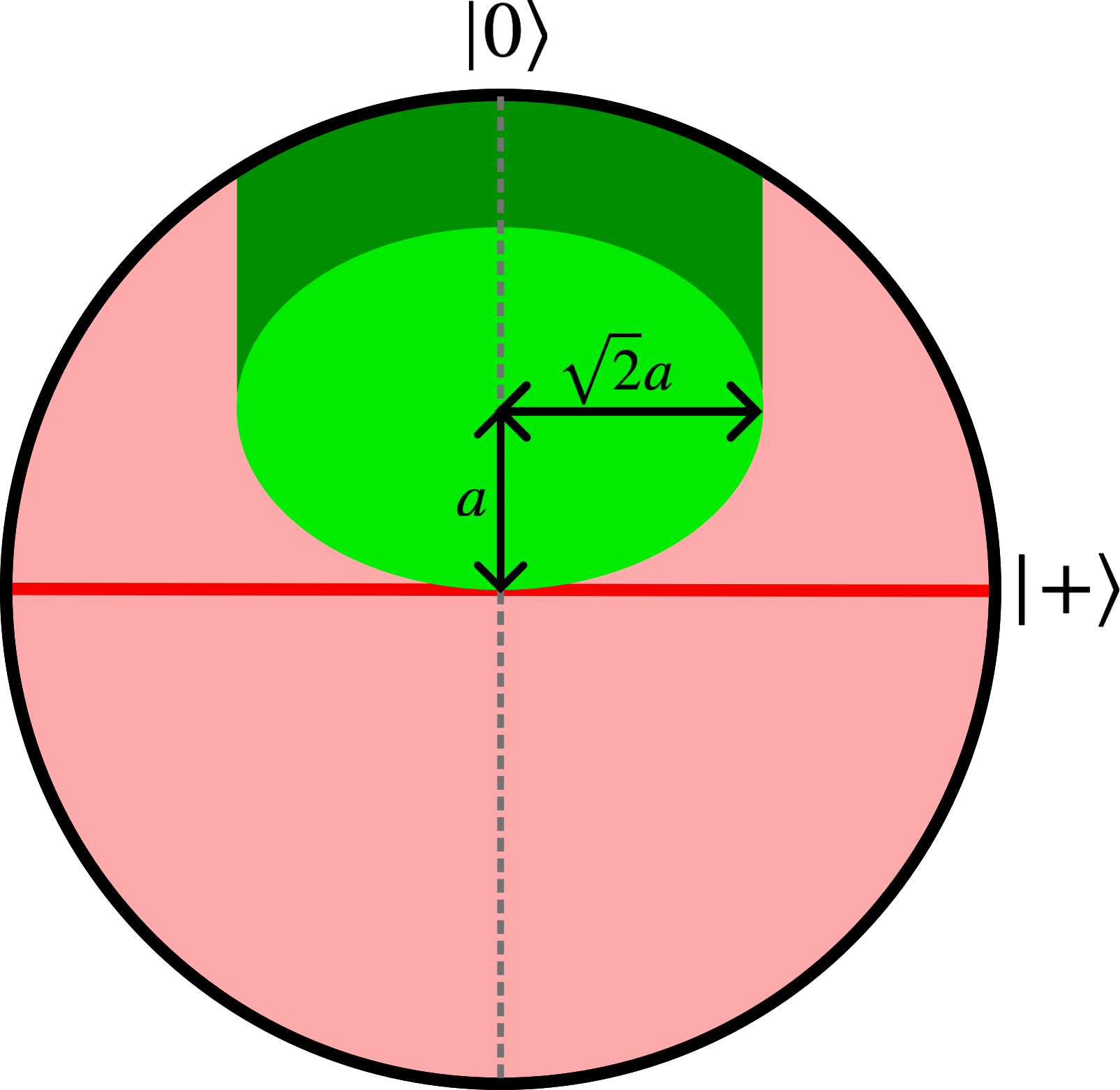}
    \caption{Cross-sectional view ($v_y=0$) of the initial state stability landscape for preserving coherence subject to pure relaxation ($a\equiv [1+\exp(-\b \D)]^{-1}-1/2$).
    The purity decreases in the dark green area, while it increases inside the light green area.
    The two green areas combined are the initial states which lead to stable points.
    The other pink areas always lead to a breakdown point.
    The same structure applies when additional dephasing is present: the ellipsoid shrinks uniformly along the $v_x$ and $v_y$ axes.
    Note that $\bv=(0,0,2a)$ is the steady state for relaxation.}
    \label{fig:relaxation}
\end{figure}

\subsubsection{Bit-flip}
\label{subsec:bit-flip-discussion}

For the bit-flip channel, $L_\a=\s^x$ and the rate $\g_\a(t) = \g>0$ in \cref{eq:L_a}. This corresponds to $R=\mathrm{diag}(0,-2\gamma,-2\gamma)$ and $\bc=\mathbf{0}$. The stable points [\cref{eq:dotP-v-qubit,eq:pD}] are given by $\dot{P}=\bv\cdot R\bv=-2\gamma(v_y^2+v_z^2)=0$, which is the $x$-axis. The breakdown condition $\grad f|_{\bv_b} \propto \bv_b$ yields the entire $(x,y)$ plane but excludes the $x$-axis (the set of stable points). 
For unstable points, we have $\dot{P}_D(\bv)=-2\gamma(v_y^2+v_z^2)<0$.
Thus, as a consequence of \cref{th:general_trajectory}, any trajectory with a decreasing trajectory purity can be realized. Consequently, all trajectories with a fixed non-zero coherence $f_0$ and monotonically decreasing $v_z$ component can be realized as $f$-preserving trajectories.

Remarkably, it is possible to reach a stable point (the $x$-axis) under the bit-flip channel while staying on the level set.
Namely, a trajectory that lies on the cylindrical level set, has a decreasing $|v_z|$ component, and ends on the $x$-axis satisfies \cref{th:general_trajectory}.
Therefore, such a trajectory has a control Hamiltonian that realizes such a trajectory; see \cref{app:bitflip-sol} for a derivation of the corresponding control Hamiltonian.
This requires enforcing $\a_2\ne 0$ in \cref{eq:h-decomp-basis}.

On the other hand, trajectories that do not terminate on the $x$-axis instead end up at a breakdown point in the $(x,y)$ plane. As we show in  \cref{app:bitflip-sol}, working with \cref{eq:basic_control} while setting the other components ($\a_2$ and $\a_3$) of $\bh$ to $0$, we find the corresponding breakdown time to be:
\beq
t_b = \frac{1}{4\gamma} \ln\left((1+ \xi)\exp(\frac{v_z^2(0)}{\xi f_0})-\xi \right)\ ,
\label{eq:t_b-bitflip}
\eeq
where $\xi = v_y^2(0)/v_x^2(0)$. Thus, $t_b$ again decreases as the bit-flip rate $\g$ and initial coherence value $f_0$ grow.

\cref{fig:bitflip-coherence} illustrates the two types of trajectories described above: any trajectory that terminates at a stable point has an infinite breakdown time, and conversely, any trajectory that terminates at a breakdown point has a finite breakdown time.
\cref{fig:t-dep-coherence} shows the time-dependence of the Bloch vector components along with the components of the control Hamiltonian in the latter case. The $h_x$ and $h_y$ components diverge at the breakdown time, where $v_z=0$.

\subsubsection{Depolarization}

The depolarizing channel is defined via $\mc{L}_D(\rho) = \frac{1}{3}\sum_{\a\in\{x,y,z\}} \gamma (\sigma^\a \rho \sigma^\a -\rho)$.
This corresponds to $R=-\frac{4}{3}\g I$ and $\bc=\mathbf{0}$.
The stable points [\cref{eq:dotP-v-qubit,eq:pD}] are given by $\dot{P}=-\frac{4}{3}\g v^2=0$, which is the origin $\bv=\bf{0}$. 
For unstable points, $-\frac{4}{3} \g v^2 <0$.
As a consequence of \cref{th:general_trajectory}, any trajectory with a decreasing trajectory purity can be realized. If the trajectory has a fixed coherence magnitude, it can be realized as an $f$-preserving trajectory.
The set of breakdown points is the entire $(x,y)$ plane, excluding the origin (the stable point).
Since $f$-preserving trajectories must have a monotonically decreasing $v_z(t)$, they must terminate at a breakdown point given that the initial state had non-zero coherence.

The breakdown time is the time it takes to reach the $(x,y)$ plane. The calculation is identical to the one leading to \cref{eq:t_b-dephasing}, except that now we obtain $\dot{v}_z^2 =  -\frac{8}{3} \g (f_0+v_z^2)$, whose solution is $v_z(t) = [\left(f_0+v^2_z(0)\right) e^{-\frac{8 \gamma  t}{3}}-f_0]^{1/2}$.
Setting $v_z(t_b)=0$ yields the breakdown time as 
\beq
t_b =\frac{3}{8 \g} \ln \left(\frac{v^2_z(0)}{f_0}+1\right)\ , 
\eeq
which is again independent of the Bloch vector trajectory.

\subsubsection{Relaxation}
\label{sec:relaxation-preservation}

So far, we have considered examples of unital channels. 
We next analyze the case of relaxation, which is non-unital ($\bc\ne\mathbf{0}$). 
For relaxation subject to coupling to a thermal environment at inverse temperature $\b$ (also known as generalized amplitude damping), the Lindblad operators are $\s^-$ and $\s^+$ with respective rates $\g_- = \gamma p_0$ (relaxation) and $\g_+ = \gamma (1-p_0)$ (absorption), where $p_0 = [1+\exp(-\b \D)]^{-1}$ is the ground state probability for a qubit with energy gap $\D$ at inverse temperature $\b$, and $\gamma$ is the relaxation rate at zero temperature. 
In this case, $R=\mathrm{diag}(-\gamma/2,-\gamma/2,-\gamma)$ and $\bc=(0,0,2\gamma a)$, where $a=p_0-1/2$.
The stable points $\dot{P}=\bv\cdot(R\bv+\bc)=0$ are given by
the ellipsoid $\frac{1}{2}v_x^2+\frac{1}{2}v_y^2+(v_z-a)^2=a^2$, centered at $(0,0,a)$ with major axis $\sqrt{2}a$ along $v_x$ and $v_y$, and minor axis $a$ along $v_z$. 
Since the ellipsoid touches the origin, the set of breakdown points is the entire $(x,y)$ plane except the origin. 
At all points inside this ``stability ellipsoid'', the dissipator-induced purity increases [$\dot{P}_D(\bv)>0$] while outside, it decreases. 
As a consequence of \cref{th:general_trajectory}, any trajectory for which the trajectory purity increases inside the ellipsoid and decreases outside can be realized.

As illustrated in \cref{fig:relaxation}, the initial state determines whether a stable point is reachable (i.e., whether the breakdown time is infinite).
In particular, when the initial state has (i) coherence magnitude $f_0 = v_x^2+v_y^2 < 2a^2$, is outside the stability ellipsoid and in the northern hemisphere (dark green) or (ii) is inside the stability ellipsoid (light green), the state can end up at a stable point on the stability ellipsoid.
The purity increases for initial points inside the stability ellipsoid and decreases for initial points outside it. 
When the level set has a radius larger than $\sqrt{2}a$ (major axis of the stability ellipsoid), every point on the level set ends up at a breakdown point.
For these initial states, the breakdown time can be found analytically and is given by \cref{eq:tb-deph+rel} below.

\subsubsection{Dephasing + Relaxation}
\label{sec:deph+rel}

If the dephasing rate is $\gamma_d$ and the relaxation rate is $\gamma_1$, the combined dissipator is 
$R=\mathrm{diag}(-\g_2, -\g_2, -\g_1)$ and $\bc=(0,0,2\g_1a)$, where $\g_2 = 2\gamma_d+\gamma_1/2$, and $a=[1+\exp(-\b \D)]^{-1}-1/2$ as in the pure relaxation case.
The region of stable points is then the ellipsoid $\g_2(v_x^2+v_y^2)+\g_1(v_z-a)^2=\g_1 a^2$, centered at $(0,0,a)$.
The presence of dephasing, therefore, simply shrinks the stability ellipsoid for pure relaxation in the $(x,y)$ plane.

The landscape of stable points follows the same structure as that of pure relaxation: when the initial state has (i) coherence $f_0\equiv v_x^2+v_y^2 < a^2\frac{\g_1}{\g_2}$, is outside the stability ellipsoid and in the northern hemisphere or (ii) is inside the stability ellipsoid, the state will end up at a stable point on the stability ellipsoid.
For states with $f_0>a^2\frac{\g_1}{\g_2}$, all trajectories necessarily end up at a breakdown point.
We show in \cref{app:breakdown-relaxation}, using \cref{eq:basic_control}, that in this case the breakdown time is given by
\begin{align}
\label{eq:tb-deph+rel}
	t_b &= \frac{1}{2\g_1}  \log {\left(\frac{(v_z(0) - a)^2 + \Omega^2}{a^2 + \Omega^2}\right)} \\
         & \qquad + \frac{a}{\Omega\g_1} \left(\arctan{ \frac{v_z(0)-a}{\Omega} } + \arctan{ \frac{a}{\Omega} }\right)\ , \notag
\end{align}
where $\Omega^2 \equiv \frac{\g_2}{\g_1}f_0-a^2$.
Setting $\gamma_d=0$ gives the breakdown time for the pure relaxation case of \cref{sec:relaxation-preservation}, with $\g_2 =\gamma_r/2$ and $\Omega$ correspondingly modified.

Similarly, for points below the stability ellipsoid (i.e., $\g_2(v_x^2+v_y^2)+\g_1(v_z-a)^2>\g_1 a^2$, $f_0<a^2 (\g_1/\g_2)$ and $v_z<a$), we obtain:
    \begin{align}
\label{eq:tb-deph+rel2}
	t_b &=  \frac{1}{2\g_1} \log {\left(\frac{(v_z(0) - a)^2 - \Omega^2}{a^2 - \Omega^2}\right)} \\
         & \qquad - \frac{a}{\Omega \g_1} \left(\arctanh{ \frac{v_z(0)-a}{\Omega} } + \arctanh{\frac{a}{\Omega} }\right)\ . \notag
     \end{align}
where $\Omega^2 \equiv -(\frac{\g_2}{\g_1}f_0-a^2)$.

\subsection{Fidelity}
\label{sec:fidelity-preservation}
We now set the target quantity to be the Uhlmann fidelity $F^2(\r,\s)$ between two states $\r$ and $\s$. We show in \cref{app:Bloch-Uhlmann} that in terms of the corresponding Bloch vectors, the fidelity becomes $f\equiv F^2(\bv,\bw) = \frac{1}{2}\left(1+\bv\cdot\bw + [(1-\|\bv\|^2)(1-\|\bw\|^2)]^{1/2}\right)$.
Thus, the fidelity, like the coherence, evades the no-go result of \cref{prop:1}: it is a function not only of $\|\bv\|$.

\subsubsection{Level set and breakdown points}
\label{sec:LS_BP-f}

The level set for fidelity is a complex surface, but when $\|\bw\|=1$, it simplifies to a plane, as indicated in \cref{fig:dephasing_fidelity,fig:bitflip_fidelity,fig:relaxation_fidelity} by the yellow surface. To identify the breakdown points, we first compute $\grad f$:
\beq
\label{eq:gradf-fid}
\grad f = \frac12 \left( \bw -k_0\bv \right)\ , \quad k_0=\left[\frac{(1-\|\bw\|^2)}{(1-\|\bv\|^2)}\right]^{1/2}\ .
\eeq
Thus, $\grad f \cross \bv=\frac12 (\bw\cross\bv)$, and using \cref{eq:alpha_3}:
\beq
\alpha_3 = \frac{(\bw-k_0\bv) \cdot (R\bv + \bc)}{\|\bw\cross\bv\|^2}\ .
\label{eq:uhlman_a3}
\eeq
It follows that it is not possible to preserve the fidelity if either ($\norm{\bv}=1$ and $\norm{\bw}\neq1$) or $\bv \propto \bw$.

For $\|\bw\|=1$, the set of points where $\grad f \propto \bv$ is $\bv=c\bw$ for any $c\in[-1,1]$ . The breakdown points are therefore $\{c\bw:c\in[-1,1]\}$ minus the set of stable points. The latter set depends on the noise models, which we consider next. Henceforth, we consider the simplifying assumption $\|\bw\|=1$.

For further analysis, we decompose the Bloch vector $\bv$ as the sum of vectors parallel to and perpendicular to $\bw$: $\bv = \alpha_w \bw + \alpha_p \mathbf{p}$ where $\mathbf{p}\cdot\bw=0$, $\|\mathbf{p}\|=1$ and $\alpha_p |_{t=0} > 0$.
Since the target property simplifies to $f = \frac12 (1+\bv\cdot\bw) = \frac12 (1+\alpha_w)$ for $\|\bw\|=1$, only $\alpha_p$ and the direction of $\mathbf{p}$ can be $t$-dependent.
To further simplify the analysis, we also assume that the fidelity-preserving Hamiltonian enforces a trajectory where the direction of $\mathbf{p}$ is fixed (we call this a \emph{fixed-$\mathbf{p}$ Hamiltonian}).
The constraint equation then simplifies to
\begin{equation}
\label{eq:fidelity-control-dynamics}
\dot{\alpha_p}\mathbf{p} = 2 \bh \cross \bv + R\bv+\bc\ .
\end{equation}
Taking the dot product with $\bv$ yields
\begin{equation}
\label{eq:fidelity-control}
\alpha_p\dot{\alpha_p} = \bv\cdot( R\bv+\bc) = \dot{P}_D(\bv)\ .
\end{equation}
Multiplying \cref{eq:fidelity-control-dynamics} by $\alpha_p$ and substituting \cref{eq:fidelity-control} gives
\begin{equation}
\bv\cdot( R\bv+\bc) \mathbf{p} = 2 \alpha_p \bh \cross \bv + \alpha_p (R\bv+\bc)\ .
\end{equation}
Comparing the coefficients of each of the orthogonal components $\bw$, $\mathbf{p}$ and $\bw \cross \mathbf{p}$ then gives the corresponding control Hamiltonian $\bh$ as a function of the instantaneous state.

\subsubsection{Dephasing/Bit-flip}
\begin{figure*}
\centering
        \subfigure[$F^2=0.75$]{\includegraphics[width=.30\textwidth]{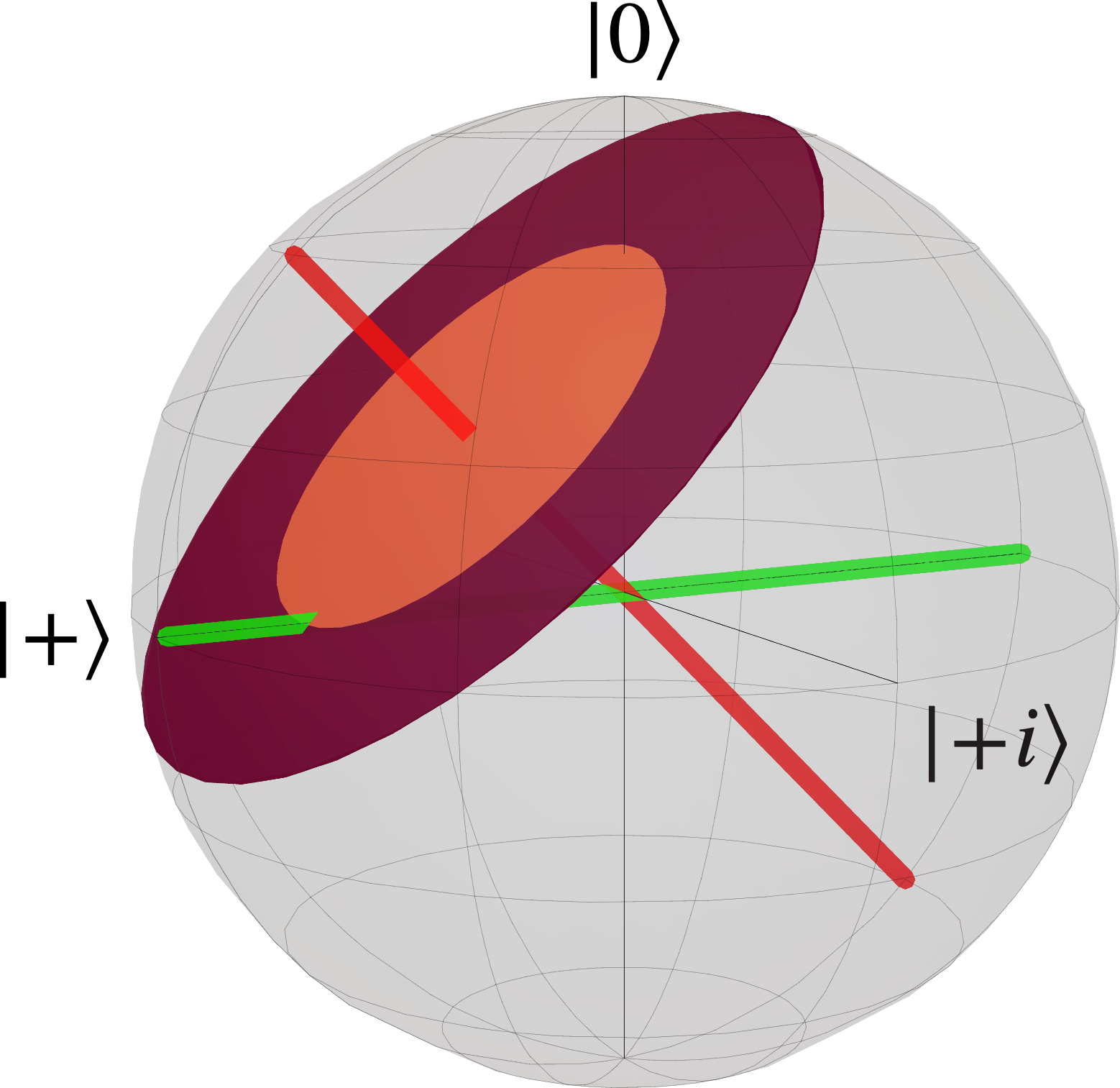}
        \label{fig:bitflip_initial_stable}}
        \subfigure[$F^2=0.9$]{\includegraphics[width=.30\textwidth]{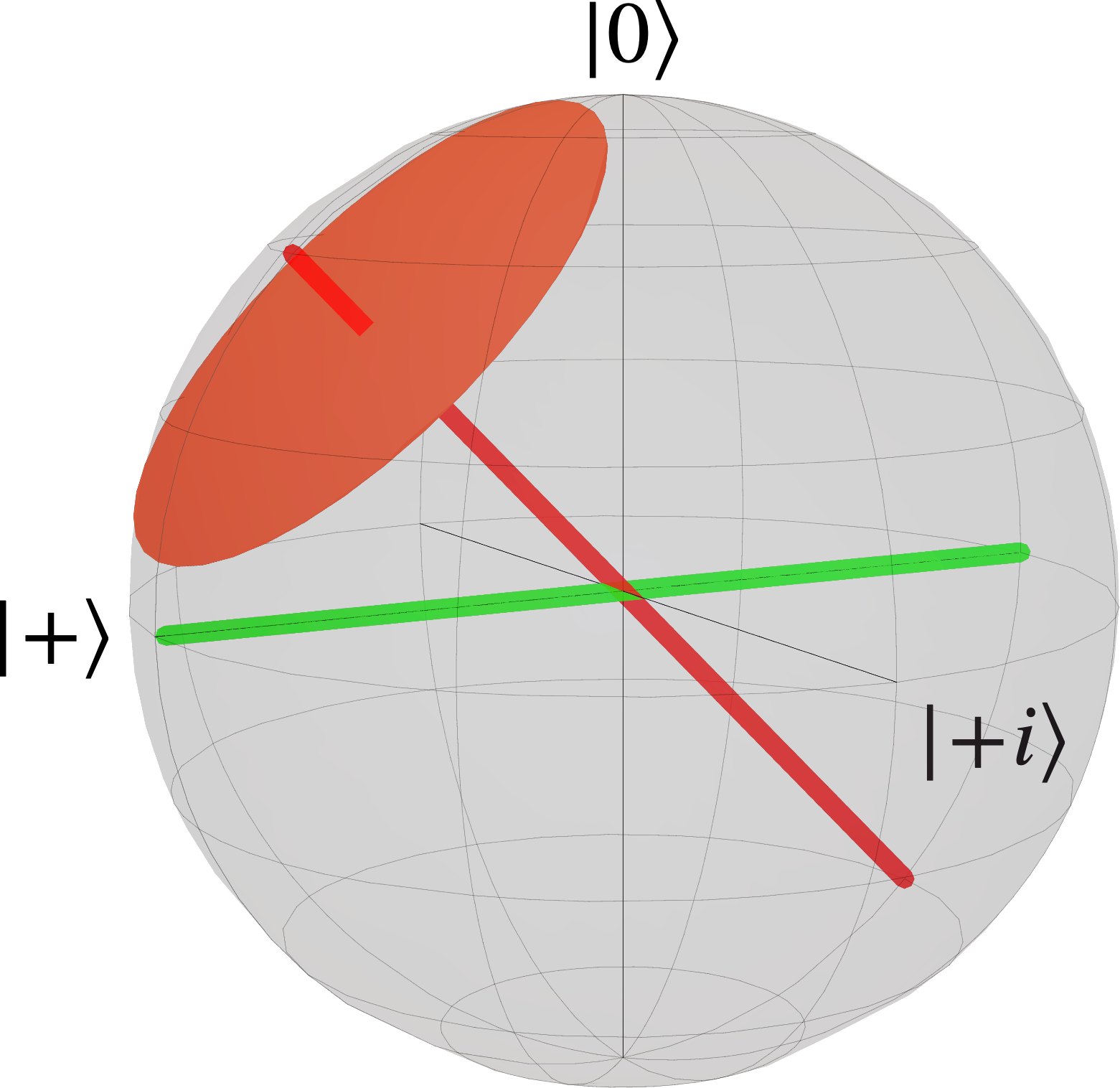}
        \label{fig:bitflip_initial_unstable}}
        \subfigure[$F^2=0.75$, with a fixed-$\mathbf{p}$ Hamiltonian]{\includegraphics[width=.30\textwidth]{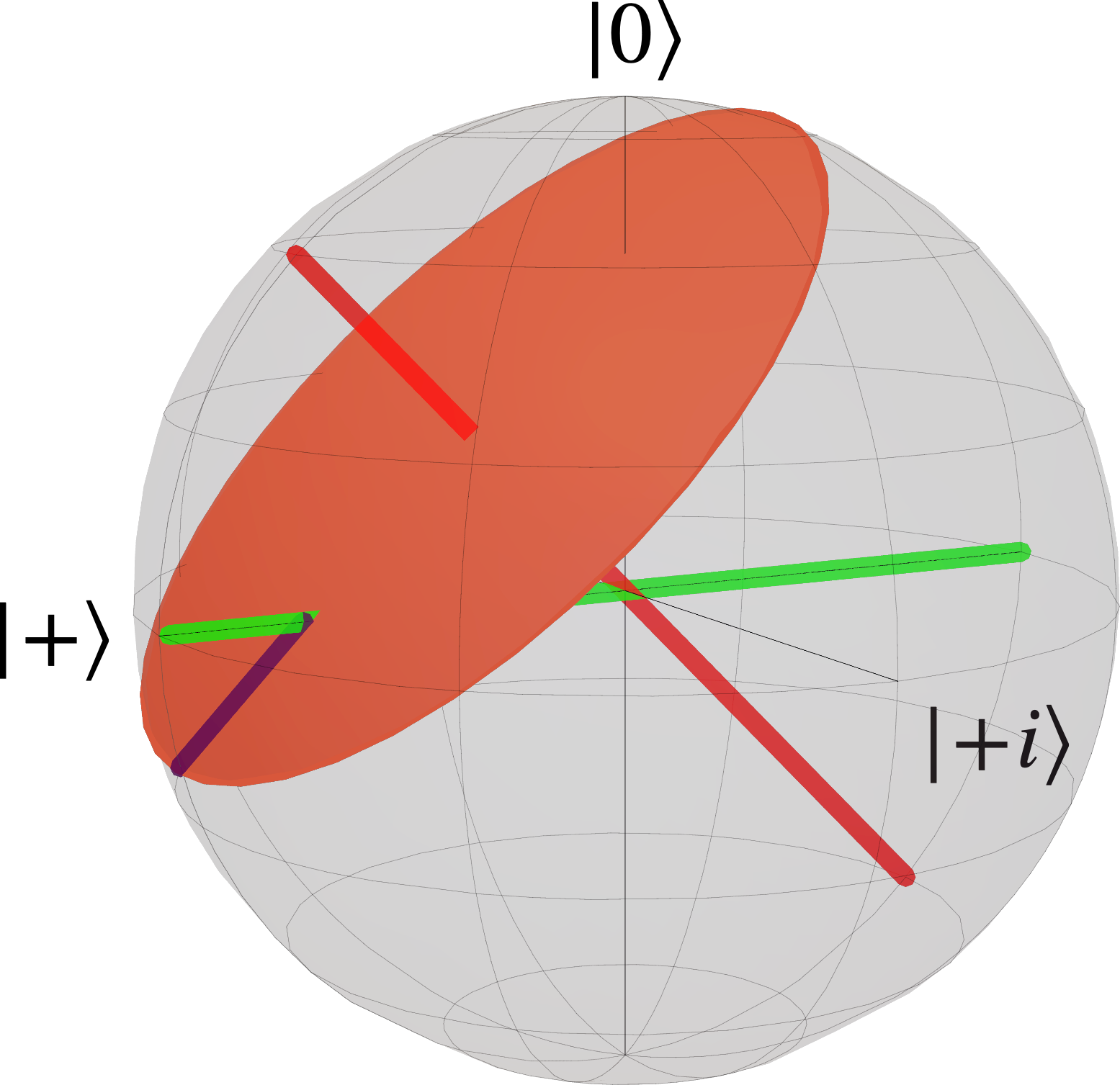}
        \label{fig:bitflip_initial_stable_specific_ham}}
    \caption{Initial state stability landscape for preserving the fidelity with $\bw = (1/\sqrt{2},0,1/\sqrt{2})$, for three different initial fidelity values, subject to the bit-flip channel.
    The red (green) region is the set of breakdown (stable) points.
    The level sets of constant fidelity are the orange and purple (a,c) and orange (b) regions inside the Bloch sphere, which form a plane.
    The purple regions represent initial states that can lead to stable points, while the orange regions represent initial states guaranteed to lead to breakdown.
    }
    \label{fig:bitflip-fidelity}
\end{figure*}

\begin{figure*}
     \centering
        \subfigure[]{\includegraphics[width=.45\textwidth]{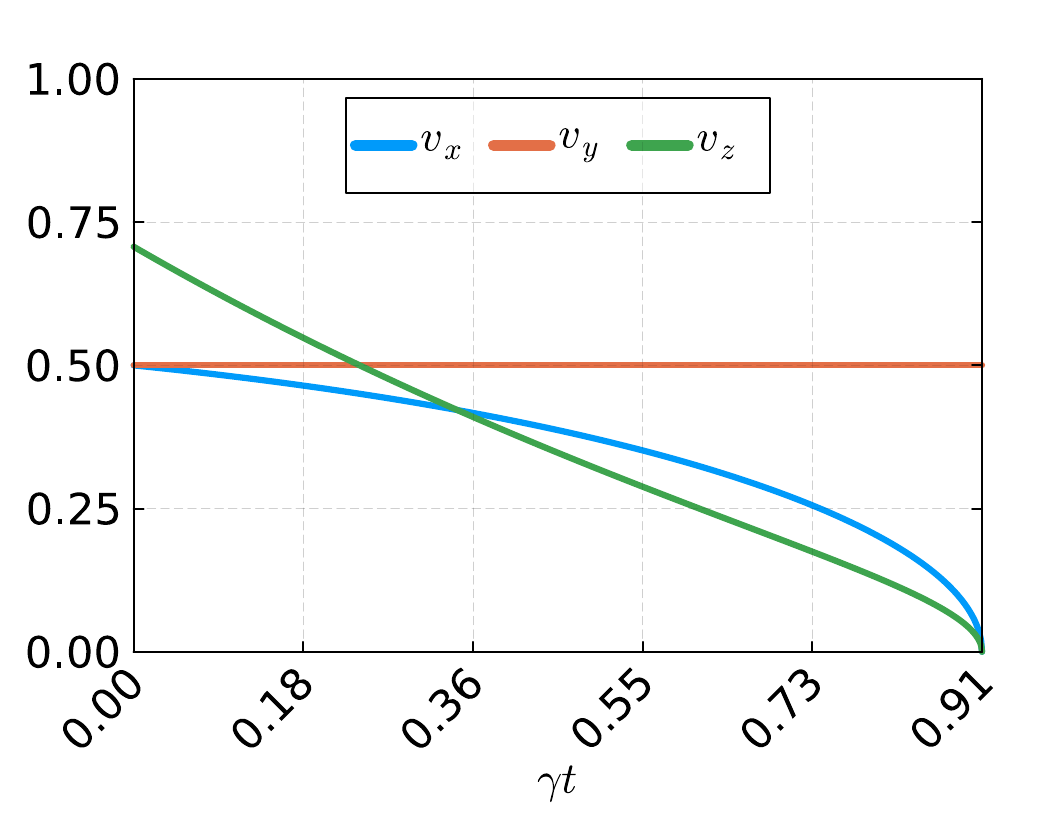}
        \label{fig:fidelity_bloch_vector}}
        \subfigure[]{\includegraphics[width=.45\textwidth]{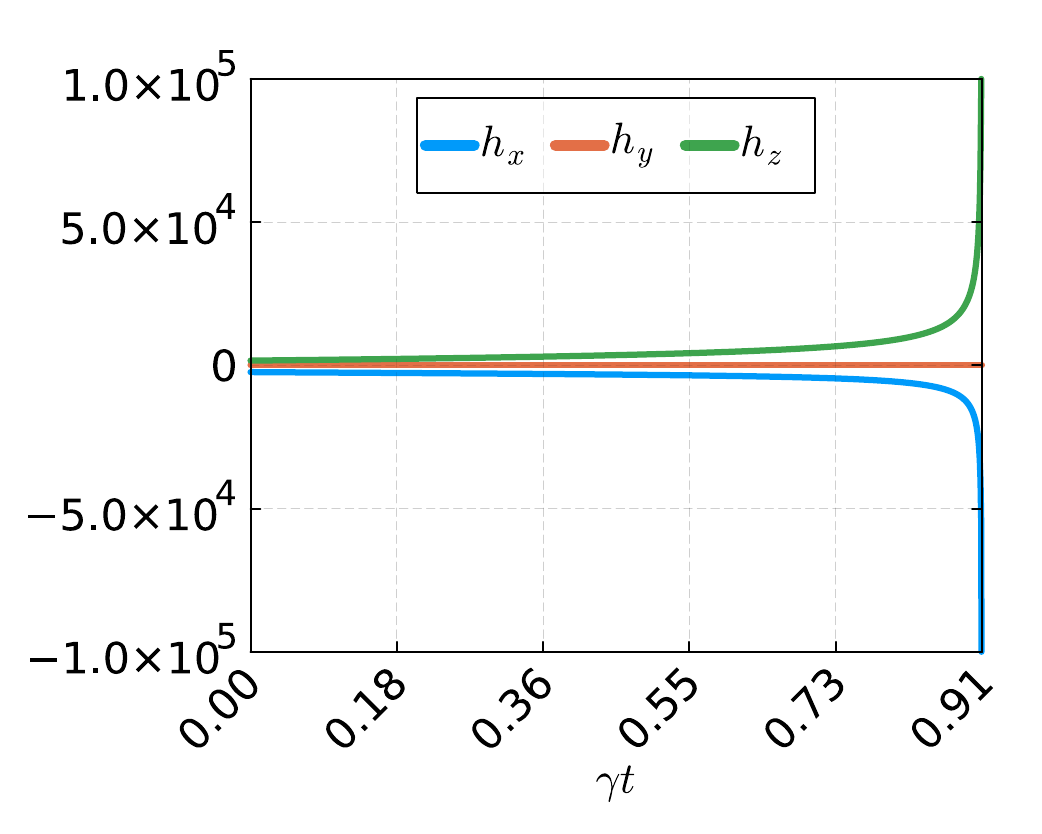}
        \label{fig:fidelity_hamiltonian}}
     \caption{Time dependence of the Bloch vector components (a) and the control Hamiltonian components (b) for preserving the fidelity at a value of $F^2=0.75$ with respect to $\bw= (0,1,0)$ subject to the bit-flip channel, using only the necessary component of the $f$-preserving Hamiltonian [\cref{eq:basic_control}]. The initial state is $(1/2,1/2,1/\sqrt{2})$. Note that the Hamiltonian diverges near the breakdown point ($v_x=v_z=0$).}
    \label{fig:t-dep-fidelity}
\end{figure*}

The stable point locus is the $z$-axis for dephasing and the $x$-axis for the bit flip channel, as in the discussion on coherence preservation in \cref{sec:coherence}.
\cref{fig:bitflip-fidelity} illustrates that the initial conditions decide whether the state can reach a stable point.
Contrasting \cref{fig:bitflip_initial_stable,fig:bitflip_initial_unstable} in particular, shows that it is possible to reach a stable point if an appropriate Hamiltonian is chosen, provided the initial fidelity is sufficiently low.
The set of all points that can lead to stable points forms a cylinder oriented along $\bw$ minus cones with axis $\bw$ and edge along the dephasing axis.
In comparison, \cref{fig:bitflip_initial_stable_specific_ham} shows the smaller set of points that terminate at a stable point for the same initial fidelity as in \cref{fig:bitflip_initial_stable} when we are constrained to apply the fixed-$\mathbf{p}$ Hamiltonian described above.

\cref{fig:t-dep-fidelity} shows the time-dependence of the Bloch vector components and the control Hamiltonian components. Similarly to \cref{fig:t-dep-coherence}, the $h_x$ and $h_y$ components diverge at the breakdown time, where $v_x=v_z=0$.

Continuing our analysis with a fixed-$\mathbf{p}$ Hamiltonian, the dissipator can be written as 
\beq
R\bv+\bc = -2\gamma [\bv-(\bv\cdot\hat{\mathbf{d}})]\hat{\mathbf{d}},
\eeq 
where $\hat{\mathbf{d}}$ represents the unit vector in the $v_z$ and $v_x$ direction for the dephasing and bit-flip channels, respectively.
Then \cref{eq:fidelity-control} simplifies to
\begin{align}
\alpha_p\dot{\alpha_p} = -2\gamma [\|\bv\|^2-(\bv\cdot\hat{\mathbf{d}})^2] = \dot{P}_D(\bv) ,
\end{align}
which is non-positive by the Cauchy-Schwartz inequality.
As a consequence of \cref{th:general_trajectory}, a trajectory with non-increasing purity can be realized.
Since $\alpha_w$ remains constant to preserve fidelity, $\alpha_p$ should be the non-increasing component.

Recall that $\bv = \alpha_w \bw + \alpha_p \mathbf{p}$, $\mathbf{p}\cdot\bw=0$, and $\|\bw\|=\|\mathbf{p}\|=1$; setting $\mathbf{p}\cdot \hat{\mathbf{d}}=\cos{\theta_p}$ and $\mathbf{w}\cdot \hat{\mathbf{d}}=\cos{\theta_w}$, we have:
\begin{align}
\label{eq:dephasing-fidelity-eq}
\alpha_p \dot{\alpha_p} &= -2\gamma (\alpha_p^2 \sin^2{\theta_p} + \alpha_w^2 \sin^2{\theta_w} \notag \\
	&\qquad - 2\alpha_w\alpha_p\cos{\theta_p}\cos{\theta_w} )\ .
\end{align}
This differential equation can be solved analytically for $\alpha_p(t)$ based on initial conditions. 
If $\a_w$ is $0$, the system decays to a stable point, leading to an infinite breakdown time.
Otherwise, if $\theta_p=0$ then $\theta_w=\pi/2$ and $\alpha_p(t)^2-\alpha_p(0)^2=-4\g \a_w t$, which gives the breakdown time  $t_b=\a_p(0)^2/(4\g \a_w)$.
The analysis for $\theta_p\neq 0$, including cases with finite and infinite $t_b$, is provided in \cref{app:dephasing-fidelity-breakdown}.

Note that the infinite $t_b$ result is in contrast to the no-go coherence-preservation result we obtained for the dephasing channel: under no circumstance is it possible to indefinitely preserve the coherence of a state, except for the trivial coherence value of $0$. But the fidelity can be preserved indefinitely for an appropriate choice of $\bw$ and initial fidelity value.

\subsubsection{Depolarization}
The set of stable points contains only one point: the origin.
The state can reach the origin given that we start with $\a_w=0$, which leads to an infinite breakdown time. 
Otherwise, repeating the same analysis as in the previous section, setting $R\bv+\bc = -4\gamma \bv/3$  yields
\begin{subequations}
\begin{align}
\alpha_p \dot{\alpha_p} &= -\frac{4\gamma}{3} (\alpha_p^2 + \alpha_w^2)\ .
\end{align}
\end{subequations}
This gives
\begin{equation}
\alpha_p^2(t) = e^{-8\gamma t/3}(\alpha_p^2(0) + \alpha_w^2) - \alpha_w^2\ .
\end{equation}
Setting $\alpha_p(t_b) = 0$ gives the breakdown time
\begin{equation}
t_b=\frac{3}{8\gamma} \ln {\frac{\alpha_p^2(0) + \alpha_w^2}{\alpha_w^2}} = \frac{3}{8\gamma} \ln {\frac{\|\bv(0)\|^2}{(\bv\cdot\mathbf{w})^2}}\ .
\end{equation}
The behavior matches that of coherence preservation: $t_b$ is independent of the control Hamiltonian used. However, not all initial states will lead to a finite $t_b$.

\subsubsection{Relaxation}

\begin{figure}
\centering
       \subfigure[\ Initial fidelity $F^2=0.75$]{\includegraphics[width=.32\textwidth]{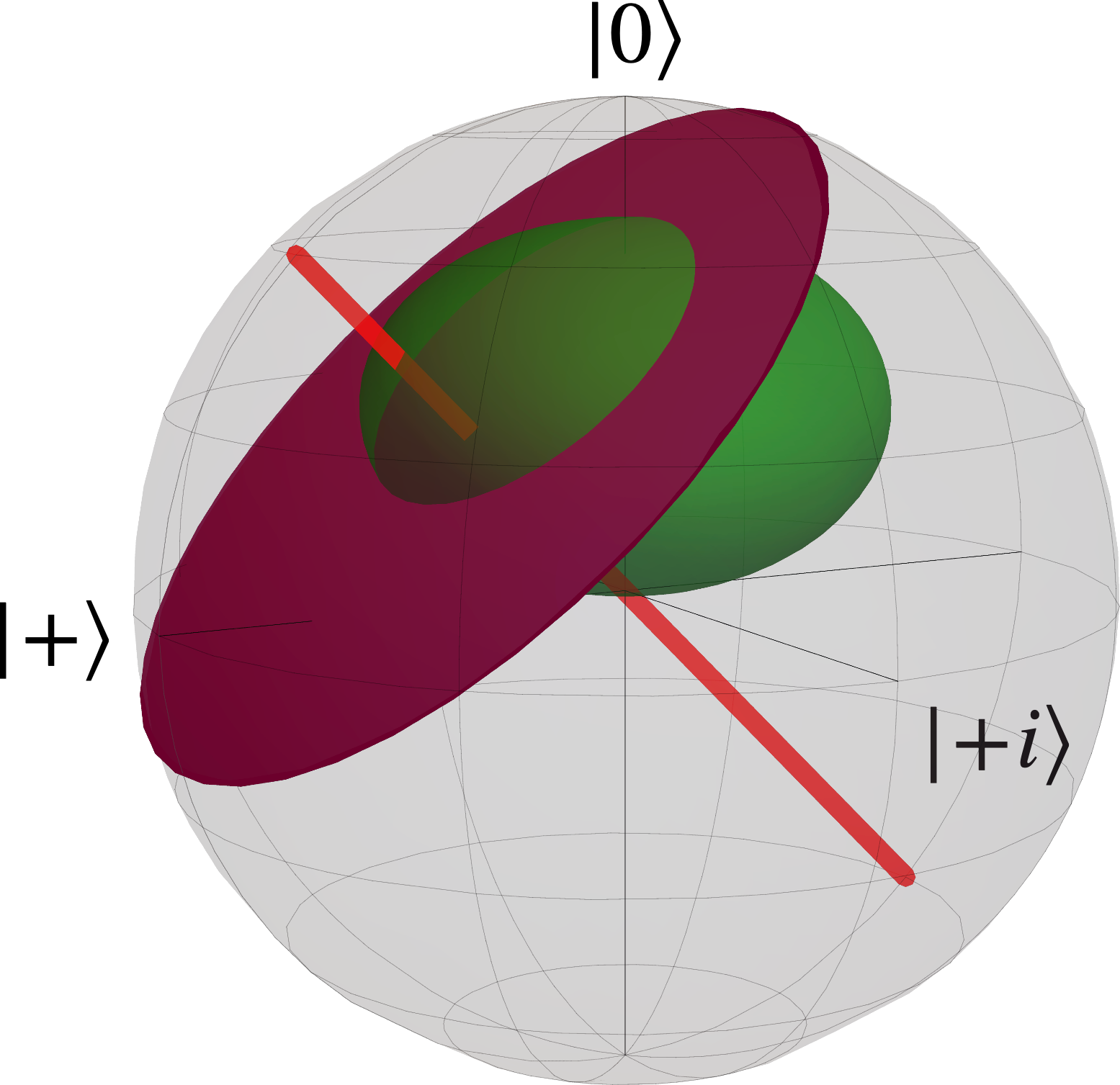}
        \label{fig:relaxation_initial_stable_fidelity}}
        \subfigure[Initial fidelity $F^2=0.9$]{\includegraphics[width=.32\textwidth]{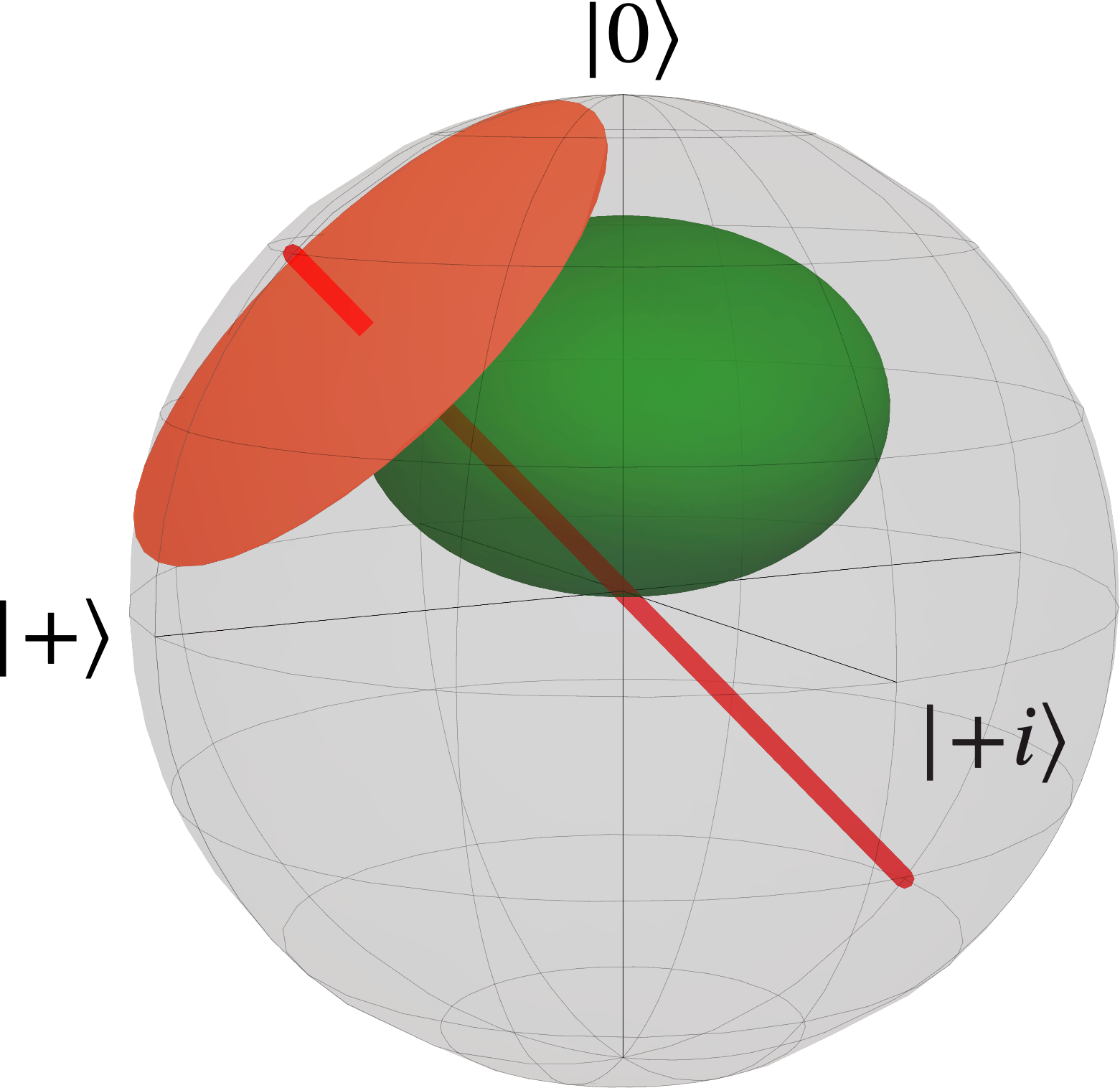}
        \label{fig:relaxation_initial_unstable_fidelity}}
    \caption{Initial state stability landscape for preserving the fidelity with $\bw = (1/\sqrt{2},0,1/\sqrt{2})$, for two different initial fidelity values, subject to the relaxation channel with $\beta \Delta=2$.
    The red (green) region is the set of breakdown (stable) points.
    The level sets of constant fidelity are the purple (a) and orange (b) regions inside the Bloch sphere, which form a plane.
    The purple regions represent initial states that can lead to stable points, while the orange regions represent initial states guaranteed to lead to breakdown.
    }
    \label{fig:relaxation-fidelity}
\end{figure}

The set of stable points forms an ellipsoid. Following the same notation as in \cref{sec:relaxation-preservation}, the dissipative part can be written as
$R\bv +\bc = -\frac{\gamma}{2} \bv + \frac{\gamma}{2} (4a-\bv\cdot\hat{\mathbf{k}})\mathbf{k}$, where $\hat{\mathbf{k}}$ is the unit vector in $v_z$ direction. Then \cref{eq:fidelity-control} simplifies to
\begin{align}
\alpha_p\dot{\alpha_p}  = -\frac{\gamma}{2} (\|\bv\|^2-4a (\bv\cdot\hat{\mathbf{k}}) + (\bv\cdot\hat{\mathbf{k}})^2)\ .
\end{align}
Since $\|\bv\|^2 = \alpha_p^2+\alpha_w^2$ and $\bv\cdot \hat{\mathbf{k}} = \alpha_p (\mathbf{p}\cdot\hat{\mathbf{k}}) + \alpha_w (\mathbf{w}\cdot\hat{\mathbf{k}})$ where $\mathbf{p}\cdot\hat{\mathbf{k}}$ and $\mathbf{w}\cdot\hat{\mathbf{k}}$ are constant, the resulting differential equation is also solvable analytically for a given initial condition. Next, we discuss how varying initial conditions lead to different dynamics.

\cref{fig:relaxation-fidelity} illustrates how the initial state determines whether a stable point can be reached. If the fidelity level set intersects the region of stable points, as in \cref{fig:relaxation_initial_stable_fidelity}, then stable points are reachable from any initial state on that level set (except when the initial state is collinear with $\bw$).
This unique control landscape allows us to preserve the fidelity for almost all states on a level set, unlike any of our previous examples.
Surprisingly, they are reachable using the simple fixed-$\mathbf{p}$ Hamiltonian.
Contrast this with the case of the dephasing channel, where using the fixed-$\mathbf{p}$ Hamiltonian limited the set of reachable stable points.
For points starting inside the ellipsoid, the purity increases (i.e., $\a_p$ increases) until the ellipsoid is reached. Similarly, for points starting outside the ellipsoid, the purity decreases (i.e., $\a_p$ decreases), and eventually, the state ends up on the ellipsoid.

In contrast, when the level set does not intersect the stability ellipsoid, all initial states necessarily reach a breakdown point, as illustrated in \cref{fig:relaxation_initial_unstable_fidelity}.

\section{Conclusions and Outlook}
\label{sec:conclusion}

We have presented a tracking control framework for using smoothly time-varying Hamiltonians to address quantum property preservation (QPP): the preservation of key target properties (e.g., coherence and fidelity) of open quantum systems. QPP can be interpreted as providing operational stability to a quantum system, a softer requirement than preserving the entire state over time. The main advantages of our tracking control approach over standard quantum error avoidance, correction, or suppression are that it avoids any encoding overhead, does not require any measurement-based feedback, and is applicable even subject to Markovian environments, where traditional Hamiltonian control approaches such as dynamical decoupling fail.

We found necessary conditions for control Hamiltonians to preserve a given target property $f$ over time, and provided analytic (non-unique) solutions for these Hamiltonians. 
We showed that independent of the initial state, the dissipator delineates certain regions in the state space where stability is guaranteed to persist indefinitely (stable points). Similarly, given a dissipator, the target property $f$ generates instability regions in the state space (breakdown points). Breakdown points are associated with a breakdown time at which the Hamiltonian diverges and after which the preservation of $f$ ceases. The goal of property preservation can, therefore, be viewed as preventing the state trajectory from passing through any breakdown point. However, the set of possible $f$-preserving trajectories is also constrained; it is characterized by whether the control Hamiltonian can counteract the dissipator action and result in unitary dynamics.

As we have defined it here, one obvious limitation of the tracking control approach to QPP is that it requires full knowledge of the state $\rho$. For an unknown state, this  would, in principle, necessitate full quantum state tomography, a computationally demanding requirement even when implemented using methods such as compressed sensing~\cite{Gross:2010aa}. An important potential avenue for simplification is the use of shadow tomography, which requires exponentially fewer measurements than full state tomography and provides partial information about the quantum state~\cite{Aaronson-shadow-tomography, Huang:2020wo}. Future work will address whether the tracking control problem can be reformulated using such partial state information.

An additional potential future simplification is related to the fact that we have presented a time-local formulation of the target preservation control problem. A natural reformulation in terms of minimizing the integrated deviation of the target property from its initial value is an interesting question. It may even sidestep the issue of the breakdown singularities of the tracking control method.

Other avenues for future research include establishing connections to quantum error avoidance and correction by allowing for the use of ancilla qubits. While this direction is orthogonal to the spirit of the quantum property preservation framework outlined in the Introduction, it is nevertheless potentially promising as an avenue for reducing the resource requirements typically assumed in quantum error correction.
For example, consider developing a tracking control framework that uses partial state information to implement a codespace-preserving Hamiltonian for stabilizer codes. In this context, we may choose the target function to be the fidelity of the input state with the $+1$ eigenspace of the stabilizer group and seek a control Hamiltonian that ensures the fidelity remains high until the next error correction cycle.

\acknowledgments 
This research was supported by the ARO MURI grant W911NF-22-S-0007.
KS acknowledges support from the Graduate Student Fellowship awarded by USC Viterbi School of Engineering.
We thank Parth Darekar for many helpful discussions during the early stages of this work.

\appendix

\section{Coherence vector}
\label{app:cohvec}

We closely follow the exposition in \cite{ODE2QME}, which also contains the proofs of all the propositions provided below.

Let $M(d, F)$ denote the vector space of $d\times d$ matrices with coefficients
in $F$, where $F \in \{\mathbb{R}, \mathbb{C}\}$. 
For our purposes it suffices to identify $\mc{B}(\mc{H})$ with $M(d,\mathbb{C})$. 
Elements of $\mathcal{B}[\mathcal{B}(\mc{H})]$, i.e., linear transformations $\mc{E}:\mathcal{B}(\mc{H})\to \mathcal{B}(\mc{H})$, are called superoperators, or maps. 

Given a nice operator basis $\{F_j\}$, we may ``coordinatize'' the  operator $X\in\mc{B}(\mc{H})$ as
\begin{equation}
  \label{eq:coordinatization1}
  X=\sum_{j=0}^{J}\boldsymbol{X}_{j}F_{j}\ ,
\end{equation}
where $J\equiv d^2-1$ and we use the notation $\boldsymbol{X}=\{\boldsymbol{X}_{j}\}$ for the vector of coordinates of the operator $X$ (we interchangeably use the $\vec{X}$ and $\boldsymbol{X}$ notations to denote a vector). I.e., for any $X\in\mc{B}(\mc{H})$: 
\beq
  \label{eq:coordinatization2}
  \boldsymbol{X}_i = \<F_i,X\> = \Tr(F_i X)\ .
\eeq

The matrix elements of a \emph{superoperator} $\mc{E}: \mc{B}(\mc{H})\to \mc{B}(\mc{H})$ are given by 
\beq
\boldsymbol{\mathcal{E}}_{ij}= \<F_i,\mathcal{E}(F_j)\> = \Tr[F_{i}\mathcal{E}(F_{j})]\ .
\label{eq:Eij}
\eeq

\begin{mydefinition}
\label{def:Herm-supop}
$\mathcal{E}$ is \emph{Hermiticity-preserving} iff:
\beq
[\mathcal{E}\left(A^{\dagger}\right)]^\dag = \mathcal{E}(A)\ \quad \forall A\in \mc{B}(\mc{H})\ .
\label{eq:Herm-pres}
\eeq
$\mc{E}$ is \emph{Hermitian} iff 
\beq
\mathcal{E}^\dag(A) = \mathcal{E}(A)\ \quad \forall A\in \mc{B}(\mc{H})\ .
\eeq
\end{mydefinition}

\begin{myproposition}
\label{prop:0}
The operator $X\in \mc{B}(\mc{H})$ is Hermitian iff $\boldsymbol{X}\in \mathbb{R}^{d^{2}}$ in a nice operator basis.
$\mathcal{E}$ is Hermiticity-preserving iff $\boldsymbol{\mathcal{E}}\in M(d^2,\mathbb{R})$.
The Liouvillian is Hermiticity-preserving.
\end{myproposition}

\begin{mycorollary}
\label{cor:1}
The quantum master equation $\dot{\rho} = \mc{L}\rho$ in a nice operator basis
is a real-valued linear ordinary differential equation for the vector 
$\boldsymbol{\rho}$ whose coordinates are $\boldsymbol{\rho}_j = \Tr(\rho F_j)$:
\beq
\dot{\boldsymbol{\rho}}=\boldsymbol{\mathcal{L}}\boldsymbol{\rho}\ , \quad \boldsymbol{\rho}\in \mathbb{R}^{d^2}\ , \quad \boldsymbol{\mathcal{L}}\in M(d^2,\mathbb{R})\ .
\label{eq:rhobold}
\eeq
\end{mycorollary}

\begin{proof}
This follows directly from \cref{prop:0} with $X=\rho$ and $\mc{E}=\mc{L}$. 
\end{proof}

By \cref{cor:1}, we can expand the density matrix in the nice operator basis as:
\beq
\rho = \mathbf{F}'\cdot \boldsymbol{\rho}  = \frac{1}{d}I + \mathbf{F}\cdot\bv\ , 
\label{eq:333}
\eeq
where $\mathbf{F}' = (F_0, F_1,\dots, F_J)$, $\boldsymbol{\rho} = (1/\sqrt{d},v_1,\dots,v_J)^T$, $\mathbf{F} = (F_1,\dots, F_J)$ collects the traceless basis-operators into a vector, and the corresponding coordinate vector $\bv = (v_1,\dots,v_J)^T \in \mathbb{R}^{J}$ is called the \emph{coherence vector}. We have:
\beq
v_j = \Tr(\rho F_j) = \boldsymbol{\rho}_j \ .
\label{eq:v_j}
\eeq

The space of coherence vectors for a $d$-dimensional quantum system (qudit) is a convex set $\mc{M}^{(d)}$ that is topologically equivalent to a sphere. 
For $d=2$ (a qubit), $\mc{M}^{(2)}$ is a sphere. For $d>2$, $\mc{M}^{(d)}$ is still a convex set but is no longer a sphere, nor is it a polytope. Instead, its faces are copies of $\mc{M}^{(d')}$ with $d'<d$. The boundary of $\mc{M}^{(d)}$ contains all states of less than maximal rank. The set $\mc{M}^{(d)}$ lies inside a sphere of radius $\sqrt{(d-1)/(2d)}$ and contains a maximal sphere of radius $\sqrt{1/[2d(d-1)]}$~\cite{bengtsson2006geometry}. 

\begin{myproposition}
\label{prop:4}
The coherence vector satisfies \cref{eq:le}. 
Moreover, the decomposition of $\mc{L}$ as $\mc{L} = \mc{L}_H+\mc{L}_D$ with $\mc{L}_H$ and $\mc{L}_D$ given by \cref{eq:L_H,eq:L_a}
induces the decomposition $\mc{L}_H(\rho) \leadsto Q\bv$, $\mc{L}_D(\rho) \leadsto R\bv + \bc$, $\bc\in\mathbb{R}^J$, $R,Q\in M(J,\mathbb{R})$, and $Q=-Q^T$ (anti-symmetric). Specifically, $R$, $Q$, and $\bc$ have elements in the nice operator basis $\{F_j\}$ given by:
\bes
\label{eq:8.4.2}
\begin{align}
\label{eq:Qjk}
Q_{ij} &\equiv \Tr[F_{i}\mathcal{L}_H(F_{j})] = Q_{ij}^*\\
\label{eq:Rjk}
R_{ij} &\equiv \Tr[F_{i}\mathcal{L}_D(F_{j})]= R_{ij}^* \\
c_j &= \frac{1}{d}\Tr[F_{j}\mathcal{L}_D(I)] = c_j^*\ .
\end{align}
\ees
\end{myproposition}

\begin{mydefinition}
The dissipator $\mathcal{L}_D$ is unital if $\mathcal{L}_D(I)= 0$.
\end{mydefinition}

\begin{myproposition}
\label{prop:5}
$\bc=\mathbf{0}$ iff $\mathcal{L}_D$ is unital.
\end{myproposition}

\begin{myproposition}
\label{prop:Rsymm-LHerm}
The following are conditions under which $R$ is symmetric:
\begin{enumerate}
\item $R$ is symmetric in the single qubit ($d=2$) case.
  \item $R$ is symmetric and $\bc=\mathbf{0}$ iff the Lindblad operators $\mc{L}_{\a}$ are Hermitian.
  \end{enumerate}
\end{myproposition}

\section{The Bloch vector}
\label{app:Bloch}

We write $\s^0=I$. The set of Pauli matrices $\{\s^0,\s^x, \s^y, \s^z\}$ forms a basis for the space of $2\times 2$ complex matrices. It is not a proper nice operator basis since it is not normalized, and this gives rise to the difference between the Bloch vector and the coherence vector for a qubit. In more detail, any qubit density matrix can be expanded in the Pauli basis as 
\beq 
\r = \frac{1}{2} \left( I +\sum_{i=x,y,z} v_i\s^i \right) = \frac{1}{2} \left( I + \bv \cdot \boldsymbol{\s} \right)\ ,
\label{eq:Bloch-vec}
\eeq 
where $\bv = (v_x, v_y, v_z)\in S^2$ is called the \emph{Bloch vector}, $S^2$ is the unit sphere in $\mathbb{R}^3$, and $\boldsymbol{\s} = (\s^x, \s^y, \s^z)$ is the vector of Pauli matrices. Since the Pauli matrices are traceless and $\Tr I=2$, the expansion ensures that $\Tr\r=1$. Since the Pauli matrices are Hermitian, this also ensures that $\r$ is Hermitian. In terms of the Bloch vector, $\r$ becomes 
\beq 
\r = \frac{1}{2} \left( \begin{array}{cc} 1+v_z & v_x - i v_y \\ v_x + i v_y & 1 - v_z \end{array} \right)\ . 
\eeq 

Positivity is the statement that the eigenvalues $\lambda_\pm$ are non-negative:
\beq
\left| \r - \lambda I \right| = 0 \Rightarrow \lambda^2 - (\Tr\r)\lambda + |\r | = 0\ ,
\eeq
where $|\r |$ denotes the determinant of $\r$. I.e., using $\Tr\r=1$:
\beq
\lambda_\pm = \frac{1}{2} (1\pm \sqrt{1-4|\r |}) \geq 0\ .
\label{eq:lambdapm}
\eeq
Positivity can now be made explicit by noting that
\beq
|\r | = \frac{1}{4} \left(1-v_z^2-(v_x^2+v_y^2)\right) = \frac{1}{4} \left(1-\|\bv\|^2\right)\ ,
\eeq
so that by \cref{eq:lambdapm}: $\lambda_\pm = \frac{1}{2} ( 1 \pm \left\| \bv \right\| )$.
Thus, positivity is equivalent to: 
\beq
{\left\| \bv \right\| \leq 1}\ ,
\label{eq:v-leq1}
\eeq 
which means that valid qubit states are represented by Bloch vectors that lie on the surface or in the interior of the unit sphere. 

For a qubit, expanding the Hamiltonian in the Pauli basis as $H = h_0 I+ \bh \cdot\boldsymbol{\s}$ yields:
\bes
\begin{align}
\mc{L}_H(\rho) & = -i[H,\rho] = -\frac{i}{2}\sum_{i\in\{x,y,z\}}h_i [\s^i,\bv\cdot\boldsymbol{\s}] \\&= -\frac{i}{2}\sum_{i,j\in\{x,y,z\}}h_i v_j [\s^i,\s^j] \\
&= \sum_{i,j,k\in\{x,y,z\}}\varepsilon_{ijk} h_i v_j\s^k = (\bh \times\bv)\cdot\boldsymbol{\s} \ .
\end{align}
\ees
Since $\dot{\r} = \frac{1}{2}(\dot{\bv}\cdot\mathbf{\sigma})$, we find $\dot{v}\cdot\boldsymbol{\s} = 2(\bh \times\bv)\cdot\boldsymbol{\s}$, so that:
\beq
Q\bv = \dot{v} = 2(\bh \times\bv)\ .
\label{eq:Blocheq}
\eeq

\cref{tab:channels} provides the dissipators for various common noise channels we study in this work. These are calculated using \cref{eq:8.4.2}, after adjusting for the difference in normalization.

\begin{table}[h]
\centering
\begin{tabular}{ |c|c|}
    \hline
    \textbf{Noise channel} & \textbf{Dissipator} $D=(R,\bc)$  \\
    \hline \rule{0pt}{2em}
    Dephasing ($Z$) & $R = \mathrm{diag}(-2\gamma, -2\gamma, 0)$, $\mathbf{c} = \mathbf{0}$  \\
    \hline \rule{0pt}{2em}
    Bit-flip ($X$) & $R = \mathrm{diag}(0, -2\gamma, -2\gamma)$, $\mathbf{c} = \mathbf{0}$ \\
    \hline \rule{0pt}{2em}
    Bit-phase-flip ($Y$) & $R = \mathrm{diag}(-2\gamma, 0, -2\gamma)$, $\mathbf{c} = \mathbf{0}$ \\
    \hline \rule{0pt}{2em}
    Depolarizing & $R = -\frac{4}{3}\gamma I$, $\mathbf{c} = \mathbf{0}$  \\
    \hline \rule{0pt}{2em}
    \multirow{2}{*}{\parbox{2.5cm}{Relaxation at temperature $T$}} &  \\
    & $R = \left(-\frac{\gamma}{2}, -\frac{\gamma}{2}, -\gamma\right)$, $\mathbf{c} = (0, 0, 2\gamma a)$ \\
    & \\
    \hline \rule{0pt}{2em}
    \multirow{3}{*}{\parbox{2.8cm}{Relaxation at temperature $T$ ($\g_1$) + Dephasing ($\g_d$)}} &  \\
    & $R = \left(-\g_2, -\g_2, -\g_1\right)$, $\mathbf{c} = (0, 0, 2\g_1 a)$\\ 
    & \\
    \hline
\end{tabular}
\caption{Dissipators for different noise channels. $\g_2 \equiv 2\g_d + \frac{\g_1}{2}$, $a\equiv[1+\exp(-\b \D)]^{-1}-1/2$ where $\D$ is the qubit energy gap and $\b=1/T$ is the inverse temperature.}
\label{tab:channels}
\end{table}

\section{Uhlmann fidelity between Bloch vectors}
\label{app:Bloch-Uhlmann}

The Uhlmann fidelity between two general density matrices is $F(\r,\chi) = \|\sqrt{\r}\sqrt{\chi}\|_1 = \Tr(\sqrt{\sqrt{\chi}\rho\sqrt{\chi}})$. 

\begin{myproposition}
In the case of two qubits, the Uhlmann fidelity expressed in terms of the respective Bloch vectors $\bv$ and $\bw$ of $\r$ and $\chi$ is:
\beq
F^2(\bv,\bw) = \frac{1}{2}\left(1+\bv\cdot\bw + [(1-\|\bv\|^2)(1-\|\bw\|^2)]^{1/2}\right)\ .
\label{eq:Fid-qubits-Bloch}
\eeq
\end{myproposition}

\begin{proof}
Following Ref.~\cite{Hubner:1992aa}, let $A$ be any $2\times 2$ positive semidefinite matrix, with eigenvalues $\lambda_1,\lambda_2\ge 0$. Then $\sqrt{A}$ has eigenvalues $\sqrt{\lambda_1},\sqrt{\lambda_2}\ge 0$, so that $\Tr\sqrt{A} = \sqrt{\lambda_1}+\sqrt{\lambda_2}$ and $\det(A) = \lambda_1\lambda_2$. Thus:
\beq
(\Tr\sqrt{A})^2 = \lambda_1 + \lambda_2 +2 \sqrt{\lambda_1\lambda_2} = \Tr A + 2 \sqrt{\det(A)}\ .
\label{eq:2-q-F-step}
\eeq
Now let $A = \sqrt{\r} \chi \sqrt{\r}$, which is a positive semidefinite matrix since it is the product of positive semidefinite matrices. Then, using the cyclic property of the trace and the product formula for determinants, we have:
\begin{subequations}
\label{eq:2-q-F}
\begin{align}
F^2(\r,\s) &= (\Tr\sqrt{A})^2 = \Tr(\sqrt{\r} \chi \sqrt{\r}) + 2\sqrt{\det(\sqrt{\r} \chi \sqrt{\r})} \\
&= \Tr(\r \chi) + 2\sqrt{\det(\r\chi)}\ .
\end{align}
\end{subequations}
Next, write $\rho = \frac{1}{2}(I + \bv\cdot \mathbf{\sigma})$ and $\chi = \frac{1}{2}(I + \bw\cdot \mathbf{\sigma})$.
    Thus, 
\begin{equation}
        \Tr(\rho\chi) = \frac{1}{2}(1 + \bv\cdot \bw)
\end{equation}
    Using $\lambda_\pm = \frac{1}{2} ( 1 \pm \left\| \bv \right\| )$ for the eigenvalues of $\r$ and the fact that the determinant is the product of the eigenvalues,
\begin{equation}
        \det(\rho\chi) = \det(\rho)\det(\chi)
        = \frac{1}{16}(1-\norm{\bv}^2)(1-\norm{\bw}^2)\ .
\end{equation}
    Thus, $\sqrt{\det(\rho\chi)} = \frac{1}{4}[(1-\norm{v}^2)(1-\norm{w}^2)]^{1/2}$, and \cref{eq:Fid-qubits-Bloch} follows from \cref{eq:2-q-F}.
    \end{proof}

\section{Derivation of control Hamiltonians for coherence preservation}
\label{app:example_details}

\subsection{Preserving the coherence magnitude for the bit-flip channel}
\label{app:bitflip-sol}

For later reference, we note that for $f=v_x^2+v_y^2$:
\begin{subequations}
\label{eq:coh-refeqs}
\begin{align}
& \grad f = (2v_x,2v_y,0) \\
& \grad f\cross \bv = (2v_y v_z , -2v_x v_z , 0)\\ 
& \| \grad f\cross \bv\|^2 = 4v_z^2 f_0^2\ ,
\end{align}
\end{subequations}
where $f_0=v_x^2(0)+v_y^2(0)$ is the coherence magnitude that is being preserved.

We can solve the dynamics equation, for which we first need to compute $\bh \cross \bv$ using $\bh  = \alpha_1 \bv + \alpha_2 \grad f + \alpha_3 \grad f \cross \bv$ [\cref{eq:h-decomp-basis}]. In doing so, the term $\alpha_1\bv$ drops out, and we use $\alpha_3 = \frac{1}{2} \frac{\grad f \cdot (R\bv + \bc)}{\|\grad f \cross \bv\|^2}$ [\cref{eq:alpha_3}] with $R=\mathrm{diag}(0,-2\gamma,-2\gamma)$, and $\bc=\mathbf{0}$.
We then have 
\begin{subequations}
\begin{align}
&R\bv+\bc = (0,-2\g v_y,-2\g v_z)\\
&\grad f \cdot (R\bv + \bc)=-4 \gamma  v_y^2\ . 
\end{align}
\end{subequations}
Thus, using \cref{eq:coh-refeqs}, we find 
\beq
\a_3 = -\frac{1}{2}\frac{\gamma  v_y^2}{f_0 v_z^2}\ ,
\label{eq:alpha3-BFC}
\eeq
and hence:
\beq
\bh  = \left(-\frac{\gamma  v_y^3}{f_0v_z}+ 2\alpha _2 v_x,\frac{\gamma  v_x v_y^2}{f_0v_z}+ 2\alpha_2 v_y, 0\right)\ ,
\label{eq:D2}
\eeq
which yields:
\begin{align}
\bh \cross \bv &= \left(\frac{\gamma  v_x v^2_y}{f_0}+2 \alpha _2 v_y v_z,\frac{\gamma 
   v_y^3}{f_0}-2 \alpha _2 v_x v_z,\right.\notag \\
   &\left.\qquad-\frac{\gamma  v_y^2}{v_z}+2\alpha_2 v_x v_y\right)\ .
   \label{eq:hxv-bitflip}
\end{align}

Then, using \cref{eq:reference_dynamics_eq}, corresponding equations for the Bloch vector components are:
\begin{subequations}
\label{eq:vdot-btflip}
    \begin{align}
\label{eq:vdot-btflip-x}
        \dot{v}_x &= \frac{2 \gamma  v_x v_y^2}{f_0}+4 \alpha _2 v_y v_z\\
\label{eq:vdot-btflip-y}
        \dot{v}_y &= 2 \gamma  v_y
   \left(\frac{v_y^2}{f_0}-1\right)-4 \alpha _2 v_x v_z\\
\label{eq:vdot-btflip-z}
        \dot{v}_z &= -2 \gamma\frac{v_z^2+v_y^2}{v_z}+4\alpha_2 v_x v_y\ .
    \end{align}
\end{subequations}
 
\subsubsection{$\a_2=0$: finite breakdown time}
For $\a_2=0$, solving \cref{eq:vdot-btflip-y} first gives:
\beq
\label{eq:vy-bitflip}
v^2_y(t) = \frac{v^2_y(0)f_0}{{v^2_x(0) e^{4 \gamma  t}+v^2_y(0)}}\ .
\eeq
This lets us solve \cref{eq:vdot-btflip-z}, yielding:
\begin{align}
v^{2}_z(t) &= 
\frac{e^{-4 \gamma  t}}{v_x^2(0)}  \left[ v_x^2(0)v_z^2(0)\notag \phantom{\frac{1}{2}}\right.\\
&\qquad \left. -f_0 v_y^2(0)\ln \left(\frac{v_y^2(0)+ e^{4\g t}v_x^2(0)}{f_0 }\right)\right]\ .
\end{align}
In this trajectory, clearly $v_y(t)>0$ $\forall t\ge 0$. Therefore the Hamiltonian [\cref{eq:D2}] diverges when $v_z(t)=0$; the solution then yields \cref{eq:t_b-bitflip}.

\subsubsection{$\a_2\ne 0$: infinite breakdown time}

Different values of $\a_2$ yield different trajectories.
Instead of analyzing the result of varying $\a_2$, we can choose a trajectory we would like the Bloch vector to take and calculate the required Hamiltonian based on that path.
Such a choice can be advantageous; e.g., as we demonstrate here, we can ensure that the trajectory ends at a stable point, which means that the coherence magnitude is indefinitely preserved.

Due to \cref{th:general_trajectory}, any trajectory can be realized as long as that trajectory's purity is the same as $\partial_u P_D$ up to a positive multiplicative factor.
Consider $\partial_u P_D$ for the bit flip channel:
\beq
\partial_u P_D [\bl(u)] = \bl\cdot (R\bl+\mathbf{c}) = -2\gamma(l_y^2+l_z^2)\ ,
\eeq
so the trajectory purity must be negative along the entire path.
Since the preserved property is the coherence magnitude, the trajectory should lie on a cylindrical level set.
As long as these two properties are satisfied, we can find a Hamiltonian that moves the Bloch vector along that trajectory.

For example, suppose we wish to realize a trajectory that ends with the state on the $x$-axis of the Bloch sphere (e.g., the curved green arrow in \cref{fig:bitflip-coherence}).
Since the trajectory must lie on the level set, which is a cylinder with radius $\sqrt{f_0}$ (recall \cref{sec:LS_BP}), one possible (non-unique) trajectory that satisfies the required conditions is
\bes
\begin{align}
\bl(u) &= (l_x(u),l_y(u),l_z(u))\\
&=(\sqrt{f_0}\sin(\frac{\pi}{2}\theta + \frac{\pi}{2}(1-\theta)u),\notag \\
& \qquad \sqrt{f_0}\cos(\frac{\pi}{2}\theta + \frac{\pi}{2}(1-\theta)u),v_z(0)\sqrt{1-u})\ , 
\end{align}
\ees 
where $u \in [0,1]$ and $\theta = \frac{2}{\pi}\arctan(v_x(0)/v_y(0))$.
If the reparameterization to physical time is given by $\varphi : [0,1] \mapsto [0,t_f]$, the realizable trajectory is given by $\bv(\varphi(u)) = \bl(u)$.

\cref{th:general_trajectory} yields $c(u)$ (where $\varphi(s) = \int_{0}^{s} c(u) du$): the trajectory purity is $P_\bl(u)=\frac{1}{2}[1+\|\bl(u)\|^2] = \frac{1}{2}[1+ f_0 + v_z^2(0)(1-u)]$, which leads to
\beq
    c(u) = \frac{\partial_u P_\bl}{\partial_u P_D} =  \frac{v_z^2(0)}{4\gamma (l_y^2(u)+l_z^2(u))} =  \frac{v_z^2(0)}{4\gamma (v_y^2(t)+v_z^2(t))}\ .
\eeq

To find an appropriate value of $\a_2$, we can use the following relation between $\bl$ and $\bv$:
\begin{subequations}
    \begin{align}
    \partial_u \varphi(u) \cdot \dot{\bv}(t) = \partial_u \bl(u)\\
    \text{or } c(u) \cdot \dot{\bv}(t) = \partial_u \bl(u) \ .
    \label{eq:alpha2-cu}
    \end{align}
\end{subequations}

The r.h.s. of \cref{eq:alpha2-cu} can be computed from the trajectory information:
\begin{align}
    \partial_u l_x(u) &= \frac{\pi}{2}(1-\theta) l_y(u) = \frac{\pi}{2}(1-\theta) v_y(t)\ .
\end{align}
Equating the components of the vectors on both sides of \cref{eq:alpha2-cu} then yields $\a_2$, e.g., via \cref{eq:vdot-btflip-x}:
    \begin{align}
&\frac{v_z^2(0)}{4\gamma (v_y^2(t)+v_z^2(t))} \left(\frac{2 \gamma  v_x(t) v_y^2(t)}{f_0}+4 \alpha _2 v_y(t) v_z(t)\right)  \notag\\ 
    & \qquad = \frac{\pi}{2 }(1-\theta) v_y(t)\ ,
    \end{align}
so that
    \begin{align}
    \alpha_2 &= \frac{\gamma}{2v_z(t)}
\left(\frac{\pi(1-\theta)(v_y^2(t)+v_z^2(t))}{v_z^2(0)}-\frac{ v_x(t) v_y(t)}{f_0}\right)\ .
    \end{align}
Using this result for $\a_2$ along with \cref{eq:alpha3-BFC} for $\a_3$ will move the Bloch vector along the prescribed trajectory while preserving the coherence. Moreover, the trajectory leads to a stable point since the endpoint is on the $x$-axis. The corresponding breakdown time is infinity.

\subsection{Breakdown time for dephasing+relaxation channel in coherence preservation}
\label{app:breakdown-relaxation}

Recall from \cref{sec:deph+rel} that the dissipator is $R=\mathrm{diag}(-\g_2, -\g_2, -\g_1)$ and $\bc=(0,0,2\g_1 a)$, where $1/T_2 = 2\g_d+\g_1/2$ and $a=[1+\exp(-\b \D)]^{-1}-1/2$.  Thus:
\begin{subequations}
\begin{align}
& R\bv+\bc = \left(-\g_2 v_x, -\g_2 v_y,-\g_1 (v_z-2a)\right)\\
& \grad f\cdot(R\bv+\bc) = -2 \g_2 f_0\ .
\end{align}
\end{subequations}
Then, using \cref{eq:basic_control,eq:coh-refeqs} the basic control Hamiltonian ($\a_2=0$) is:
        \begin{align}
            \mathbf{h} &= 
             -\frac{\g_2}{2v_z} (v_y,-v_x,0)\ ,
        \end{align}
and using \cref{eq:reference_dynamics_eq}, the corresponding equations for the Bloch vector components are:
    \begin{subequations}
    \begin{align}
        \dot{v}_x &= 0 \\
        \dot{v}_y &= 0 \\
        \dot{v}_z &= -\frac{\g_2 f_0}{v_z} -\g_1 v_z + 2\g_1 a \\
         &= -\frac{\g_1}{v_z} \left( \left(v_z-a\right)^2 + \left(\frac{\g_2}{\g_1}f_0-a^2\right) \right) \ .
        \label{eq:vzdot16}
    \end{align}
    \end{subequations}

As argued in \cref{sec:deph+rel}, for points with $f_0 > a^2\frac{\g_1}{\g_2}$, all trajectories necessarily end up at a breakdown point. The constant term in \cref{eq:vzdot16} is positive for such points.
Substituting $\Omega^2 \equiv \frac{\g_2}{\g_1}f_0-a^2$ and integrating, we obtain:
\begin{subequations}
\begin{align}
     -\g_1 \int_0^t ds &= \int_{v_z(0)}^{v_z(t)} \frac{v_z dv_z}{\left(v_z-a\right)^2 + \Omega^2}\\
     &= \frac12  \log ( \left( v_z - a \right)^2 + \Omega^2) \notag \\
     & \qquad + \left. \frac{a}{\Omega} \arctan{ \frac{v_z-a}{\Omega} } \right|_{v_z(0)}^{v_z(t)}\ .
\end{align}
\end{subequations}
Substituting $v_z(t_b)=0$ finally gives the breakdown time as in \cref{eq:tb-deph+rel}, and similarly, when $f_0<a^2 (\g_1/\g_2)$ and $v_z<a$ we obtain \cref{eq:tb-deph+rel2}.

\section{Derivation of control Hamiltonians for Uhlmann fidelity preservation}
\label{sec:pres-Uhlmann}

\subsection{$\alpha_3$ for various noise channels}

Substituting the dissipator $D=(R,\bc)$ for various noise channels into \cref{eq:uhlman_a3} yields \cref{tab:channel_substitutions}.

\begin{table}[h]
\centering
\begin{tabular}{ |c|c|c|}
    \hline
    \textbf{Noise channel} & $\alpha_3$ \\
    \hline \rule{0pt}{2em}
    Dephasing & $-2\gamma \frac{(\mathbf{w} - k_0 \mathbf{v}) \cdot (v_x, v_y, 0)}{\|\mathbf{w} \times \mathbf{v}\|^2}$ \\
    \hline \rule{0pt}{2em}
    Bit-flip &  $-2\gamma \frac{(\mathbf{w} - k_0 \mathbf{v}) \cdot (0, v_y, v_z)}{\|\mathbf{w} \times \mathbf{v}\|^2}$ \\
    \hline \rule{0pt}{2em}
    Depolarizing & $-\frac{4\gamma}{3} \frac{(\mathbf{w} - k_0 \mathbf{v}) \cdot \mathbf{v}}{\|\mathbf{w} \times \mathbf{v}\|^2}$ \\
    \hline \rule{0pt}{2em}
    Relaxation & $ -\frac{\gamma}{2} \frac{(\mathbf{w} - k_0 \mathbf{v}) \cdot (v_x, v_y, 2(1 - v_z))}{\|\mathbf{w} \times \mathbf{v}\|^2}$ \\
    \hline
\end{tabular}
\caption{$\alpha_3$ for different noise channels.}
\label{tab:channel_substitutions}
\end{table}

\begin{figure}[ht]
\includegraphics[width=0.36\textwidth]{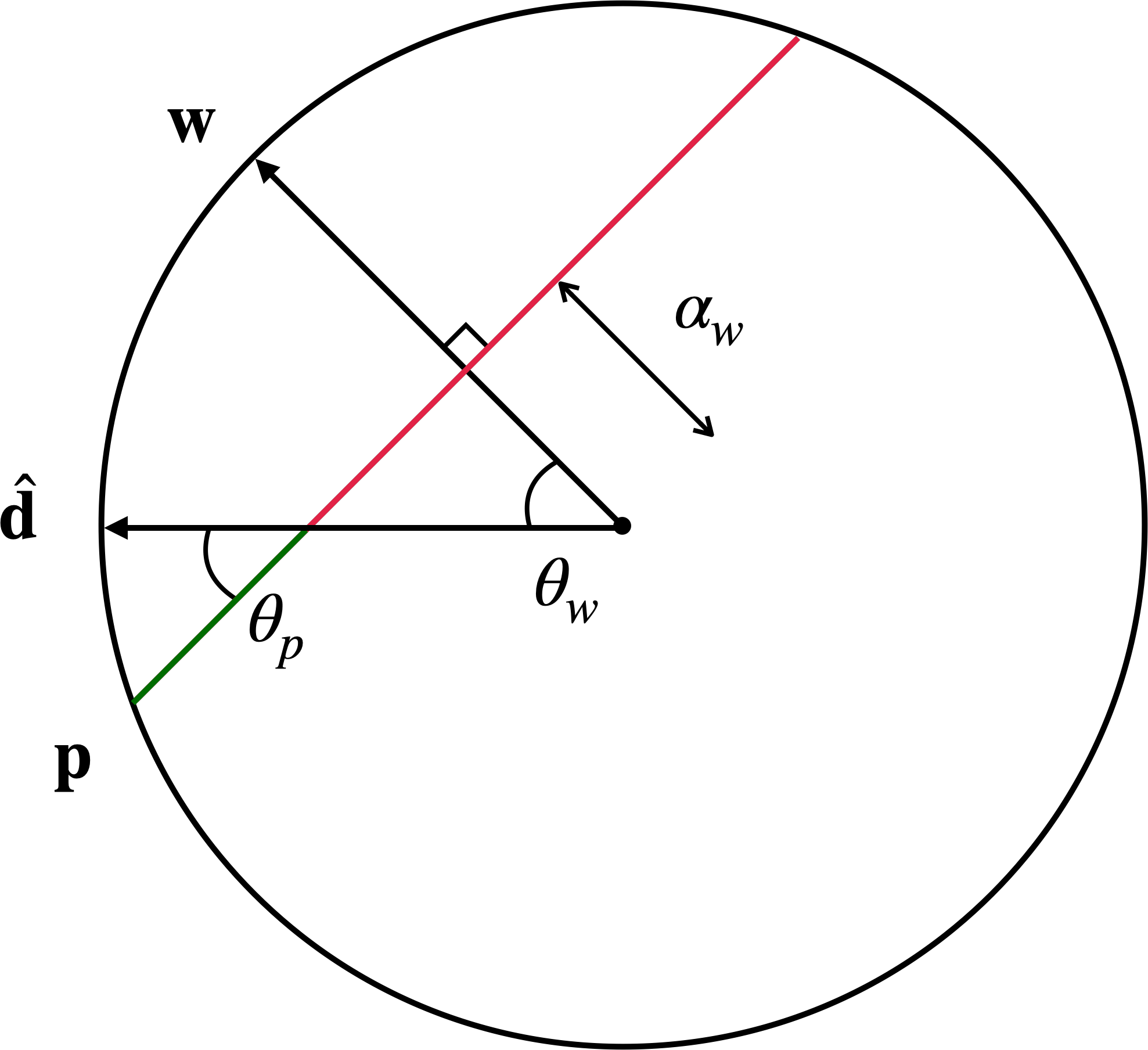}
        \caption{Cross-sectional view of the initial states leading to stable points in \cref{fig:bitflip_initial_stable_specific_ham}.}
	\label{fig:bitflip_initial_stable_specific_ham_detailed}
\end{figure}

\subsection{Preserving Uhlmann fidelity under dephasing/bit-flip}
\label{app:dephasing-fidelity-breakdown}

\cref{fig:bitflip_initial_stable_specific_ham_detailed} provides a simple way to characterize initial states that lead to stable points when using the fixed-$\mathbf{p}$ Hamiltonian:
\begin{enumerate}
\item $\mathbf{p}$, $\mathbf{w}$ and $\hat{\mathbf{d}}$ should be coplanar, i.e. $\mathbf{p}\cdot(\mathbf{w}\cross \hat{\mathbf{d}})=0$.
\item the initial state should be on the correct side of the sphere: $\hat{\mathbf{d}}$ should lie between $\mathbf{p}$ and $\mathbf{w}$, i.e.,  $\theta_p + \theta_w = \pi/2$ since $\mathbf{p}$ and $\mathbf{w}$ are orthogonal.
\item the initial state lies on the green segment in \cref{fig:bitflip_initial_stable_specific_ham_detailed}: $\a_p > \a_w \tan\theta_w $.
\end{enumerate}

All other points have a finite breakdown point, and the corresponding breakdown time $t_b$ can be calculated: rearranging \cref{eq:dephasing-fidelity-eq} gives
\begin{subequations}
\begin{align}
-2\gamma &= \frac{\alpha_p \dot{\a_p}}{\alpha_p^2 \sin^2{\theta_p} + \alpha_w^2 \sin^2{\theta_w} - 2\alpha_w\alpha_p\cos{\theta_p}\cos{\theta_w} }\\
&= \frac{\alpha_p \dot{\a_p}}{\a_p^2 - \a_p \left( \frac{2\a_w \cos{\theta_p}\cos{\theta_w}}{\sin^2{\theta_p}}\right) + \left(\frac{\a_w^2 \sin^2{\theta_w}}{\sin^2{\theta_p}}\right)} \frac{1}{\sin^2{\theta_p}}\\
&= \frac{\alpha_p \dot{\a_p}}{\left(\a_p - c_1 \right)^2 - c_2}\;\frac{1}{\sin^2{\theta_p}}\ ,
\end{align}
\end{subequations}
where
\begin{subequations}
\begin{align*}
c_1 &= \frac{\a_w \cos{\theta_p}\cos{\theta_w}}{\sin^2{\theta_p}}\\
c_2 &= \frac{\a_w^2 \cos{(\theta_p+\theta_w)} \cos{(\theta_p-\theta_w)}}{\sin^4{\theta_p}}\ .
\end{align*}
\end{subequations}
Integrating both sides yields
\begin{equation}
-2\gamma t =  \left. \frac12 \log{[(\a_p-c_1)^2-c_2]} + g(\a_p) \right|_{\a_p(0)}^{\a_p(t)} \frac{1}{\sin^2{\theta_p}}\ ,
\end{equation}
where
\begin{equation}
g(\a_p) = 
     \begin{cases}
      -\frac{c_1}{\sqrt{\abs{c_2}}}\arctan{\frac{c_1-\a_p}{\sqrt{\abs{c_2}}}} &\text{ if}\ c_2 < 0 \\[1em]
      \frac{c_1}{c_1-\a_p} &\text{ if}\ c_2 = 0\\[1em]
      \frac{c_1}{\sqrt{\abs{c_2}}}\arctanh{\frac{c_1-\a_p}{\sqrt{\abs{c_2}}}} &\text{ if}\ c_2 > 0\ .
    \end{cases}
\end{equation} 
The breakdown time can be found by setting $\a_p(t_b) = 0$, which results in
\begin{align}
t_b = \frac1{2\gamma \sin^2{\theta_p}} \left[\log{\left( \frac{(\a_p(0)-c_1)^2-c_2}{c_1^2-c_2} \right)} + g(\a_p(0))-g(0)\right]\ .
\end{align}
Note that $c_1^2-c_2$ is always non-negative.
\bibliography{refs}

\begin{thebibliography}{56}%
\makeatletter
\providecommand \@ifxundefined [1]{%
 \@ifx{#1\undefined}
}%
\providecommand \@ifnum [1]{%
 \ifnum #1\expandafter \@firstoftwo
 \else \expandafter \@secondoftwo
 \fi
}%
\providecommand \@ifx [1]{%
 \ifx #1\expandafter \@firstoftwo
 \else \expandafter \@secondoftwo
 \fi
}%
\providecommand \natexlab [1]{#1}%
\providecommand \enquote  [1]{``#1''}%
\providecommand \bibnamefont  [1]{#1}%
\providecommand \bibfnamefont [1]{#1}%
\providecommand \citenamefont [1]{#1}%
\providecommand \href@noop [0]{\@secondoftwo}%
\providecommand \href [0]{\begingroup \@sanitize@url \@href}%
\providecommand \@href[1]{\@@startlink{#1}\@@href}%
\providecommand \@@href[1]{\endgroup#1\@@endlink}%
\providecommand \@sanitize@url [0]{\catcode `\\12\catcode `\$12\catcode
  `\&12\catcode `\#12\catcode `\^12\catcode `\_12\catcode `\%12\relax}%
\providecommand \@@startlink[1]{}%
\providecommand \@@endlink[0]{}%
\providecommand \url  [0]{\begingroup\@sanitize@url \@url }%
\providecommand \@url [1]{\endgroup\@href {#1}{\urlprefix }}%
\providecommand \urlprefix  [0]{URL }%
\providecommand \Eprint [0]{\href }%
\providecommand \doibase [0]{https://doi.org/}%
\providecommand \selectlanguage [0]{\@gobble}%
\providecommand \bibinfo  [0]{\@secondoftwo}%
\providecommand \bibfield  [0]{\@secondoftwo}%
\providecommand \translation [1]{[#1]}%
\providecommand \BibitemOpen [0]{}%
\providecommand \bibitemStop [0]{}%
\providecommand \bibitemNoStop [0]{.\EOS\space}%
\providecommand \EOS [0]{\spacefactor3000\relax}%
\providecommand \BibitemShut  [1]{\csname bibitem#1\endcsname}%
\let\auto@bib@innerbib\@empty
\bibitem [{\citenamefont {Tarn}\ \emph {et~al.}(1980)\citenamefont {Tarn},
  \citenamefont {Huang},\ and\ \citenamefont {Clark}}]{TARN1980109}%
  \BibitemOpen
  \bibfield  {author} {\bibinfo {author} {\bibfnamefont {T.}~\bibnamefont
  {Tarn}}, \bibinfo {author} {\bibfnamefont {G.}~\bibnamefont {Huang}},\ and\
  \bibinfo {author} {\bibfnamefont {J.~W.}\ \bibnamefont {Clark}},\ }\bibfield
  {title} {\bibinfo {title} {Modelling of quantum mechanical control systems},\
  }\href {https://doi.org/https://doi.org/10.1016/0270-0255(80)90011-1}
  {\bibfield  {journal} {\bibinfo  {journal} {Mathematical Modelling}\ }\textbf
  {\bibinfo {volume} {1}},\ \bibinfo {pages} {109} (\bibinfo {year}
  {1980})}\BibitemShut {NoStop}%
\bibitem [{\citenamefont {Huang}\ \emph {et~al.}(1983)\citenamefont {Huang},
  \citenamefont {Tarn},\ and\ \citenamefont
  {Clark}}]{huang1983controllability}%
  \BibitemOpen
  \bibfield  {author} {\bibinfo {author} {\bibfnamefont {G.~M.}\ \bibnamefont
  {Huang}}, \bibinfo {author} {\bibfnamefont {T.~J.}\ \bibnamefont {Tarn}},\
  and\ \bibinfo {author} {\bibfnamefont {J.~W.}\ \bibnamefont {Clark}},\
  }\bibfield  {title} {\bibinfo {title} {On the controllability of
  quantum-mechanical systems},\ }\href {https://doi.org/10.1063/1.525634}
  {\bibfield  {journal} {\bibinfo  {journal} {Journal of Mathematical Physics}\
  }\textbf {\bibinfo {volume} {24}},\ \bibinfo {pages} {2608} (\bibinfo {year}
  {1983})}\BibitemShut {NoStop}%
\bibitem [{\citenamefont {D'Alessandro}(2007)}]{Mikobook}%
  \BibitemOpen
  \bibfield  {author} {\bibinfo {author} {\bibfnamefont {D.}~\bibnamefont
  {D'Alessandro}},\ }\href@noop {} {\emph {\bibinfo {title} {Introduction to
  quantum control and dynamics}}}\ (\bibinfo  {publisher} {CRC press},\
  \bibinfo {address} {Boca Raton, FL},\ \bibinfo {year} {2007})\BibitemShut
  {NoStop}%
\bibitem [{\citenamefont {Wiseman}\ and\ \citenamefont
  {Milburn}(2010)}]{Wiseman:book}%
  \BibitemOpen
  \bibfield  {author} {\bibinfo {author} {\bibfnamefont {H.}~\bibnamefont
  {Wiseman}}\ and\ \bibinfo {author} {\bibfnamefont {G.}~\bibnamefont
  {Milburn}},\ }\href
  {https://www.google.com/books/edition/Quantum_Measurement_and_Control/ZNjvHaH8qA4C?hl=en}
  {\emph {\bibinfo {title} {Quantum Measurement and Control}}}\ (\bibinfo
  {publisher} {Cambridge University Press},\ \bibinfo {year}
  {2010})\BibitemShut {NoStop}%
\bibitem [{\citenamefont {Jacobs}(2014)}]{jacobs2014quantum}%
  \BibitemOpen
  \bibfield  {author} {\bibinfo {author} {\bibfnamefont {K.}~\bibnamefont
  {Jacobs}},\ }\href
  {https://www.cambridge.org/core/books/quantum-measurement-theory-and-its-applications/120E32FFBEBF6EE0F6EC6F84D51DC907}
  {\emph {\bibinfo {title} {Quantum measurement theory and its applications}}}\
  (\bibinfo  {publisher} {Cambridge University Press},\ \bibinfo {year}
  {2014})\BibitemShut {NoStop}%
\bibitem [{\citenamefont {Brumer}\ and\ \citenamefont
  {Shapiro}(2003)}]{Brumer:book}%
  \BibitemOpen
  \bibfield  {author} {\bibinfo {author} {\bibfnamefont {P.}~\bibnamefont
  {Brumer}}\ and\ \bibinfo {author} {\bibfnamefont {M.}~\bibnamefont
  {Shapiro}},\ }\href@noop {} {\emph {\bibinfo {title} {{Principles of the
  Quantum Control of Molecular Processes}}}}\ (\bibinfo  {publisher}
  {{Wiley}},\ \bibinfo {year} {2003})\BibitemShut {NoStop}%
\bibitem [{\citenamefont {Brif}\ \emph {et~al.}(2010)\citenamefont {Brif},
  \citenamefont {Chakrabarti},\ and\ \citenamefont {Rabitz}}]{Brif:2010fu}%
  \BibitemOpen
  \bibfield  {author} {\bibinfo {author} {\bibfnamefont {C.}~\bibnamefont
  {Brif}}, \bibinfo {author} {\bibfnamefont {R.}~\bibnamefont {Chakrabarti}},\
  and\ \bibinfo {author} {\bibfnamefont {H.}~\bibnamefont {Rabitz}},\
  }\bibfield  {title} {\bibinfo {title} {Control of quantum phenomena: past,
  present and future},\ }\href
  {http://stacks.iop.org/1367-2630/12/i=7/a=075008} {\bibfield  {journal}
  {\bibinfo  {journal} {New Journal of Physics}\ }\textbf {\bibinfo {volume}
  {12}},\ \bibinfo {pages} {075008} (\bibinfo {year} {2010})}\BibitemShut
  {NoStop}%
\bibitem [{\citenamefont {Koch}\ \emph {et~al.}(2022)\citenamefont {Koch},
  \citenamefont {Boscain}, \citenamefont {Calarco}, \citenamefont {Dirr},
  \citenamefont {Filipp}, \citenamefont {Glaser}, \citenamefont {Kosloff},
  \citenamefont {Montangero}, \citenamefont {Schulte-Herbr{\"u}ggen},
  \citenamefont {Sugny},\ and\ \citenamefont {Wilhelm}}]{Koch:2022aa}%
  \BibitemOpen
  \bibfield  {author} {\bibinfo {author} {\bibfnamefont {C.~P.}\ \bibnamefont
  {Koch}}, \bibinfo {author} {\bibfnamefont {U.}~\bibnamefont {Boscain}},
  \bibinfo {author} {\bibfnamefont {T.}~\bibnamefont {Calarco}}, \bibinfo
  {author} {\bibfnamefont {G.}~\bibnamefont {Dirr}}, \bibinfo {author}
  {\bibfnamefont {S.}~\bibnamefont {Filipp}}, \bibinfo {author} {\bibfnamefont
  {S.~J.}\ \bibnamefont {Glaser}}, \bibinfo {author} {\bibfnamefont
  {R.}~\bibnamefont {Kosloff}}, \bibinfo {author} {\bibfnamefont
  {S.}~\bibnamefont {Montangero}}, \bibinfo {author} {\bibfnamefont
  {T.}~\bibnamefont {Schulte-Herbr{\"u}ggen}}, \bibinfo {author} {\bibfnamefont
  {D.}~\bibnamefont {Sugny}},\ and\ \bibinfo {author} {\bibfnamefont {F.~K.}\
  \bibnamefont {Wilhelm}},\ }\bibfield  {title} {\bibinfo {title} {Quantum
  optimal control in quantum technologies. strategic report on current status,
  visions and goals for research in europe},\ }\href
  {https://doi.org/10.1140/epjqt/s40507-022-00138-x} {\bibfield  {journal}
  {\bibinfo  {journal} {EPJ Quantum Technology}\ }\textbf {\bibinfo {volume}
  {9}},\ \bibinfo {pages} {19} (\bibinfo {year} {2022})}\BibitemShut {NoStop}%
\bibitem [{\citenamefont {Ramakrishna}\ and\ \citenamefont
  {Rabitz}(1996)}]{Ramakrishna:96}%
  \BibitemOpen
  \bibfield  {author} {\bibinfo {author} {\bibfnamefont {V.}~\bibnamefont
  {Ramakrishna}}\ and\ \bibinfo {author} {\bibfnamefont {H.}~\bibnamefont
  {Rabitz}},\ }\bibfield  {title} {\bibinfo {title} {Relation between quantum
  computing and quantum controllability},\ }\href
  {https://doi.org/10.1103/PhysRevA.54.1715} {\bibfield  {journal} {\bibinfo
  {journal} {Physical Review A}\ }\textbf {\bibinfo {volume} {54}},\ \bibinfo
  {pages} {1715} (\bibinfo {year} {1996})}\BibitemShut {NoStop}%
\bibitem [{\citenamefont {Palao}\ and\ \citenamefont
  {Kosloff}(2002)}]{Palao:02}%
  \BibitemOpen
  \bibfield  {author} {\bibinfo {author} {\bibfnamefont {J.~P.}\ \bibnamefont
  {Palao}}\ and\ \bibinfo {author} {\bibfnamefont {R.}~\bibnamefont
  {Kosloff}},\ }\bibfield  {title} {\bibinfo {title} {Quantum computing by an
  optimal control algorithm for unitary transformations},\ }\href
  {https://doi.org/10.1103/PhysRevLett.89.188301} {\bibfield  {journal}
  {\bibinfo  {journal} {Physical Review Letters}\ }\textbf {\bibinfo {volume}
  {89}},\ \bibinfo {pages} {188301} (\bibinfo {year} {2002})}\BibitemShut
  {NoStop}%
\bibitem [{\citenamefont {Grace}\ \emph {et~al.}(2007)\citenamefont {Grace},
  \citenamefont {Brif}, \citenamefont {Rabitz}, \citenamefont {Walmsley},
  \citenamefont {Kosut},\ and\ \citenamefont {Lidar}}]{Grace:2007aa}%
  \BibitemOpen
  \bibfield  {author} {\bibinfo {author} {\bibfnamefont {M.}~\bibnamefont
  {Grace}}, \bibinfo {author} {\bibfnamefont {C.}~\bibnamefont {Brif}},
  \bibinfo {author} {\bibfnamefont {H.}~\bibnamefont {Rabitz}}, \bibinfo
  {author} {\bibfnamefont {I.~A.}\ \bibnamefont {Walmsley}}, \bibinfo {author}
  {\bibfnamefont {R.~L.}\ \bibnamefont {Kosut}},\ and\ \bibinfo {author}
  {\bibfnamefont {D.~A.}\ \bibnamefont {Lidar}},\ }\bibfield  {title} {\bibinfo
  {title} {Optimal control of quantum gates and suppression of decoherence in a
  system of interacting two-level particles},\ }\href
  {https://iopscience.iop.org/article/10.1088/0953-4075/40/9/S06} {\bibfield
  {journal} {\bibinfo  {journal} {Journal of Physics B: Atomic, Molecular and
  Optical Physics}\ }\textbf {\bibinfo {volume} {40}},\ \bibinfo {pages} {S103}
  (\bibinfo {year} {2007})}\BibitemShut {NoStop}%
\bibitem [{\citenamefont {Motzoi}\ \emph {et~al.}(2009)\citenamefont {Motzoi},
  \citenamefont {Gambetta}, \citenamefont {Rebentrost},\ and\ \citenamefont
  {Wilhelm}}]{Motzoi:2009aa}%
  \BibitemOpen
  \bibfield  {author} {\bibinfo {author} {\bibfnamefont {F.}~\bibnamefont
  {Motzoi}}, \bibinfo {author} {\bibfnamefont {J.~M.}\ \bibnamefont
  {Gambetta}}, \bibinfo {author} {\bibfnamefont {P.}~\bibnamefont
  {Rebentrost}},\ and\ \bibinfo {author} {\bibfnamefont {F.~K.}\ \bibnamefont
  {Wilhelm}},\ }\bibfield  {title} {\bibinfo {title} {Simple pulses for
  elimination of leakage in weakly nonlinear qubits},\ }\href
  {https://doi.org/10.1103/PhysRevLett.103.110501} {\bibfield  {journal}
  {\bibinfo  {journal} {Physical Review Letters}\ }\textbf {\bibinfo {volume}
  {103}},\ \bibinfo {pages} {110501} (\bibinfo {year} {2009})}\BibitemShut
  {NoStop}%
\bibitem [{\citenamefont {Hsieh}\ \emph {et~al.}(2010)\citenamefont {Hsieh},
  \citenamefont {Wu}, \citenamefont {Rabitz},\ and\ \citenamefont
  {Lidar}}]{Hsieh:10}%
  \BibitemOpen
  \bibfield  {author} {\bibinfo {author} {\bibfnamefont {M.}~\bibnamefont
  {Hsieh}}, \bibinfo {author} {\bibfnamefont {R.}~\bibnamefont {Wu}}, \bibinfo
  {author} {\bibfnamefont {H.}~\bibnamefont {Rabitz}},\ and\ \bibinfo {author}
  {\bibfnamefont {D.}~\bibnamefont {Lidar}},\ }\bibfield  {title} {\bibinfo
  {title} {Optimal control landscape for the generation of unitary
  transformations with constrained dynamics},\ }\href
  {https://doi.org/10.1103/PhysRevA.81.062352} {\bibfield  {journal} {\bibinfo
  {journal} {Physical Review A}\ }\textbf {\bibinfo {volume} {81}},\ \bibinfo
  {pages} {062352} (\bibinfo {year} {2010})}\BibitemShut {NoStop}%
\bibitem [{\citenamefont {Doria}\ \emph {et~al.}(2011)\citenamefont {Doria},
  \citenamefont {Calarco},\ and\ \citenamefont {Montangero}}]{Doria:2011aa}%
  \BibitemOpen
  \bibfield  {author} {\bibinfo {author} {\bibfnamefont {P.}~\bibnamefont
  {Doria}}, \bibinfo {author} {\bibfnamefont {T.}~\bibnamefont {Calarco}},\
  and\ \bibinfo {author} {\bibfnamefont {S.}~\bibnamefont {Montangero}},\
  }\bibfield  {title} {\bibinfo {title} {Optimal control technique for
  many-body quantum dynamics},\ }\href
  {https://doi.org/10.1103/PhysRevLett.106.190501} {\bibfield  {journal}
  {\bibinfo  {journal} {Physical Review Letters}\ }\textbf {\bibinfo {volume}
  {106}},\ \bibinfo {pages} {190501} (\bibinfo {year} {2011})}\BibitemShut
  {NoStop}%
\bibitem [{\citenamefont {Ahn}\ \emph {et~al.}(2002)\citenamefont {Ahn},
  \citenamefont {Doherty},\ and\ \citenamefont {Landahl}}]{Ahn:01}%
  \BibitemOpen
  \bibfield  {author} {\bibinfo {author} {\bibfnamefont {C.}~\bibnamefont
  {Ahn}}, \bibinfo {author} {\bibfnamefont {A.~C.}\ \bibnamefont {Doherty}},\
  and\ \bibinfo {author} {\bibfnamefont {A.~J.}\ \bibnamefont {Landahl}},\
  }\bibfield  {title} {\bibinfo {title} {Continuous quantum error correction
  via quantum feedback control},\ }\href
  {https://doi.org/10.1103/PhysRevA.65.042301} {\bibfield  {journal} {\bibinfo
  {journal} {Physical Review A}\ }\textbf {\bibinfo {volume} {65}},\ \bibinfo
  {pages} {042301} (\bibinfo {year} {2002})}\BibitemShut {NoStop}%
\bibitem [{\citenamefont {Sarovar}\ \emph {et~al.}(2004)\citenamefont
  {Sarovar}, \citenamefont {Ahn}, \citenamefont {Jacobs},\ and\ \citenamefont
  {Milburn}}]{Sarovar:04}%
  \BibitemOpen
  \bibfield  {author} {\bibinfo {author} {\bibfnamefont {M.}~\bibnamefont
  {Sarovar}}, \bibinfo {author} {\bibfnamefont {C.}~\bibnamefont {Ahn}},
  \bibinfo {author} {\bibfnamefont {K.}~\bibnamefont {Jacobs}},\ and\ \bibinfo
  {author} {\bibfnamefont {G.~J.}\ \bibnamefont {Milburn}},\ }\bibfield
  {title} {\bibinfo {title} {Practical scheme for error control using
  feedback},\ }\href {https://doi.org/10.1103/PhysRevA.69.052324} {\bibfield
  {journal} {\bibinfo  {journal} {Physical Review A}\ }\textbf {\bibinfo
  {volume} {69}},\ \bibinfo {pages} {052324} (\bibinfo {year}
  {2004})}\BibitemShut {NoStop}%
\bibitem [{\citenamefont {Sarovar}\ and\ \citenamefont
  {Milburn}(2005)}]{Sarovar:05}%
  \BibitemOpen
  \bibfield  {author} {\bibinfo {author} {\bibfnamefont {M.}~\bibnamefont
  {Sarovar}}\ and\ \bibinfo {author} {\bibfnamefont {G.~J.}\ \bibnamefont
  {Milburn}},\ }\bibfield  {title} {\bibinfo {title} {Continuous quantum error
  correction by cooling},\ }\href {https://doi.org/10.1103/PhysRevA.72.012306}
  {\bibfield  {journal} {\bibinfo  {journal} {Physical Review A}\ }\textbf
  {\bibinfo {volume} {72}},\ \bibinfo {pages} {012306} (\bibinfo {year}
  {2005})}\BibitemShut {NoStop}%
\bibitem [{\citenamefont {Mabuchi}(2009)}]{Mabuchi:2009aa}%
  \BibitemOpen
  \bibfield  {author} {\bibinfo {author} {\bibfnamefont {H.}~\bibnamefont
  {Mabuchi}},\ }\bibfield  {title} {\bibinfo {title} {Continuous quantum error
  correction as classical hybrid control},\ }\href
  {https://doi.org/10.1088/1367-2630/11/10/105044} {\bibfield  {journal}
  {\bibinfo  {journal} {New Journal of Physics}\ }\textbf {\bibinfo {volume}
  {11}},\ \bibinfo {pages} {105044} (\bibinfo {year} {2009})}\BibitemShut
  {NoStop}%
\bibitem [{\citenamefont {Alicki}\ and\ \citenamefont
  {Lendi}(2007)}]{alicki_quantum_2007}%
  \BibitemOpen
  \bibfield  {author} {\bibinfo {author} {\bibfnamefont {R.}~\bibnamefont
  {Alicki}}\ and\ \bibinfo {author} {\bibfnamefont {K.}~\bibnamefont {Lendi}},\
  }\href {https://link.springer.com/book/10.1007/3-540-70861-8} {\emph
  {\bibinfo {title} {Quantum {Dynamical} {Semigroups} and {Applications}}}}\
  (\bibinfo  {publisher} {Springer Science \& Business Media},\ \bibinfo {year}
  {2007})\BibitemShut {NoStop}%
\bibitem [{\citenamefont {Breuer}\ and\ \citenamefont
  {Petruccione}(2002)}]{Breuer:book}%
  \BibitemOpen
  \bibfield  {author} {\bibinfo {author} {\bibfnamefont {H.-P.}\ \bibnamefont
  {Breuer}}\ and\ \bibinfo {author} {\bibfnamefont {F.}~\bibnamefont
  {Petruccione}},\ }\href
  {http://www.oxfordscholarship.com/view/10.1093/acprof:oso/9780199213900.001.0001/acprof-9780199213900}
  {\emph {\bibinfo {title} {The Theory of Open Quantum Systems}}}\ (\bibinfo
  {publisher} {Oxford University Press},\ \bibinfo {address} {Oxford},\
  \bibinfo {year} {2002})\BibitemShut {NoStop}%
\bibitem [{\citenamefont {Zanardi}\ and\ \citenamefont
  {Rasetti}(1997)}]{Zanardi:97c}%
  \BibitemOpen
  \bibfield  {author} {\bibinfo {author} {\bibfnamefont {P.}~\bibnamefont
  {Zanardi}}\ and\ \bibinfo {author} {\bibfnamefont {M.}~\bibnamefont
  {Rasetti}},\ }\bibfield  {title} {\bibinfo {title} {Noiseless quantum
  codes},\ }\href {http://link.aps.org/doi/10.1103/PhysRevLett.79.3306}
  {\bibfield  {journal} {\bibinfo  {journal} {{Phys.~Rev.~Lett.}}\ }\textbf
  {\bibinfo {volume} {79}},\ \bibinfo {pages} {3306} (\bibinfo {year}
  {1997})}\BibitemShut {NoStop}%
\bibitem [{\citenamefont {Lidar}\ \emph {et~al.}(1998)\citenamefont {Lidar},
  \citenamefont {Chuang},\ and\ \citenamefont {Whaley}}]{Lidar:1998fk}%
  \BibitemOpen
  \bibfield  {author} {\bibinfo {author} {\bibfnamefont {D.~A.}\ \bibnamefont
  {Lidar}}, \bibinfo {author} {\bibfnamefont {I.~L.}\ \bibnamefont {Chuang}},\
  and\ \bibinfo {author} {\bibfnamefont {K.~B.}\ \bibnamefont {Whaley}},\
  }\bibfield  {title} {\bibinfo {title} {Decoherence-free subspaces for quantum
  computation},\ }\href {http://link.aps.org/doi/10.1103/PhysRevLett.81.2594}
  {\bibfield  {journal} {\bibinfo  {journal} {Phys. Rev. Lett.}\ }\textbf
  {\bibinfo {volume} {81}},\ \bibinfo {pages} {2594} (\bibinfo {year}
  {1998})}\BibitemShut {NoStop}%
\bibitem [{\citenamefont {Shor}(1995)}]{shor_scheme_1995}%
  \BibitemOpen
  \bibfield  {author} {\bibinfo {author} {\bibfnamefont {P.~W.}\ \bibnamefont
  {Shor}},\ }\bibfield  {title} {\bibinfo {title} {Scheme for reducing
  decoherence in quantum computer memory},\ }\href
  {https://doi.org/10.1103/PhysRevA.52.R2493} {\bibfield  {journal} {\bibinfo
  {journal} {Phys. Rev. A}\ }\textbf {\bibinfo {volume} {52}},\ \bibinfo
  {pages} {R2493} (\bibinfo {year} {1995})}\BibitemShut {NoStop}%
\bibitem [{\citenamefont {Steane}(1996)}]{Steane:96a}%
  \BibitemOpen
  \bibfield  {author} {\bibinfo {author} {\bibfnamefont {A.~M.}\ \bibnamefont
  {Steane}},\ }\bibfield  {title} {\bibinfo {title} {Error correcting codes in
  quantum theory},\ }\href {http://link.aps.org/doi/10.1103/PhysRevLett.77.793}
  {\bibfield  {journal} {\bibinfo  {journal} {Phys. Rev. Lett.}\ }\textbf
  {\bibinfo {volume} {77}},\ \bibinfo {pages} {793} (\bibinfo {year}
  {1996})}\BibitemShut {NoStop}%
\bibitem [{\citenamefont {Gottesman}(1996)}]{Gottesman:1996fk}%
  \BibitemOpen
  \bibfield  {author} {\bibinfo {author} {\bibfnamefont {D.}~\bibnamefont
  {Gottesman}},\ }\bibfield  {title} {\bibinfo {title} {Class of quantum
  error-correcting codes saturating the quantum hamming bound},\ }\href
  {https://doi.org/10.1103/PhysRevA.54.1862} {\bibfield  {journal} {\bibinfo
  {journal} {{Phys. Rev. A}}\ }\textbf {\bibinfo {volume} {54}},\ \bibinfo
  {pages} {1862} (\bibinfo {year} {1996})}\BibitemShut {NoStop}%
\bibitem [{\citenamefont {Lloyd}(2000)}]{Lloyd:2000aa}%
  \BibitemOpen
  \bibfield  {author} {\bibinfo {author} {\bibfnamefont {S.}~\bibnamefont
  {Lloyd}},\ }\bibfield  {title} {\bibinfo {title} {Coherent quantum
  feedback},\ }\href {https://doi.org/10.1103/PhysRevA.62.022108} {\bibfield
  {journal} {\bibinfo  {journal} {Physical Review A}\ }\textbf {\bibinfo
  {volume} {62}},\ \bibinfo {pages} {022108} (\bibinfo {year}
  {2000})}\BibitemShut {NoStop}%
\bibitem [{\citenamefont {Viola}\ and\ \citenamefont {Lloyd}(1998)}]{Viola:98}%
  \BibitemOpen
  \bibfield  {author} {\bibinfo {author} {\bibfnamefont {L.}~\bibnamefont
  {Viola}}\ and\ \bibinfo {author} {\bibfnamefont {S.}~\bibnamefont {Lloyd}},\
  }\bibfield  {title} {\bibinfo {title} {Dynamical suppression of decoherence
  in two-state quantum systems},\ }\href
  {https://link.aps.org/doi/10.1103/PhysRevA.58.2733} {\bibfield  {journal}
  {\bibinfo  {journal} {Phys. Rev. A}\ }\textbf {\bibinfo {volume} {58}},\
  \bibinfo {pages} {2733} (\bibinfo {year} {1998})}\BibitemShut {NoStop}%
\bibitem [{\citenamefont {Viola}\ \emph {et~al.}(1999)\citenamefont {Viola},
  \citenamefont {Lloyd},\ and\ \citenamefont {Knill}}]{Viola:99a}%
  \BibitemOpen
  \bibfield  {author} {\bibinfo {author} {\bibfnamefont {L.}~\bibnamefont
  {Viola}}, \bibinfo {author} {\bibfnamefont {S.}~\bibnamefont {Lloyd}},\ and\
  \bibinfo {author} {\bibfnamefont {E.}~\bibnamefont {Knill}},\ }\bibfield
  {title} {\bibinfo {title} {Universal control of decoupled quantum systems},\
  }\href {http://link.aps.org/doi/10.1103/PhysRevLett.83.4888} {\bibfield
  {journal} {\bibinfo  {journal} {Phys. Rev. Lett.}\ }\textbf {\bibinfo
  {volume} {83}},\ \bibinfo {pages} {4888} (\bibinfo {year}
  {1999})}\BibitemShut {NoStop}%
\bibitem [{\citenamefont {Khodjasteh}\ \emph {et~al.}(2011)\citenamefont
  {Khodjasteh}, \citenamefont {Dobrovitski},\ and\ \citenamefont
  {Viola}}]{Khodjasteh:2011aa}%
  \BibitemOpen
  \bibfield  {author} {\bibinfo {author} {\bibfnamefont {K.}~\bibnamefont
  {Khodjasteh}}, \bibinfo {author} {\bibfnamefont {V.~V.}\ \bibnamefont
  {Dobrovitski}},\ and\ \bibinfo {author} {\bibfnamefont {L.}~\bibnamefont
  {Viola}},\ }\bibfield  {title} {\bibinfo {title} {Pointer states via
  engineered dissipation},\ }\href {https://doi.org/10.1103/PhysRevA.84.022336}
  {\bibfield  {journal} {\bibinfo  {journal} {Physical Review A}\ }\textbf
  {\bibinfo {volume} {84}},\ \bibinfo {pages} {022336} (\bibinfo {year}
  {2011})}\BibitemShut {NoStop}%
\bibitem [{\citenamefont {Knill}\ \emph {et~al.}(2000)\citenamefont {Knill},
  \citenamefont {Laflamme},\ and\ \citenamefont {Viola}}]{Knill:2000dq}%
  \BibitemOpen
  \bibfield  {author} {\bibinfo {author} {\bibfnamefont {E.}~\bibnamefont
  {Knill}}, \bibinfo {author} {\bibfnamefont {R.}~\bibnamefont {Laflamme}},\
  and\ \bibinfo {author} {\bibfnamefont {L.}~\bibnamefont {Viola}},\ }\bibfield
   {title} {\bibinfo {title} {Theory of quantum error correction for general
  noise},\ }\href {http://link.aps.org/doi/10.1103/PhysRevLett.84.2525}
  {\bibfield  {journal} {\bibinfo  {journal} {{Phys.~Rev.~Lett.}}\ }\textbf
  {\bibinfo {volume} {84}},\ \bibinfo {pages} {2525} (\bibinfo {year}
  {2000})}\BibitemShut {NoStop}%
\bibitem [{\citenamefont {Blume-Kohout}\ \emph {et~al.}(2008)\citenamefont
  {Blume-Kohout}, \citenamefont {Ng}, \citenamefont {Poulin},\ and\
  \citenamefont {Viola}}]{Blume-Kohout:2008ec}%
  \BibitemOpen
  \bibfield  {author} {\bibinfo {author} {\bibfnamefont {R.}~\bibnamefont
  {Blume-Kohout}}, \bibinfo {author} {\bibfnamefont {H.~K.}\ \bibnamefont
  {Ng}}, \bibinfo {author} {\bibfnamefont {D.}~\bibnamefont {Poulin}},\ and\
  \bibinfo {author} {\bibfnamefont {L.}~\bibnamefont {Viola}},\ }\bibfield
  {title} {\bibinfo {title} {Characterizing the structure of preserved
  information in quantum processes},\ }\href
  {http://link.aps.org/doi/10.1103/PhysRevLett.100.030501} {\bibfield
  {journal} {\bibinfo  {journal} {Physical Review Letters}\ }\textbf {\bibinfo
  {volume} {100}},\ \bibinfo {pages} {030501} (\bibinfo {year}
  {2008})}\BibitemShut {NoStop}%
\bibitem [{\citenamefont {Lidar}\ and\ \citenamefont
  {Brun}(2013)}]{Lidar-Brun:book}%
  \BibitemOpen
  \bibinfo {editor} {\bibfnamefont {D.}~\bibnamefont {Lidar}}\ and\ \bibinfo
  {editor} {\bibfnamefont {T.}~\bibnamefont {Brun}},\ eds.,\ \href
  {http://www.cambridge.org/9780521897877} {\emph {\bibinfo {title} {Quantum
  Error Correction}}}\ (\bibinfo  {publisher} {Cambridge University Press},\
  \bibinfo {address} {{Cambridge, UK}},\ \bibinfo {year} {2013})\BibitemShut
  {NoStop}%
\bibitem [{\citenamefont {Hecht}\ \emph {et~al.}(2024)\citenamefont {Hecht},
  \citenamefont {Saurav}, \citenamefont {Vlachos}, \citenamefont {Lidar},\ and\
  \citenamefont {Levenson-Falk}}]{tbp-tracking}%
  \BibitemOpen
  \bibfield  {author} {\bibinfo {author} {\bibfnamefont {M.}~\bibnamefont
  {Hecht}}, \bibinfo {author} {\bibfnamefont {K.}~\bibnamefont {Saurav}},
  \bibinfo {author} {\bibfnamefont {E.}~\bibnamefont {Vlachos}}, \bibinfo
  {author} {\bibfnamefont {D.~A.}\ \bibnamefont {Lidar}},\ and\ \bibinfo
  {author} {\bibfnamefont {E.~M.}\ \bibnamefont {Levenson-Falk}},\ }\href@noop
  {} {\bibinfo {title} {{in preparation}}} (\bibinfo {year} {2024})\BibitemShut
  {NoStop}%
\bibitem [{\citenamefont {Hirschorn}(1979)}]{Hirschorn:79}%
  \BibitemOpen
  \bibfield  {author} {\bibinfo {author} {\bibfnamefont {R.~M.}\ \bibnamefont
  {Hirschorn}},\ }\bibfield  {title} {\bibinfo {title} {Invertibility of
  nonlinear control systems},\ }\href {https://doi.org/10.1137/0317022}
  {\bibfield  {journal} {\bibinfo  {journal} {SIAM Journal on Control and
  Optimization}\ }\textbf {\bibinfo {volume} {17}},\ \bibinfo {pages} {289}
  (\bibinfo {year} {1979})}\BibitemShut {NoStop}%
\bibitem [{\citenamefont {Hirschorn}\ and\ \citenamefont
  {Davis}(1988)}]{Hirschorn:88}%
  \BibitemOpen
  \bibfield  {author} {\bibinfo {author} {\bibfnamefont {R.~M.}\ \bibnamefont
  {Hirschorn}}\ and\ \bibinfo {author} {\bibfnamefont {J.~H.}\ \bibnamefont
  {Davis}},\ }\bibfield  {title} {\bibinfo {title} {{Global Output Tracking for
  Nonlinear Systems}},\ }\href {https://doi.org/10.1137/0326074} {\bibfield
  {journal} {\bibinfo  {journal} {SIAM Journal on Control and Optimization}\
  }\textbf {\bibinfo {volume} {26}},\ \bibinfo {pages} {1321} (\bibinfo {year}
  {1988})}\BibitemShut {NoStop}%
\bibitem [{\citenamefont {Jakubczyk}\ and\ \citenamefont
  {Lamnabhi-Lagarrigue}(1993)}]{Jakubczyk:93}%
  \BibitemOpen
  \bibfield  {author} {\bibinfo {author} {\bibfnamefont {B.}~\bibnamefont
  {Jakubczyk}}\ and\ \bibinfo {author} {\bibfnamefont {F.}~\bibnamefont
  {Lamnabhi-Lagarrigue}},\ }\bibfield  {title} {\bibinfo {title} {{Tracking
  through singularities: regularity of the control}},\ }\href
  {https://www.sciencedirect.com/science/article/pii/016769119390068H}
  {\bibfield  {journal} {\bibinfo  {journal} {Systems \& Control Letters}\
  }\textbf {\bibinfo {volume} {21}},\ \bibinfo {pages} {271} (\bibinfo {year}
  {1993})}\BibitemShut {NoStop}%
\bibitem [{\citenamefont {{Y. Chen, P. Gross, V. Ramakrishna, H. Rabitz, and K.
  Mease}}(1995)}]{Chen:95}%
  \BibitemOpen
  \bibfield  {author} {\bibinfo {author} {\bibnamefont {{Y. Chen, P. Gross, V.
  Ramakrishna, H. Rabitz, and K. Mease}}},\ }\bibfield  {title} {\bibinfo
  {title} {{Competitive tracking of molecular objectives described by quantum
  mechanics}},\ }\href {https://doi.org/10.1063/1.468998} {\bibfield  {journal}
  {\bibinfo  {journal} {J. Chem. Phys.}\ }\textbf {\bibinfo {volume} {102}},\
  \bibinfo {pages} {8001} (\bibinfo {year} {1995})}\BibitemShut {NoStop}%
\bibitem [{\citenamefont {Lu}\ and\ \citenamefont {Rabitz}(1995)}]{Lu:95}%
  \BibitemOpen
  \bibfield  {author} {\bibinfo {author} {\bibfnamefont {Z.-M.}\ \bibnamefont
  {Lu}}\ and\ \bibinfo {author} {\bibfnamefont {H.}~\bibnamefont {Rabitz}},\
  }\bibfield  {title} {\bibinfo {title} {Unified formulation for control and
  inversion of molecular dynamics},\ }\href
  {https://doi.org/10.1021/j100037a021} {\bibfield  {journal} {\bibinfo
  {journal} {The Journal of Physical Chemistry}\ }\textbf {\bibinfo {volume}
  {99}},\ \bibinfo {pages} {13731} (\bibinfo {year} {1995})}\BibitemShut
  {NoStop}%
\bibitem [{\citenamefont {{P. Gross, H. Singh, H. Rabitz, K. Mease, and G.
  Huang}}(1993)}]{Gross:93}%
  \BibitemOpen
  \bibfield  {author} {\bibinfo {author} {\bibnamefont {{P. Gross, H. Singh, H.
  Rabitz, K. Mease, and G. Huang}}},\ }\bibfield  {title} {\bibinfo {title}
  {{Inverse quantum-mechanical control: A means for design and a test of
  intuition}},\ }\href
  {https://journals.aps.org/pra/abstract/10.1103/PhysRevA.47.4593} {\bibfield
  {journal} {\bibinfo  {journal} {Phys. Rev. A}\ }\textbf {\bibinfo {volume}
  {47}},\ \bibinfo {pages} {4593} (\bibinfo {year} {1993})}\BibitemShut
  {NoStop}%
\bibitem [{\citenamefont {{W. Zhu, M. Smit, and H. Rabitz}}(1999)}]{Zhu:99}%
  \BibitemOpen
  \bibfield  {author} {\bibinfo {author} {\bibnamefont {{W. Zhu, M. Smit, and
  H. Rabitz}}},\ }\bibfield  {title} {\bibinfo {title} {{Managing singular
  behavior in the tracking control of quantum dynamical observables}},\ }\href
  {https://doi.org/10.1063/1.477857} {\bibfield  {journal} {\bibinfo  {journal}
  {J. Chem. Phys.}\ }\textbf {\bibinfo {volume} {110}},\ \bibinfo {pages}
  {1905} (\bibinfo {year} {1999})}\BibitemShut {NoStop}%
\bibitem [{\citenamefont {{W. Zhu, J. Botina, and H. Rabitz}}(1998)}]{Zhu:98}%
  \BibitemOpen
  \bibfield  {author} {\bibinfo {author} {\bibnamefont {{W. Zhu, J. Botina, and
  H. Rabitz}}},\ }\bibfield  {title} {\bibinfo {title} {{Rapidly convergent
  iteration methods for quantum optimal control of population}},\ }\href
  {https://doi.org/10.1063/1.475576} {\bibfield  {journal} {\bibinfo  {journal}
  {J. Chem. Phys.}\ }\textbf {\bibinfo {volume} {108}},\ \bibinfo {pages}
  {1953} (\bibinfo {year} {1998})}\BibitemShut {NoStop}%
\bibitem [{\citenamefont {{W. Zhu and H. Rabitz}}(2003)}]{Zhu:03}%
  \BibitemOpen
  \bibfield  {author} {\bibinfo {author} {\bibnamefont {{W. Zhu and H.
  Rabitz}}},\ }\bibfield  {title} {\bibinfo {title} {{Quantum control design
  via adaptive tracking}},\ }\href {https://doi.org/10.1063/1.1582847}
  {\bibfield  {journal} {\bibinfo  {journal} {J. Chem. Phys.}\ }\textbf
  {\bibinfo {volume} {119}},\ \bibinfo {pages} {3619} (\bibinfo {year}
  {2003})}\BibitemShut {NoStop}%
\bibitem [{\citenamefont {Rivas}\ \emph {et~al.}(2010)\citenamefont {Rivas},
  \citenamefont {Huelga},\ and\ \citenamefont
  {Plenio}}]{rivas2010entanglement}%
  \BibitemOpen
  \bibfield  {author} {\bibinfo {author} {\bibfnamefont {A.}~\bibnamefont
  {Rivas}}, \bibinfo {author} {\bibfnamefont {S.~F.}\ \bibnamefont {Huelga}},\
  and\ \bibinfo {author} {\bibfnamefont {M.~B.}\ \bibnamefont {Plenio}},\
  }\bibfield  {title} {\bibinfo {title} {Entanglement and non-markovianity of
  quantum evolutions},\ }\href {https://doi.org/10.1103/PhysRevLett.105.050403}
  {\bibfield  {journal} {\bibinfo  {journal} {Physical Review Letters}\
  }\textbf {\bibinfo {volume} {105}},\ \bibinfo {pages} {050403} (\bibinfo
  {year} {2010})}\BibitemShut {NoStop}%
\bibitem [{\citenamefont {Breuer}\ \emph {et~al.}(2016)\citenamefont {Breuer},
  \citenamefont {Laine}, \citenamefont {Piilo},\ and\ \citenamefont
  {Vacchini}}]{breuer2016colloquium}%
  \BibitemOpen
  \bibfield  {author} {\bibinfo {author} {\bibfnamefont {H.-P.}\ \bibnamefont
  {Breuer}}, \bibinfo {author} {\bibfnamefont {E.-M.}\ \bibnamefont {Laine}},
  \bibinfo {author} {\bibfnamefont {J.}~\bibnamefont {Piilo}},\ and\ \bibinfo
  {author} {\bibfnamefont {B.}~\bibnamefont {Vacchini}},\ }\bibfield  {title}
  {\bibinfo {title} {Colloquium: Non-markovian dynamics in open quantum
  systems},\ }\href {https://doi.org/10.1103/RevModPhys.88.021002} {\bibfield
  {journal} {\bibinfo  {journal} {Reviews of Modern Physics}\ }\textbf
  {\bibinfo {volume} {88}},\ \bibinfo {pages} {021002} (\bibinfo {year}
  {2016})}\BibitemShut {NoStop}%
\bibitem [{\citenamefont {Chru{\'s}ci{\'n}ski}(2022)}]{Chruscinski:2022aa}%
  \BibitemOpen
  \bibfield  {author} {\bibinfo {author} {\bibfnamefont {D.}~\bibnamefont
  {Chru{\'s}ci{\'n}ski}},\ }\bibfield  {title} {\bibinfo {title} {Dynamical
  maps beyond markovian regime},\ }\href
  {https://doi.org/https://doi.org/10.1016/j.physrep.2022.09.003} {\bibfield
  {journal} {\bibinfo  {journal} {Physics Reports}\ }\textbf {\bibinfo {volume}
  {992}},\ \bibinfo {pages} {1} (\bibinfo {year} {2022})}\BibitemShut {NoStop}%
\bibitem [{\citenamefont {Bengtsson}\ and\ \citenamefont
  {Zyczkowski}(2006)}]{bengtsson2006geometry}%
  \BibitemOpen
  \bibfield  {author} {\bibinfo {author} {\bibfnamefont {I.}~\bibnamefont
  {Bengtsson}}\ and\ \bibinfo {author} {\bibfnamefont {K.}~\bibnamefont
  {Zyczkowski}},\ }\href
  {https://www.cambridge.org/core/books/geometry-of-quantum-states/4BA9DCEED5BB16B222A917EAAAD17028}
  {\emph {\bibinfo {title} {Geometry of quantum states: an introduction to
  quantum entanglement}}}\ (\bibinfo  {publisher} {Cambridge University
  Press},\ \bibinfo {year} {2006})\BibitemShut {NoStop}%
\bibitem [{\citenamefont {Kasatkin}\ \emph {et~al.}(2023)\citenamefont
  {Kasatkin}, \citenamefont {Gu},\ and\ \citenamefont {Lidar}}]{ODE2QME}%
  \BibitemOpen
  \bibfield  {author} {\bibinfo {author} {\bibfnamefont {V.}~\bibnamefont
  {Kasatkin}}, \bibinfo {author} {\bibfnamefont {L.}~\bibnamefont {Gu}},\ and\
  \bibinfo {author} {\bibfnamefont {D.~A.}\ \bibnamefont {Lidar}},\ }\bibfield
  {title} {\bibinfo {title} {Which differential equations correspond to the
  lindblad equation?},\ }\href
  {https://doi.org/10.1103/PhysRevResearch.5.043163} {\bibfield  {journal}
  {\bibinfo  {journal} {Physical Review Research}\ }\textbf {\bibinfo {volume}
  {5}},\ \bibinfo {pages} {043163} (\bibinfo {year} {2023})}\BibitemShut
  {NoStop}%
\bibitem [{\citenamefont {Zanardi}(1998)}]{Zanardi:98a}%
  \BibitemOpen
  \bibfield  {author} {\bibinfo {author} {\bibfnamefont {P.}~\bibnamefont
  {Zanardi}},\ }\bibfield  {title} {\bibinfo {title} {Dissipation and
  decoherence in a quantum register},\ }\href
  {https://doi.org/10.1103/PhysRevA.57.3276} {\bibfield  {journal} {\bibinfo
  {journal} {Physical Review A}\ }\textbf {\bibinfo {volume} {57}},\ \bibinfo
  {pages} {3276} (\bibinfo {year} {1998})}\BibitemShut {NoStop}%
\bibitem [{\citenamefont {Shabani}\ and\ \citenamefont
  {Lidar}(2005)}]{ShabaniLidar:05a}%
  \BibitemOpen
  \bibfield  {author} {\bibinfo {author} {\bibfnamefont {A.}~\bibnamefont
  {Shabani}}\ and\ \bibinfo {author} {\bibfnamefont {D.~A.}\ \bibnamefont
  {Lidar}},\ }\bibfield  {title} {\bibinfo {title} {Theory of
  initialization-free decoherence-free subspaces and subsystems},\ }\href
  {https://doi.org/10.1103/PhysRevA.72.042303} {\bibfield  {journal} {\bibinfo
  {journal} {Physical Review A}\ }\textbf {\bibinfo {volume} {72}},\ \bibinfo
  {pages} {042303} (\bibinfo {year} {2005})}\BibitemShut {NoStop}%
\bibitem [{\citenamefont {Lidar}\ and\ \citenamefont
  {Schneider}(2005)}]{LidarSchneider:04}%
  \BibitemOpen
  \bibfield  {author} {\bibinfo {author} {\bibfnamefont {D.~A.}\ \bibnamefont
  {Lidar}}\ and\ \bibinfo {author} {\bibfnamefont {S.}~\bibnamefont
  {Schneider}},\ }\bibfield  {title} {\bibinfo {title} {{Stabilizing qubit
  coherence via tracking-control}},\ }\href
  {https://arxiv.org/abs/quant-ph/0410048} {\bibfield  {journal} {\bibinfo
  {journal} {Quantum Inf. Comput.}\ }\textbf {\bibinfo {volume} {5}},\ \bibinfo
  {pages} {350} (\bibinfo {year} {2005})}\BibitemShut {NoStop}%
\bibitem [{\citenamefont {Glauber}(1963)}]{Glauber:1963aa}%
  \BibitemOpen
  \bibfield  {author} {\bibinfo {author} {\bibfnamefont {R.~J.}\ \bibnamefont
  {Glauber}},\ }\bibfield  {title} {\bibinfo {title} {Coherent and incoherent
  states of the radiation field},\ }\href
  {https://doi.org/10.1103/PhysRev.131.2766} {\bibfield  {journal} {\bibinfo
  {journal} {Physical Review}\ }\textbf {\bibinfo {volume} {131}},\ \bibinfo
  {pages} {2766} (\bibinfo {year} {1963})}\BibitemShut {NoStop}%
\bibitem [{\citenamefont {Styliaris}\ \emph {et~al.}(2018)\citenamefont
  {Styliaris}, \citenamefont {Campos~Venuti},\ and\ \citenamefont
  {Zanardi}}]{Styliaris:2018aa}%
  \BibitemOpen
  \bibfield  {author} {\bibinfo {author} {\bibfnamefont {G.}~\bibnamefont
  {Styliaris}}, \bibinfo {author} {\bibfnamefont {L.}~\bibnamefont
  {Campos~Venuti}},\ and\ \bibinfo {author} {\bibfnamefont {P.}~\bibnamefont
  {Zanardi}},\ }\bibfield  {title} {\bibinfo {title} {Coherence-generating
  power of quantum dephasing processes},\ }\href
  {https://doi.org/10.1103/PhysRevA.97.032304} {\bibfield  {journal} {\bibinfo
  {journal} {Physical Review A}\ }\textbf {\bibinfo {volume} {97}},\ \bibinfo
  {pages} {032304} (\bibinfo {year} {2018})}\BibitemShut {NoStop}%
\bibitem [{\citenamefont {Gross}\ \emph {et~al.}(2010)\citenamefont {Gross},
  \citenamefont {Liu}, \citenamefont {Flammia}, \citenamefont {Becker},\ and\
  \citenamefont {Eisert}}]{Gross:2010aa}%
  \BibitemOpen
  \bibfield  {author} {\bibinfo {author} {\bibfnamefont {D.}~\bibnamefont
  {Gross}}, \bibinfo {author} {\bibfnamefont {Y.-K.}\ \bibnamefont {Liu}},
  \bibinfo {author} {\bibfnamefont {S.~T.}\ \bibnamefont {Flammia}}, \bibinfo
  {author} {\bibfnamefont {S.}~\bibnamefont {Becker}},\ and\ \bibinfo {author}
  {\bibfnamefont {J.}~\bibnamefont {Eisert}},\ }\bibfield  {title} {\bibinfo
  {title} {Quantum state tomography via compressed sensing},\ }\href
  {https://doi.org/10.1103/PhysRevLett.105.150401} {\bibfield  {journal}
  {\bibinfo  {journal} {Physical Review Letters}\ }\textbf {\bibinfo {volume}
  {105}},\ \bibinfo {pages} {150401} (\bibinfo {year} {2010})}\BibitemShut
  {NoStop}%
\bibitem [{\citenamefont {Aaronson}(2018)}]{Aaronson-shadow-tomography}%
  \BibitemOpen
  \bibfield  {author} {\bibinfo {author} {\bibfnamefont {S.}~\bibnamefont
  {Aaronson}},\ }\bibfield  {title} {\bibinfo {title} {Shadow tomography of
  quantum states},\ }in\ \href {https://doi.org/10.1145/3188745.3188802} {\emph
  {\bibinfo {booktitle} {Proceedings of the 50th Annual ACM SIGACT Symposium on
  Theory of Computing}}},\ \bibinfo {series and number} {STOC 2018}\ (\bibinfo
  {publisher} {Association for Computing Machinery},\ \bibinfo {address} {New
  York, NY, USA},\ \bibinfo {year} {2018})\ pp.\ \bibinfo {pages}
  {325--338}\BibitemShut {NoStop}%
\bibitem [{\citenamefont {Huang}\ \emph {et~al.}(2020)\citenamefont {Huang},
  \citenamefont {Kueng},\ and\ \citenamefont {Preskill}}]{Huang:2020wo}%
  \BibitemOpen
  \bibfield  {author} {\bibinfo {author} {\bibfnamefont {H.-Y.}\ \bibnamefont
  {Huang}}, \bibinfo {author} {\bibfnamefont {R.}~\bibnamefont {Kueng}},\ and\
  \bibinfo {author} {\bibfnamefont {J.}~\bibnamefont {Preskill}},\ }\bibfield
  {title} {\bibinfo {title} {Predicting many properties of a quantum system
  from very few measurements},\ }\href
  {https://doi.org/10.1038/s41567-020-0932-7} {\bibfield  {journal} {\bibinfo
  {journal} {Nature Physics}\ }\textbf {\bibinfo {volume} {16}},\ \bibinfo
  {pages} {1050} (\bibinfo {year} {2020})}\BibitemShut {NoStop}%
\bibitem [{\citenamefont {H{\"u}bner}(1992)}]{Hubner:1992aa}%
  \BibitemOpen
  \bibfield  {author} {\bibinfo {author} {\bibfnamefont {M.}~\bibnamefont
  {H{\"u}bner}},\ }\bibfield  {title} {\bibinfo {title} {Explicit computation
  of the bures distance for density matrices},\ }\href
  {https://doi.org/https://doi.org/10.1016/0375-9601(92)91004-B} {\bibfield
  {journal} {\bibinfo  {journal} {Physics Letters A}\ }\textbf {\bibinfo
  {volume} {163}},\ \bibinfo {pages} {239} (\bibinfo {year}
  {1992})}\BibitemShut {NoStop}%
\end{thebibliography}%

\end{document}